\documentclass[final]{IEEEtran}

\usepackage{amsmath, amssymb, amsbsy, bbm, array, enumerate, calc, scalerel}

\newcommand{\mi}{\mathit} 
\DeclareMathAlphabet{\mib}{T1}{cmr}{bx}{it}
\newcommand{\mc}{\mathcal}

\newcommand{\mf}{\mathfrak}

\newcommand{\mbb}{\mathbb} 
\newcommand{\mbbm}{\mathbbm}
\newcommand{\mbi}{\boldsymbol}

\newcommand{\cref}[1]{\arabic{#1}}
\newcommand{\ctag}[2]{\addtocounter{equation}{1}\setcounter{#1}{\value{equation}}
                      \tag*{#2 (\cref{#1})}}

\newcommand{\geql}{\eqcirc}
\newcommand{\gle}{\prec}   

\newcommand{\feql}{\operatorname{\pmb{\eqcirc}}}
\newcommand{\fle}{\operatorname{\pmb{\prec}}}

\newcommand{\ceql}{\operatorname{\triangleq}}
\newcommand{\cle}{\operatorname{\vartriangleleft}}

\newcommand{\nceql}{\operatorname{\not\triangleq}}  
\newcommand{\ncle}{\operatorname{\ntriangleleft}}

\newcommand{\gz}{\bar{0}}
\newcommand{\gu}{\bar{1}}

\newcommand{\fcom}[1]{\pmb{\overline{\phantom{#1}}}\mspace{-9.75mu}#1}
\newcommand{\fcomd}[1]{\pmb{\overline{\phantom{#1}}}\mspace{-47.50mu}#1}
\newcommand{\fcoml}{\pmb{\overline{\parbox{7pt}{\mbox{}\vspace{6.6pt}\mbox{}}}}}

\newcommand{\fvee}{\operatorname{\pmb{\vee}}}
\newcommand{\bigfvee}{\operatornamewithlimits{\pmb{\bigvee}}}

\newcommand{\fwedge}{\operatorname{\pmb{\wedge}}}
\newcommand{\bigfwedge}{\operatornamewithlimits{\pmb{\bigwedge}}}

\newcommand{\frightarrow}{\operatorname{\pmb{\Rightarrow}}}

\newcommand{\fdiamond}{\operatorname{\pmb{\diamond}}}

\newcommand{\del}{\operatorname{\Delta}}
\newcommand{\fdel}{\operatorname{\pmb{\Delta}}}

\newcommand{\eqvl}{\equiv}

\newcommand{\emptyseq}{\ell}

\newtheorem{theorem}{Theorem}
\newtheorem{lemma}[theorem]{Lemma}
\newtheorem{corollary}[theorem]{Corollary}

\newenvironment{proof}{\begin{IEEEproof}}
                      {\end{IEEEproof}}  

\newcommand{\qed}{\mbox{}}

\begin{document}

\title{Hyperresolution for Multi-step Fuzzy Inference in G\"{o}del Logic}

\author{Du\v{s}an~Guller%
\thanks{D. Guller is with the Department of Applied Informatics, Comenius University, Mlynsk\'a dolina, 842 48 Bratislava, Slovakia
        (e-mail: guller@fmph.uniba.sk).}%
}

\markboth{Du\v{s}an Guller \hfill Hyperresolution for Multi-step Fuzzy Inference in G\"{o}del Logic \hspace{5mm}}{Hyperresolution for Multi-step Fuzzy Inference in G\"{o}del Logic}

\maketitle

\begin{abstract}
This paper is a continuation of our work concerning the logical and computational foundations of multi-step fuzzy inference. 
We bring further results on the implementation of the Mamdani-Assilian type of fuzzy rules and inference in G\"{o}del logic with truth constants.
In our previous work, we have provided translation of Mamdani-Assilian fuzzy rules to formulae of G\"{o}del logic, and subsequently, 
to suitable clausal form.
Moreover, we have outlined a class of problems regarding general properties of fuzzy inference and 
shown its reduction to a class of deduction/unsatisfiability problems. 
We now focus on solving such problems using an adapted hyperresolution calculus.
\end{abstract}


\section{Introduction}
\label{S1}

Steady progress in artificial, computational intelligence, soft computing, data science, and knowledge representation brings new challenges
of processing complex, large amount of data, with deep and non-trivial semantical dependencies.
Among the standard methods (including machine learning, evolutionary computation, Bayesian, Markov networks, etc.), 
fuzzy logic, based on the many-valued logic approach, is well suited to cope with these challenges.
The reason is that it can provide mathematically correct and advanced semantics of uncertainty in general 
for capturing the semantics of knowledge in a specific area.
Besides a strong declarative power, the fuzzy logic approach can offer computationally efficient solutions, 
e.g. exploiting methods of automated reasoning.
The conjunction of both the methods can constitute so-called fuzzy reasoning --
a kind of abstract inference, generally with multiple inference steps.

Based on our preliminary results in \cite{Guller2012b,Guller2015a,Guller2014,Guller2015c,Guller2016c,Guller2018a,Guller2018b,Guller2019b},
we have decided to provide the logical and computational foundations of multi-step fuzzy inference.
Particularly, a good choice has been the implementation of the Mamdani-Assilian type of fuzzy rules and inference 
in G\"{o}del logic with truth constants.
In \cite{Guller2023c}, we have proposed translation of Mamdani-Assilian fuzzy rules to formulae of G\"{o}del logic, and subsequently, 
to clausal form.
A fuzzy inference system is viewed as a logical theory, which allow us to simulate fuzzy inference itself, and moreover, 
to formulate general properties of fuzzy inference as deduction problems.
The translation to clausal form preserves satisfiability, so deduction problems can be reduced to unsatisfiability ones.
Notice that multi-step fuzzy inference is a kind of non-monotonic reasoning.
To deal with the non-monotonicity, we had to devise an adequate axiomatisation reflecting a kind of  
linear discrete time with the starting point $0$ and without an endpoint.
In addition, we had to incorporate the universum (often a closed interval of real numbers) of fuzzy sets (appearing in the fuzzy inference system) into our axiomatisation.
We have introduced a notion of fuzzy variable assignment and exploited it in the definition of multi-step fuzzy inference, 
viewed as a sequence of "linked" fuzzy variable assignments, where the values of fuzzy variables at the next step are computed 
from the values valid at the preceding step using the Mamdani-Assilian method.   
If the fuzzy sets in question are continuous functions on the universum, it suffices consider only the countable universum 
(e.g. a closed interval of rational numbers).  
Moreover, in certain circumstances, we can obtain a finite approximation of the universum.
We may assume that the universum is a closed interval $[u_b,u_e]$ of real numbers, and 
the fuzzy sets in question (a finite number) are 
semi-differentiable at every inner point of the interval, right-differentiable at $u_b$, left-differentiable at $u_e$. 
There exist lower and upper bounds on all left derivatives, and analogously, lower and upper bounds on all right derivatives.
If two fuzzy sets $A$ and $B$ are different at some inner point $u$, then there exists an open subinterval $(u-\delta,u+\delta)$
where $A$ and $B$ have different values at every point of $(u-\delta,u+\delta)$.
This yields that the universum $[u_b,u_e]$ can be replaced with a finite set of witness points $\{u_b<u_1<\cdots<u_\lambda<u_e\}$
which splits the closed interval $[u_b,u_e]$ into subsequent closed subintervals of the same length so that
if two fuzzy sets $A$ and $B$ are different, then there exists a witness point $u^*$ witnessing $A(u^*)\neq B(u^*)$. 
Such a finite approximation is sufficient for many problems, covering those ones posed in \cite{Guller2023c}, and
allows us apply a modified hyperresolution calculus from \cite{Guller2019b}, inferring over clausal translations.
The modified hyperresolution calculus will be proved to be refutation sound and complete.
We extend it with some admissible rules, which help shorten hyperresolution derivations.
At the end, we shall continue with the example on a fuzzy inference system modelling an engine with inner combustion and cooling medium 
from \cite{Guller2023c}.
We formulate an instance of the reachability problem, translate it to a clausal theory, and 
construct a refutation of the clausal theory using hyperresolution, which solves the given problem.

The paper is organised as follows.
Section \ref{S2} gives the basic notions and notation concerning the first-order G\"{o}del logic.
Section \ref{S3} deals with translation to clausal form.
Section \ref{S4} introduces a modified hyperresolution calculus.
Section \ref{S5} is devoted to multi-step fuzzy inference.
Section~\ref{S6} contains an illustration example on hyperresolution, a continuation of that in \cite{Guller2023c}.
Section~\ref{S7} brings conclusions.

\subsection{Preliminaries}
\label{S1.1}

$\mbb{N}$, $\mbb{Q}$, $\mbb{R}$ designates the set of natural, rational, real numbers, and
$=$, $\leq$, $<$ denotes the standard equality, order, strict order on $\mbb{N}$, $\mbb{Q}$, $\mbb{R}$.
We denote $\mbb{R}_0^+=\{c \,|\, 0\leq c\in \mbb{R}\}$, $\mbb{R}^+=\{c \,|\, 0<c\in \mbb{R}\}$,
$[0,1]=\{c \,|\, c\in \mbb{R}, 0\leq c\leq 1\}$; $[0,1]$ is called the unit interval.
Let $X$, $Y$, $Z$ be sets and $f : X\longrightarrow Y$ a mapping.
By $\|X\|$ we denote the set-theoretic cardinality of $X$.
The relationship of $X$ being a finite subset of $Y$ is denoted as $X\subseteq_{\mc F} Y$.
Let $Z\subseteq X$.
We designate 
$f[Z]=\{f(z) \,|\, z\in Z\}$; $f[Z]$ is called the image of $Z$ under $f$; 
$f|_Z=\{(z,f(z)) \,|\, z\in Z\}$; $f|_Z$ is called the restriction of $f$ onto $Z$.
Let $\gamma\leq \omega$.
A sequence $\delta$ of $X$ is a bijection $\delta : \gamma\longrightarrow X$.
Recall that $X$ is countable if and only if there exists a sequence of $X$.
Let $I$ be an index set, and $S_i\neq \emptyset$, $i\in I$, be sets.
A selector ${\mc S}$ over $\{S_i \,|\, i\in I\}$ is a mapping ${\mc S} : I\longrightarrow \bigcup \{S_i \,|\, i\in I\}$ such that
for all $i\in I$, ${\mc S}(i)\in S_i$.
We denote ${\mc S}\mi{el}(\{S_i \,|\, i\in I\})=\{{\mc S} \,|\, {\mc S}\ \text{\it is a selector over}\ \{S_i \,|\, i\in I\}\}$.
Let $c\in \mbb{R}^+$.
$\log c$ denotes the binary logarithm of $c$.
Let $f, g : \mbb{N}\longrightarrow \mbb{R}_0^+$.
$f$ is of the order of $g$, in symbols $f\in O(g)$, iff there exist $n_0$ and $c^*\in \mbb{R}_0^+$ such that
for all $n\geq n_0$, $f(n)\leq c^*\cdot g(n)$.

\section{First-order G\"{o}del logic}
\label{S2}

Throughout the paper, we shall use the common notions and notation of first-order logic.
By ${\mc L}$ we denote a first-order language. 
$\mi{Var}_{\mc L}$, $\mi{Func}_{\mc L}$, $\mi{Pred}_{\mc L}$, $\mi{Term}_{\mc L}$, $\mi{GTerm}_{\mc L}$, 
$\mi{Atom}_{\mc L}$ denotes
the set of all variables, function symbols, predicate symbols, terms, ground terms, atoms of ${\mc L}$.
$\mi{ar}_{\mc L} : \mi{Func}_{\mc L}\cup \mi{Pred}_{\mc L}\longrightarrow \mbb{N}$ denotes 
the mapping assigning an arity to every function and predicate symbol of ${\mc L}$.
$\mi{cn}\in \mi{Func}_{\mc L}$ such that $\mi{ar}_{\mc L}(\mi{cn})=0$ is called a constant symbol.
Let $\{0,1\}\subseteq C_{\mc L}\subseteq [0,1]$ be countable.
We assume a countable set of truth constants of ${\mc L}$
$\overline{C}_{\mc L}=\{\bar{c} \,|\, c\in C_{\mc L}\}$;
$\gz$, $\gu$ denotes the false, the true in ${\mc L}$; $\bar{c}$, $0<c<1$, is called an intermediate truth constant.
Let $x\in \overline{C}_{\mc L}$ and $X\subseteq \overline{C}_{\mc L}$.
Then there exists a unique $c\in C_{\mc L}$ such that $\bar{c}=x$.
We denote $\underline{x}=c$ and
$\underline{X}=\{\underline{x} \,|\, \underline{x}\in C_{\mc L}, x\in X\}$.
By $\mi{Form}_{\mc L}$ we designate the set of all formulae of ${\mc L}$ built up 
from $\mi{Atom}_{\mc L}$, $\overline{C}_{\mc L}$, $\mi{Var}_{\mc L}$
using the connectives: $\neg$, negation, $\del$, Delta, $\wedge$, conjunction, $\vee$, disjunction, $\rightarrow$, implication,  
$\leftrightarrow$, equivalence, $\geql$, equality, $\gle$, strict order, and
the quantifiers: $\forall$, the universal one, $\exists$, the existential
one.\footnote{We assume a decreasing connective and quantifier precedence:
              $\forall$, $\exists$, $\neg$, $\del$, $\geql$, $\gle$, $\wedge$, $\vee$, $\rightarrow$, $\leftrightarrow$.}
In the paper, we shall assume that ${\mc L}$ is a countable first-order language; 
hence, all the above mentioned sets of symbols and expressions are 
countable.\footnote{If the first-order language in question is not explicitly designated,
                    we shall write denotations without index.} 
Let $\varepsilon$, $\varepsilon_i$, $1\leq i\leq m$, $\upsilon_i$, $1\leq i\leq n$, be
either an expression or a set of expressions or a set of sets of expressions of ${\mc L}$, in general.
By $\mi{vars}(\varepsilon_1,\dots,\varepsilon_m)\subseteq \mi{Var}_{\mc L}$,
$\mi{freevars}(\varepsilon_1,\dots,\varepsilon_m)\subseteq \mi{Var}_{\mc L}$,
$\mi{boundvars}(\varepsilon_1,\dots,\varepsilon_m)\subseteq \mi{Var}_{\mc L}$,
$\mi{funcs}(\varepsilon_1,\dots,\varepsilon_m)\subseteq \mi{Func}_{\mc L}$,
$\mi{preds}(\varepsilon_1,\dots,\varepsilon_m)\subseteq \mi{Pred}_{\mc L}$,
$\mi{atoms}(\varepsilon_1,\dots,\varepsilon_m)\subseteq \mi{Atom}_{\mc L}$,
$\mi{tcons}(\varepsilon_1,\dots,\varepsilon_m)\subseteq \overline{C}_{\mc L}$
we denote the set of all variables, free variables, bound variables, function symbols, predicate symbols, atoms, truth constants of ${\mc L}$
occurring in $\varepsilon_1,\dots,\varepsilon_m$.
$\varepsilon$ is closed iff $\mi{freevars}(\varepsilon)=\emptyset$.
By $\emptyseq$ we denote the empty sequence.
Let $\varepsilon_1,\dots,\varepsilon_m$ and $\upsilon_1,\dots,\upsilon_n$ be sequences.
The length of $\varepsilon_1,\dots,\varepsilon_m$ is defined as $|\varepsilon_1,\dots,\varepsilon_m|=m$.
We define the concatenation of $\varepsilon_1,\dots,\varepsilon_m$ and $\upsilon_1,\dots,\upsilon_n$
as $(\varepsilon_1,\dots,\varepsilon_m),(\upsilon_1,\dots,\upsilon_n)=\varepsilon_1,\dots,\varepsilon_m,\upsilon_1,\dots,\upsilon_n$.
Note that the concatenation is 
associative.\footnote{Several simultaneous applications of the concatenation will be written without parentheses.}

Let $t\in \mi{Term}_{\mc L}$, $\phi\in \mi{Form}_{\mc L}$, $T\subseteq_{\mc F} \mi{Form}_{\mc L}$.
We define the size of $t$ by recursion on the structure of $t$ as follows:
\begin{equation*}
|t|=\left\{\begin{array}{ll}
           1                    &\ \text{\it if}\ t\in \mi{Var}_{\mc L}, \\[1mm]
           1+\sum_{i=1}^n |t_i| &\ \text{\it if}\ t=f(t_1,\dots,t_n).
           \end{array}
    \right. 
\end{equation*}
Subsequently, we define the size of $\phi$ by recursion on the structure of $\phi$ as follows:
\begin{equation*}
|\phi|=\left\{\begin{array}{ll}
              1+\sum_{i=1}^n |t_i| &\ \text{\it if}\ \phi=p(t_1,\dots,t_n)\in \mi{Atom}_{\mc L}, \\[1mm]
              1                    &\ \text{\it if}\ \phi\in \overline{C}_{\mc L}, \\[1mm]
              1+|\phi_1|           &\ \text{\it if}\ \phi=\diamond \phi_1, \\[1mm]
              1+|\phi_1|+|\phi_2|  &\ \text{\it if}\ \phi=\phi_1\diamond \phi_2, \\[1mm]
              2+|\phi_1|           &\ \text{\it if}\ \phi=Q x\, \phi_1.
              \end{array}
       \right.
\end{equation*}
Note that $|t|, |\phi|\geq 1$.
The size of $T$ is defined as $|T|=\sum_{\phi\in T} |\phi|$.
By $\mi{varseq}(\phi)$, $\mi{vars}(\mi{varseq}(\phi))\subseteq \mi{Var}_{\mc L}$, 
we denote the sequence of all variables of ${\mc L}$ occurring in $\phi$ which is built up via the left-right preorder traversal of $\phi$.
For example, $\mi{varseq}(\exists w\, (\forall x\, p(x,x,z)\vee \exists y\, q(x,y,z)))=w,x,x,x,z,y,x,y,z$ and $|w,x,x,x,z,y,x,y,z|=9$. 
A sequence of variables will often be denoted as $\bar{x}$, $\bar{y}$, $\bar{z}$, etc.
Let $Q\in \{\forall,\exists\}$ and $\bar{x}=x_1,\dots,x_n$ be a sequence of variables of ${\mc L}$.
By $Q \bar{x}\, \phi$ we denote $Q x_1\dots Q x_n\, \phi$.
Let $I\subseteq_{\mc F} \mbb{N}$ and $t_i\in \mi{Term}_{\mc L}$, $i\in I$.
We denote $\big\langle \{(i,t_i) \,|\, i\in I\} \big\rangle=t_{j_1},\dots,t_{j_\kappa}$
where $\{j_k \,|\, 1\leq k\leq \kappa\}=I$, and for all $1\leq k<\kappa$, $j_k<j_{k+1}$.
For example, $\big\langle \{(5,t_5),(2,t_2),(7,t_7),(4,t_4)\} \big\rangle=t_2,t_4,t_5,t_7$.
Let $x_1,\dots,x_n\in \mi{freevars}(\phi)$.
$\phi$ may be denoted as $\phi(x_1,\dots,x_n)$.
Let $t_1,\dots,t_n\in \mi{Term}_{\mc L}$ be closed.
By $\phi(t_1,\dots,t_n)\in \mi{Form}_{\mc L}$ or $\phi(x_1/t_1,\dots,x_n/t_n)\in \mi{Form}_{\mc L}$
we denote the instance of $\phi(x_1,\dots,x_n)$ built up by substituting $t_1,\dots,t_n$ for $x_1,\dots,x_n$
in the standard manner.

G\"{o}del logic is interpreted by the standard $\mbi{G}$-algebra 
augmented by the operators $\feql$, $\fle$, $\fdel$ for the connectives $\geql$, $\gle$, $\del$, respectively.
\begin{equation*}
\mbi{G}=([0,1],\leq,\fvee,\fwedge,\frightarrow,\fcoml,\feql,\fle,\fdel,0,1)
\end{equation*}
where $\fvee$, $\fwedge$ denotes the supremum, infimum operator on $[0,1]$;
\begin{alignat*}{2}
a\frightarrow b &= \left\{\begin{array}{ll}
                          1 &\ \text{\it if}\ a\leq b, \\[1mm]
                          b &\ \text{\it else};
                          \end{array}
                   \right. 
& 
\fcom{a}        &= \left\{\begin{array}{ll}
                          1 &\ \text{\it if}\ a=0, \\[1mm]
                          0 &\ \text{\it else};
                          \end{array}
                   \right. 
\\[2mm]
a\feql b        &= \left\{\begin{array}{ll}
                          1 &\ \text{\it if}\ a=b, \\[1mm]
                          0 &\ \text{\it else};
                          \end{array}
                   \right. 
& \qquad
a\fle b         &= \left\{\begin{array}{ll}
                          1 &\ \text{\it if}\ a<b, \\[1mm]
                          0 &\ \text{\it else};
                          \end{array}
                   \right. 
\\[2mm]
\fdel a         &= \left\{\begin{array}{ll}
                          1 &\ \text{\it if}\ a=1, \\[1mm]
                          0 &\ \text{\it else}.
                          \end{array}
                   \right.
\end{alignat*}
Recall that $\mbi{G}$ is a complete linearly ordered lattice algebra;
$\fvee$, $\fwedge$ is commutative, associative, idempotent, monotone; 
$0$, $1$ is its neutral 
element; 
%
%
the residuum operator $\frightarrow$ of $\fwedge$ satisfies the condition of residuation:
\begin{equation}
\label{eq0a}
\text{for all}\ a, b, c\in \mbi{G},\ a\fwedge b\leq c\Longleftrightarrow a\leq b\frightarrow c;
\end{equation}
G\"{o}del negation $\fcoml$ satisfies the condition:
\begin{equation}
\label{eq0b}
\text{for all}\ a\in \mbi{G},\ \fcom{a}=a\frightarrow 0;
\end{equation}
$\fdel$ satisfies the 
condition:\footnote{We assume a decreasing operator precedence: $\fcoml$, $\fdel$, $\feql$, $\fle$, $\fwedge$, $\fvee$, $\frightarrow$.}
\begin{equation}
\label{eq0kk}
\text{for all}\ a\in \mbi{G},\ \fdel a=a\feql 1.
\end{equation}
Note that the following properties hold:
\begin{alignat}{1}
\notag
& \hspace{-2.24mm} \text{for all}\ a, b, c\in \mbi{G}, \\[1mm]
\ctag{ceq0d}{(distributivity of $\fvee$ over $\fwedge$)}
& a\fvee b\fwedge c=(a\fvee b)\fwedge (a\fvee c), \\[1mm]
\ctag{ceq0c}{(distributivity of $\fwedge$ over $\fvee$)}
& a\fwedge (b\fvee c)=a\fwedge b\fvee a\fwedge c, \\[1mm]
\label{eq0f}
& a\frightarrow b\fvee c=(a\frightarrow b)\fvee (a\frightarrow c), \\[1mm]
\label{eq0e}
& a\frightarrow b\fwedge c=(a\frightarrow b)\fwedge (a\frightarrow c), \\[1mm]
\label{eq0h}
& a\fvee b\frightarrow c=(a\frightarrow c)\fwedge (b\frightarrow c), \\[1mm]
\label{eq0g}
& a\fwedge b\frightarrow c=(a\frightarrow c)\fvee (b\frightarrow c), \\[1mm]
\label{eq0i}
& a\frightarrow (b\frightarrow c)=a\fwedge b\frightarrow c, \\[1mm]
\label{eq0j}
& ((a\frightarrow b)\frightarrow b)\frightarrow b=a\frightarrow b, \\[1mm]
\label{eq0k}
& (a\frightarrow b)\frightarrow c=((a\frightarrow b)\frightarrow b)\fwedge (b\frightarrow c)\fvee c, \\[1mm]
\label{eq0jj}
& (a\frightarrow b)\frightarrow 0=((a\frightarrow 0)\frightarrow 0)\fwedge (b\frightarrow 0).
\end{alignat}

An interpretation ${\mc I}$ for ${\mc L}$ is a triple
$\big({\mc U}_{\mc I},\{f^{\mc I} \,|\, f\in \mi{Func}_{\mc L}\},\{p^{\mc I} \,|\, p\in \mi{Pred}_{\mc L}\}\big)$ defined as follows: 
${\mc U}_{\mc I}\neq \emptyset$ is the universum of ${\mc I}$;
every $f\in \mi{Func}_{\mc L}$ is interpreted as a function $f^{\mc I} : {\mc U}_{\mc I}^{\mi{ar}_{\mc L}(f)}\longrightarrow {\mc U}_{\mc I}$;
every $p\in \mi{Pred}_{\mc L}$ is interpreted as a $[0,1]$-relation $p^{\mc I} : {\mc U}_{\mc I}^{\mi{ar}_{\mc L}(p)}\longrightarrow [0,1]$.
A variable assignment in ${\mc I}$ is a mapping $\mi{Var}_{\mc L}\longrightarrow {\mc U}_{\mc I}$. 
We denote the set of all variable assignments in ${\mc I}$ as ${\mc S}_{\mc I}$.
Let $e\in {\mc S}_{\mc I}$ and $u\in {\mc U}_{\mc I}$.
A variant $e[x/u]\in {\mc S}_{\mc I}$ of $e$ with respect to $x$ and $u$ is defined by
\begin{equation*}
e[x/u](z)=\left\{\begin{array}{ll}
                 u    &\ \text{\it if}\ z=x, \\[1mm]
                 e(z) &\ \text{\it else}.
                 \end{array}
          \right.
\end{equation*}  
In ${\mc I}$ with respect to $e$, 
we define the value $\|t\|_e^{\mc I}\in {\mc U}_{\mc I}$ of $t$ by recursion on the structure of $t$,
the value $\|\bar{x}\|_e^{\mc I}\in {\mc U}_{\mc I}^{|\bar{x}|}$ of $\bar{x}$,
the truth value $\|\phi\|_e^{\mc I}\in [0,1]$ of $\phi$ by recursion on the structure of $\phi$, as follows:
{\footnotesize
\begin{alignat*}{2}
&    t\in \mi{Var}_{\mc L}, 
& &\ \|t\|_e^{\mc I}=e(t); \\[1mm]
&    t=f(t_1,\dots,t_n), 
& &\ \|t\|_e^{\mc I}=f^{\mc I}(\|t_1\|_e^{\mc I},\dots,\|t_n\|_e^{\mc I}); \\[2mm]
&    \bar{x}=x_1,\dots,x_{|\bar{x}|}, 
& &\ \|\bar{x}\|_e^{\mc I}=e(x_1),\dots,e(x_{|\bar{x}|}); \\[2mm]
&    \phi=p(t_1,\dots,t_n), 
& &\ \|\phi\|_e^{\mc I}=p^{\mc I}(\|t_1\|_e^{\mc I},\dots,\|t_n\|_e^{\mc I}); \\[1mm]
&    \phi=c\in \overline{C}_{\mc L},     
& &\ \|\phi\|_e^{\mc I}=\underline{c}; \\[1mm]
&    \phi=\neg \phi_1,
& &\ \|\phi\|_e^{\mc I}=\fcomd{\|\phi_1\|_e^{\mc I}}; \\[1mm]
&    \phi=\del \phi_1,
& &\ \|\phi\|_e^{\mc I}=\fdel \|\phi_1\|_e^{\mc I}; \\[1mm]
&    \phi=\phi_1\diamond \phi_2,
& &\ \|\phi\|_e^{\mc I}=\|\phi_1\|_e^{\mc I}\fdiamond \|\phi_2\|_e^{\mc I}, \quad \diamond\in \{\wedge,\vee,\rightarrow,\geql,\gle\}; \\[1mm]
&    \phi=\phi_1\leftrightarrow \phi_2,
& &\ \|\phi\|_e^{\mc I}=(\|\phi_1\|_e^{\mc I}\frightarrow \|\phi_2\|_e^{\mc I})\fwedge
                        (\|\phi_2\|_e^{\mc I}\frightarrow \|\phi_1\|_e^{\mc I}); \\[1mm]
&    \phi=\forall x\, \phi_1,
& &\ \|\phi\|_e^{\mc I}=\bigfwedge_{u\in {\mc U}_{\mc I}} \|\phi_1\|_{e[x/u]}^{\mc I}; \\[1mm]
&    \phi=\exists x\, \phi_1,
& &\ \|\phi\|_e^{\mc I}=\bigfvee_{u\in {\mc U}_{\mc I}} \|\phi_1\|_{e[x/u]}^{\mc I}.
\end{alignat*}}%
Let $\phi$ be closed.
Then, for all $e, e'\in {\mc S}_{\mc I}$, $\|\phi\|_e^{\mc I}=\|\phi\|_{e'}^{\mc I}$.
Note that ${\mc S}_{\mc I}\neq \emptyset$.
We denote $\|\phi\|^{\mc I}=\|\phi\|_e^{\mc I}$.

Let ${\mc L}'$ be a first-order language, and ${\mc I}$, ${\mc I}'$ be interpretations for ${\mc L}$, ${\mc L}'$, respectively.
${\mc L}'$ is an expansion of ${\mc L}$ iff $\mi{Func}_{{\mc L}'}\supseteq \mi{Func}_{\mc L}$ and 
$\mi{Pred}_{{\mc L}'}\supseteq \mi{Pred}_{\mc L}$;
on the other side, we say that ${\mc L}$ is a reduct of ${\mc L}'$.
${\mc I}'$ is an expansion of ${\mc I}$ to ${\mc L}'$
iff ${\mc L}'$ is an expansion of ${\mc L}$, ${\mc U}_{{\mc I}'}={\mc U}_{\mc I}$,
for all $f\in \mi{Func}_{\mc L}$, $f^{{\mc I}'}=f^{\mc I}$,
for all $p\in \mi{Pred}_{\mc L}$, $p^{{\mc I}'}=p^{\mc I}$;
on the other side, we say that ${\mc I}$ is a reduct of ${\mc I}'$ to ${\mc L}$, in symbols ${\mc I}={\mc I}'|_{\mc L}$.

A theory of ${\mc L}$ is a set of formulae of ${\mc L}$.
$\phi$ is true in ${\mc I}$ with respect to $e$, written as ${\mc I}\models_e \phi$, iff $\|\phi\|_e^{\mc I}=1$.
${\mc I}$ is a model of $\phi$, in symbols ${\mc I}\models \phi$, iff, for all $e\in {\mc S}_{\mc I}$, ${\mc I}\models_e \phi$.
Let $\phi'\in \mi{Form}_{\mc L}$ and $T\subseteq \mi{Form}_{\mc L}$.
${\mc I}$ is a model of $T$, in symbols ${\mc I}\models T$, iff, for all $\phi\in T$, ${\mc I}\models \phi$.
$\phi$ is a logically valid formula iff, for every interpretation ${\mc I}$ for ${\mc L}$, ${\mc I}\models \phi$.
$\phi$ is equivalent to $\phi'$, in symbols $\phi\eqvl \phi'$, 
iff, for every interpretation ${\mc I}$ for ${\mc L}$ and $e\in {\mc S}_{\mc I}$, $\|\phi\|_e^{\mc I}=\|\phi'\|_e^{\mc I}$.

\section{Translation to clausal form}
\label{S3}

In \cite{Guller2023c}, Section III, we have introduced a clausal fragment in G\"{o}del logic and 
proposed translation of formulae to clausal form.
We briefly recall the basic notions and notation.

Let $a\in \mi{Form}_{\mc L}$.
$a$ is a quantified atom of ${\mc L}$ iff $a=Q x\, p(t_0,\dots,t_n)$
where $p(t_0,\dots,t_n)\in \mi{Atom}_{\mc L}$, $x\in \mi{vars}(p(t_0,\dots,t_n))$,
for all $i\leq n$, either $t_i=x$ or $x\not\in \mi{vars}(t_i)$.
$\mi{QAtom}_{\mc L}\subseteq \mi{Form}_{\mc L}$ denotes the set of all quantified atoms of ${\mc L}$.
$\mi{QAtom}_{\mc L}^Q\subseteq \mi{QAtom}_{\mc L}$, $Q\in \{\forall,\exists\}$, denotes the set of all quantified atoms of ${\mc L}$ 
of the form $Q x\, a$.
Let $\varepsilon_i$, $1\leq i\leq m$, be
either an expression or a set of expressions or a set of sets of expressions of ${\mc L}$, in general.
By $\mi{qatoms}(\varepsilon_1,\dots,\varepsilon_m)\subseteq \mi{QAtom}_{\mc L}$ we denote the set of all quantified atoms of ${\mc L}$
occurring in $\varepsilon_1,\dots,\varepsilon_m$.
We denote $\mi{qatoms}^Q(\varepsilon_1,\dots,\varepsilon_m)=\mi{qatoms}(\varepsilon_1,\dots,\varepsilon_m)\cap \mi{QAtom}_{\mc L}^Q$, 
$Q\in \{\forall,\exists\}$.
Let $p(t_1,\dots,t_n)\in \mi{Atom}_{\mc L}$, $c\in \overline{C}_{\mc L}$, $Q x\, p(t_0,\dots,t_n)\in \mi{QAtom}_{\mc L}$.
We denote
\begin{alignat*}{1}
& \mi{freetermseq}(p(t_1,\dots,t_n))=t_1,\dots,t_n, \\
& \mi{freetermseq}(c)=\emptyseq, \\
& \mi{freetermseq}(Q x\, p(t_0,\dots,t_n))= \\
& \hspace{36.85mm} \big\langle \{(i,t_i) \,|\, i\leq n, x\not\in \mi{vars}(t_i)\} \big\rangle, \\
& \mi{boundindset}(Q x\, p(t_0,\dots,t_n))=\{i \,|\, i\leq n, t_i=x\}\neq \emptyset.
\end{alignat*}

Order clauses in G\"{o}del logic are introduced as follows.
Let $l\in \mi{Form}_{\mc L}$.
$l$ is an order literal of ${\mc L}$ iff $l=\varepsilon_1\diamond \varepsilon_2$,
$\varepsilon_i\in \mi{Atom}_{\mc L}\cup \overline{C}_{\mc L}\cup \mi{QAtom}_{\mc L}$, $\diamond\in \{\geql,\gle\}$.
The set of all order literals of ${\mc L}$ is designated as $\mi{OrdLit}_{\mc L}\subseteq \mi{Form}_{\mc L}$.
An order clause of ${\mc L}$ is a finite set of order literals of ${\mc L}$.
Since $=$ is symmetric, $\geql$ is commutative;
hence, for all $\varepsilon_1\geql \varepsilon_2\in \mi{OrdLit}_{\mc L}$, we identify
$\varepsilon_1\geql \varepsilon_2$ with $\varepsilon_2\geql \varepsilon_1\in \mi{OrdLit}_{\mc L}$ with respect to order clauses.
An order clause $\{l_0,\dots,l_n\}\neq \emptyset$ is written in the form $l_0\vee\cdots\vee l_n$.
The empty order clause $\emptyset$ is denoted as $\square$.
An order clause $\{l\}$ is called unit and denoted as $l$;
if it does not cause the ambiguity with the denotation of the single order literal $l$ in a given context.
We designate the set of all order clauses of ${\mc L}$ as $\mi{OrdCl}_{\mc L}$.
Let $l, l_0,\dots,l_n\in \mi{OrdLit}_{\mc L}$ and $C, C'\in \mi{OrdCl}_{\mc L}$.
We define the size of $C$ as $|C|=\sum_{l\in C} |l|$.
By $l_0\vee\cdots\vee l_n\vee C$ we denote $\{l_0,\dots,l_n\}\cup C$
where, for all $i, i'\leq n$ and $i\neq i'$, $l_i\not\in C$, $l_i\neq l_{i'}$.
By $C\vee C'$ we denote $C\cup C'$.
$C$ is a subclause of $C'$, in symbols $C\sqsubseteq C'$, iff $C\subseteq C'$.
An order clausal theory of ${\mc L}$ is a set of order clauses of ${\mc L}$.
A unit order clausal theory is a set of unit order clauses; in other words, we say that an order clausal theory is unit.

Let $\phi, \phi'\in \mi{Form}_{\mc L}$, $T, T'\subseteq \mi{Form}_{\mc L}$, $S, S'\subseteq \mi{OrdCl}_{\mc L}$,
${\mc I}$ be an interpretation for ${\mc L}$, $e\in {\mc S}_{\mc I}$.
$C$ is true in ${\mc I}$ with respect to $e$, written as ${\mc I}\models_e C$, iff there exists $l^*\in C$ such that ${\mc I}\models_e l^*$.
${\mc I}$ is a model of $C$, in symbols ${\mc I}\models C$, iff, for all $e\in {\mc S}_{\mc I}$, ${\mc I}\models_e C$.
${\mc I}$ is a model of $S$, in symbols ${\mc I}\models S$, iff, for all $C\in S$, ${\mc I}\models C$.
Let $\varepsilon_1\in \{\phi,T,C,S\}$ and $\varepsilon_2\in \{\phi',T',C',S'\}$.
$\varepsilon_2$ is a logical consequence of $\varepsilon_1$, in symbols $\varepsilon_1\models \varepsilon_2$,
iff, for every interpretation ${\mc I}$ for ${\mc L}$, if ${\mc I}\models \varepsilon_1$, then ${\mc I}\models \varepsilon_2$.
$\varepsilon_1$ is satisfiable iff there exists a model of $\varepsilon_1$.
$\varepsilon_1$ is equisatisfiable to $\varepsilon_2$ iff $\varepsilon_1$ is satisfiable if and only if $\varepsilon_2$ is satisfiable.
Let $S\subseteq_{\mc F} \mi{OrdCl}_{\mc L}$.
We define the size of $S$ as $|S|=\sum_{C\in S} |C|$.
We shall assume a fresh function symbol $\tilde{f}_0$ such that $\tilde{f}_0\not\in \mi{Func}_{\mc L}$.
We denote $\mbb{I}=\mbb{N}\times \mbb{N}$; $\mbb{I}$ will be exploited as a countably infinite set of indices.
We shall assume a countably infinite set of fresh predicate symbols $\tilde{\mbb{P}}=\{\tilde{p}_\mbbm{i} \,|\, \mbbm{i}\in \mbb{I}\}$ such that 
$\tilde{\mbb{P}}\cap \mi{Pred}_{\mc L}=\emptyset$.

The main Theorem 5 of Section III, \cite{Guller2023c} provides 
the reduction of a deduction problem $T\models \phi$, $\phi\in \mi{Form}_{\mc L}$, $T\subseteq \mi{Form}_{\mc L}$, 
to an unsatisfiability problem of a certain order clausal theory.

\begin{theorem}[\mbox{\rm Theorem 5, Section III, \cite{Guller2023c}}]
\label{T1}
Let $n_0\in \mbb{N}$, $\phi\in \mi{Form}_{\mc L}$, $T\subseteq \mi{Form}_{\mc L}$. 
There exist $J_T^\phi\subseteq \{(i,j) \,|\, i\geq n_0\}\subseteq \mbb{I}$ and
$S_T^\phi\subseteq \mi{OrdCl}_{{\mc L}\cup \{\tilde{p}_\mbbm{j} \,|\, \mbbm{j}\in J_T^\phi\}}$ such that 
\begin{enumerate}[\rm (i)]
\item
there exists an interpretation ${\mf A}$ for ${\mc L}$, and ${\mf A}\models T$, ${\mf A}\not\models \phi$, if and only if 
there exists an interpretation ${\mf A}'$ for ${\mc L}\cup \{\tilde{p}_\mbbm{j} \,|\, \mbbm{j}\in J_T^\phi\}$ and ${\mf A}'\models S_T^\phi$, 
satisfying ${\mf A}={\mf A}'|_{\mc L}$; 
\item
$T\models \phi$ if and only if $S_T^\phi$ is unsatisfiable;
\item
if $T\subseteq_{\mc F} \mi{Form}_{\mc L}$, then $J_T^\phi\subseteq_{\mc F} \{(i,j) \,|\, i\geq n_0\}\subseteq \mbb{I}$, 
$\|J_T^\phi\|\in O(|T|+|\phi|)$, 
$S_T^\phi\subseteq_{\mc F} \mi{OrdCl}_{{\mc L}\cup \{\tilde{p}_\mbbm{j} \,|\, \mbbm{j}\in J_T^\phi\}}$, $|S_T^\phi|\in O(|T|^2+|\phi|^2)$; 
the number of all elementary operations of the translation of $T$ and $\phi$ to $S_T^\phi$ is in $O(|T|^2+|\phi|^2)$;
the time complexity of the translation of $T$ and $\phi$ to $S_T^\phi$ is 
in $O(|T|^2\cdot \log (1+n_0+|T|)+|\phi|^2\cdot (\log (1+n_0)+\log |\phi|))$;
\item
$\mi{tcons}(S_T^\phi)-\{\gz,\gu\}\subseteq (\mi{tcons}(\phi)\cup \mi{tcons}(T))-\{\gz,\gu\}$.
\end{enumerate}
\end{theorem}

\begin{proof}
The translation is based on the interpolation rules, Tables II and III, \cite{Guller2023c}.
Cf. Theorem 5, \cite{Guller2023c}.
\qed
\end{proof}

Notice that the renaming subformulae technique and structure preserving translation
have been described, among others, in \cite{Tse70,PLGR86,Boy92,Hah94b,NOROWE98,She04}.

\section{Hyperresolution over order clauses}
\label{S4}

In this section, we propose a modified order hyperresolution calculus from \cite{Guller2019b}, operating over order clausal theories, and 
prove its refutational soundness/completeness.

\subsection{Order hyperresolution rules}
\label{S4.1}

At first, we introduce some basic notions and notation concerning chains of order literals.
A chain of ${\mc L}$ is a sequence $\varepsilon_0\diamond_0 \varepsilon_1,\dots,\varepsilon_n\diamond_n \varepsilon_{n+1}$,
$\varepsilon_i\diamond_i \varepsilon_{i+1}\in \mi{OrdLit}_{\mc L}$. 
Let $\Xi=\varepsilon_0\diamond_0 \varepsilon_1,\dots,\varepsilon_n\diamond_n \varepsilon_{n+1}$ be a chain of ${\mc L}$.
$\varepsilon_0$ is the beginning element of $\Xi$ and $\varepsilon_{n+1}$ the ending element of $\Xi$.
$\varepsilon_0\, \Xi\, \varepsilon_{n+1}$ denotes $\Xi$ together with its beginning and ending element.
$\Xi$ is an equality chain iff, for all $i\leq n$, $\diamond_i=\geql$.
$\Xi$ is an increasing chain iff there exists $i^*\leq n$ such that $\diamond_{i^*}=\gle$.
$\Xi$ is a contradiction iff $\Xi$ is an increasing chain of the form 
$\varepsilon_0\, \Xi\, \gz$ or $\gu\, \Xi\, \varepsilon_{n+1}$ or $\varepsilon_0\, \Xi\, \varepsilon_0$.
Let $S\subseteq \mi{OrdCl}_{\mc L}$ be unit. 
$\Xi$ is a chain, an equality chain, an increasing chain, a contradiction of $S$ 
iff, for all $i\leq n$, $\varepsilon_i\diamond_i \varepsilon_{i+1}\in S$.

We assume that the reader is familiar with the standard notions and notation of substitutions.
In Appendix, Subsection~\ref{S10.1},
we introduce a few definitions and denotations; some of them are slightly different from the standard ones, but found to be more convenient.

We shall assume a countably infinite set of fresh function symbols $\tilde{\mbb{W}}=\{\tilde{w}_\mbbm{i} \,|\, \mbbm{i}\in \mbb{I}\}$ such that
$\tilde{\mbb{W}}\cap (\mi{Func}_{\mc L}\cup \{\tilde{f}_0\})=\emptyset$.
Let $S\subseteq \mi{OrdCl}_{\mc L}$ and ${\mc L}'$ be an expansion of ${\mc L}$.
We denote $\mi{GOrdCl}_{\mc L}=\{C \,|\, C\in \mi{OrdCl}_{\mc L}\ \text{\it is closed}\}$,
$\mi{GInst}_{{\mc L}'}(S)=\{C \,|\, C\in \mi{GOrdCl}_{{\mc L}'}\ \text{\it is an instance of}\ S\}$,
$\mi{ordtcons}(S)=\{c_1\gle c_2 \,|\, c_1, c_2\in \mi{tcons}(S)\cup \{\gz,\gu\}, \underline{c_1}<\underline{c_2}\}\subseteq \mi{GOrdCl}_{\mc L}$.
Let $x_1,\dots,x_n\in \mi{freevars}(S)$.
$S$ may be denoted as $S(x_1,\dots,x_n)$.
Let $t_1,\dots,t_n\in \mi{Term}_{\mc L}$ be closed.
By $S(t_1,\dots,t_n)\subseteq \mi{OrdCl}_{\mc L}$ or $S(x_1/t_1,\dots,x_n/t_n)\subseteq \mi{OrdCl}_{\mc L}$
we denote the instance of $S(x_1,\dots,x_n)$ built up by substituting $t_1,\dots,t_n$ for $x_1,\dots,x_n$
in the standard manner.

We now introduce an order hyperresolution calculus.
Particularly, we provide two variants of it: basic and general one.
The basic calculus operates over closed order clauses, whereas the general one over arbitrary order clauses.
In Theorem \ref{T3} (Refutational Soundness and Completeness), a completeness argument is firstly given for the basic variant, and subsequently,
for the general one by means of Lemma~\ref{le6} (Lifting Lemma).
Let $\kappa\geq 1$, ${\mc L}_{\kappa-1}$, ${\mc L}_\kappa$ be first-order languages;
$S_{\kappa-1}\subseteq \mi{GOrdCl}_{{\mc L}_{\kappa-1}}$, $S_\kappa\subseteq \mi{GOrdCl}_{{\mc L}_\kappa}$ for the basic variant;
$S_{\kappa-1}\subseteq \mi{OrdCl}_{{\mc L}_{\kappa-1}}$, $S_\kappa\subseteq \mi{OrdCl}_{{\mc L}_\kappa}$ for the general variant.
Order hyperresolution rules are defined with respect to $\kappa$, ${\mc L}_{\kappa-1}$, $S_{\kappa-1}$.
The first rule is the central order hyperresolution one with an obvious intuition.
\begin{alignat}{1}
\ctag{ceq4hr0}{({\it Basic order hyperresolution rule})} \\[1mm]
\notag
& \dfrac{l_0\vee C_0,\dots,l_n\vee C_n\in S_{\kappa-1}}
        {\displaystyle{\bigvee_{i=0}^n C_i\in S_\kappa}}; \\[1mm]
\notag
& l_0,\dots,l_n\ \text{\it is a contradiction}. 
\end{alignat}
We say that $\bigvee_{i=0}^n C_i$ is a basic order hyperresolvent of $l_0\vee C_0,\dots,l_n\vee C_n$.
The basic order hyperresolution rule can be generalised as follows:
\begin{alignat}{1}
\ctag{ceq4hr5}{({\it Order hyperresolution rule})} \\[1mm]
\notag
& \dfrac{\displaystyle{
         \bigvee_{j=0}^{k_0} \varepsilon_j^0\diamond_j^0 \upsilon_j^0\vee \bigvee_{j=1}^{m_0} l_j^0,...,\!
         \bigvee_{j=0}^{k_n} \varepsilon_j^n\diamond_j^n \upsilon_j^n\vee \bigvee_{j=1}^{m_n} l_j^n\in \mi{Vrnt}(S_{\kappa-1})}}
        {\displaystyle{\Big(\bigvee_{i=0}^n \bigvee_{j=1}^{m_i} l_j^i\Big)\theta\in S_\kappa}}; \\[1mm]
\notag
& \begin{array}{l}
  \text{\it for all}\ i<i'\leq n, \\
  \quad
  \mi{freevars}(\bigvee_{j=0}^{k_i} \varepsilon_j^i\diamond_j^i \upsilon_j^i\vee \bigvee_{j=1}^{m_i} l_j^i)\cap \mbox{} \\
  \quad
  \mi{freevars}(\bigvee_{j=0}^{k_{i'}} \varepsilon_j^{i'}\diamond_j^{i'} \upsilon_j^{i'}\vee \bigvee_{j=1}^{m_{i'}} l_j^{i'})=\emptyset, \\[1mm]
  \theta\in \mi{mgu}\Big(\bigvee_{j=0}^{k_0} \varepsilon_j^0\diamond_j^0 \upsilon_j^0,l_1^0,\dots,l_{m_0}^0,\dots, \\
  \phantom{\theta\in \mi{mgu}\Big(}
                         \bigvee_{j=0}^{k_n} \varepsilon_j^n\diamond_j^n \upsilon_j^n,l_1^n,\dots,l_{m_n}^n, \\   
  \phantom{\theta\in \mi{mgu}\Big(}
                         \{\upsilon_0^0,\varepsilon_0^1\},\dots,\{\upsilon_0^{n-1},\varepsilon_0^n\},\{a,b\}\Big), \\
  \mi{dom}(\theta)=\mi{freevars}\big(\{\varepsilon_j^i\diamond_j^i \upsilon_j^i \,|\, j\leq k_i, i\leq n\}, \\
  \phantom{\mi{dom}(\theta)=\mi{freevars}\big(} 
                                     \{l_j^i \,|\, 1\leq j\leq m_i, i\leq n\}\big), \\[1mm]
  a=\varepsilon_0^0, b=\gu,\ \text{\it or}\ a=\upsilon_0^n, b=\gz,\ \text{\it or}\ a=\varepsilon_0^0, b=\upsilon_0^n, \\
  \text{\it there exists}\ i^*\leq n\ \text{\it such that}\ \diamond_0^{i^*}=\gle.    
  \end{array}
\end{alignat}
$\big(\bigvee_{i=0}^n \bigvee_{j=1}^{m_i} l_j^i\big)\theta$ is an order hyperresolvent
of $\bigvee_{j=0}^{k_0} \varepsilon_j^0\diamond_j^0 \upsilon_j^0\vee \bigvee_{j=1}^{m_0} l_j^0,\dots,
    \bigvee_{j=0}^{k_n} \varepsilon_j^n\diamond_j^n \upsilon_j^n\vee \bigvee_{j=1}^{m_n} l_j^n$.
Other auxiliary basic and general order rules are defined in Tables \ref{tab4}--\ref{tab6}.
\begin{table}[t]
\caption{Auxiliary basic and general order trichotomy rules}\label{tab4}
\vspace{-6mm}
\centering
\begin{minipage}[t]{\linewidth}
\footnotesize
\begin{IEEEeqnarray*}{L}
\hline \hline \\[-6mm]
\end{IEEEeqnarray*}
\begin{alignat}{1}
\ctag{ceq4hr000}{({\it Basic order trichotomy rule})} \\[1mm]
\notag
& \dfrac{\begin{array}{l}
         a\in \mi{atoms}(S_{\kappa-1}), b\in \mi{tcons}(S_{\kappa-1})-\{\gz,\gu\}, \\
         \mi{qatoms}(S)=\emptyset
         \end{array}}
        {a\gle b\vee a\geql b\vee b\gle a\in S_\kappa}.
\end{alignat}
\begin{alignat}{1}
\ctag{ceq4hr00}{({\it Basic order trichotomy rule})} \\[1mm]
\notag
& \dfrac{\begin{array}{l}
         a, b\in \mi{atoms}(S_{\kappa-1})\cup (\mi{tcons}(S_{\kappa-1})-\{\gz,\gu\}), \{a,b\}\not\subseteq \overline{C}_{\mc L}, \\
         \mi{qatoms}(S)\neq \emptyset
         \end{array}}
        {a\gle b\vee a\geql b\vee b\gle a\in S_\kappa}.
\end{alignat}
$a\gle b\vee a\geql b\vee b\gle a$ is a basic order trichotomy resolvent of $a$ and $b$.
\begin{alignat}{1}
\ctag{ceq4hr555}{({\it Order trichotomy rule})} \\[1mm]
\notag
& \dfrac{\begin{array}{l}
         a\in \mi{atoms}(S_{\kappa-1}), b\in \mi{tcons}(S_{\kappa-1})-\{\gz,\gu\}, \\
         \mi{qatoms}(S)=\emptyset
         \end{array}}
        {a\gle b\vee a\geql b\vee b\gle a\in S_\kappa}.
\end{alignat}
\begin{alignat}{1}
\ctag{ceq4hr55}{({\it Order trichotomy rule})} \\[1mm]
\notag
& \dfrac{\begin{array}{l}
         a, b\in \mi{atoms}(\mi{Vrnt}(S_{\kappa-1}))\cup (\mi{tcons}(S_{\kappa-1})-\{\gz,\gu\}), \{a,b\}\not\subseteq \overline{C}_{\mc L}, \\
         \mi{qatoms}(S)\neq \emptyset
         \end{array}}
        {a\gle b\vee a\geql b\vee b\gle a\in S_\kappa}; \\[1mm]
\notag
& \mi{vars}(a)\cap \mi{vars}(b)=\emptyset.
\end{alignat}
$a\gle b\vee a\geql b\vee b\gle a$ is an order trichotomy resolvent of $a$ and $b$.
\begin{IEEEeqnarray*}{L}
\hline \hline \\[2mm]
\end{IEEEeqnarray*}
\end{minipage}
\vspace{-6mm}
\end{table}
The trichotomy rules induce some total order over derived atoms and truth constants
in both cases $\mi{qatoms}(S)=\emptyset$ and $\mi{qatoms}(S)\neq \emptyset$, 
which is exploited in the proof of the completeness of the calculus, Theorem \ref{T3}.
\begin{table}[t]  
\caption{Auxiliary basic and general order quantification rules}\label{tab5}
\vspace{-6mm}
\centering
\begin{minipage}[t]{\linewidth}
\footnotesize
\begin{IEEEeqnarray*}{L}
\hline \hline \\[-6mm] 
\end{IEEEeqnarray*}
\begin{alignat}{1}
\ctag{ceq4hr1}{({\it Basic order $\forall$-quantification rule})} \\[1mm]
\notag
& \dfrac{\forall x\, a\in \mi{qatoms}^\forall(S_{\kappa-1})}
        {\forall x\, a\gle a\gamma\vee \forall x\, a\geql a\gamma\in S_\kappa}; \\[1mm]
\notag
& t\in \mi{GTerm}_{{\mc L}_{\kappa-1}}, \gamma=x/t\in \mi{Subst}_{{\mc L}_{\kappa-1}}, \mi{dom}(\gamma)=\{x\}=\mi{vars}(a).
\end{alignat}
$\forall x\, a\gle a\gamma\vee \forall x\, a\geql a\gamma$ is a basic order $\forall$-quantification resolvent of $\forall x\, a$.
\begin{alignat}{1}
\ctag{ceq4hr2}{({\it Basic order $\exists$-quantification rule})} \\[1mm]
\notag
& \dfrac{\exists x\, a\in \mi{qatoms}^\exists(S_{\kappa-1})}
        {a\gamma\gle \exists x\, a\vee a\gamma\geql \exists x\, a\in S_\kappa}; \\[1mm]
\notag
& t\in \mi{GTerm}_{{\mc L}_{\kappa-1}}, \gamma=x/t\in \mi{Subst}_{{\mc L}_{\kappa-1}}, \mi{dom}(\gamma)=\{x\}=\mi{vars}(a).
\end{alignat}
$a\gamma\gle \exists x\, a\vee a\gamma\geql \exists x\, a$ is a basic order $\exists$-quantification resolvent of $\exists x\, a$.
\begin{alignat}{1}
\ctag{ceq4hr6}{({\it Order $\forall$-quantification rule})} \\[1mm]
\notag
& \dfrac{\forall x\, a\in \mi{qatoms}^\forall(S_{\kappa-1})}
        {\forall x\, a\gle a\vee \forall x\, a\geql a\in S_\kappa}.
\end{alignat}
$\forall x\, a\gle a\vee \forall x\, a\geql a$ is an order $\forall$-quantification resolvent of $\forall x\, a$.
\begin{alignat}{1}
\ctag{ceq4hr7}{({\it Order $\exists$-quantification rule})} \\[1mm]
\notag
& \dfrac{\exists x\, a\in \mi{qatoms}^\exists(S_{\kappa-1})}
        {a\gle \exists x\, a\vee a\geql \exists x\, a\in S_\kappa}.
\end{alignat}
$a\gle \exists x\, a\vee a\geql \exists x\, a$ is an order $\exists$-quantification resolvent of $\exists x\, a$.
\begin{IEEEeqnarray*}{L}
\hline \hline \\[2mm]
\end{IEEEeqnarray*}  
\end{minipage}
\vspace{-6mm}
\end{table}   
The quantification rules order a derived quantified atom and its instances with respect to the semantic interpretation of the quantifiers 
by the infimum operator $\fwedge$ and the supremum one $\fvee$ for the universal and existential quantifier, respectively.
\begin{table}[t]
\caption{Auxiliary basic and general order witnessing rules}\label{tab6}
\vspace{-6mm}
\centering
\begin{minipage}[t]{\linewidth}
\footnotesize
\begin{IEEEeqnarray*}{L}
\hline \hline \\[-6mm]   
\end{IEEEeqnarray*}      
\begin{alignat}{1}
\ctag{ceq4hr3}{({\it Basic order $\forall$-witnessing rule})} \\[1mm]
\notag
& \dfrac{\begin{array}{l}
         \forall x\, a\in \mi{qatoms}^\forall(S_{\kappa-1}), \\
         b\in \mi{atoms}(S_{\kappa-1})\cup \mi{tcons}(S_{\kappa-1})\cup \mi{qatoms}(S_{\kappa-1})
         \end{array}}
        {a\gamma\gle b\vee b\geql \forall x\, a\vee b\gle \forall x\, a\in S_\kappa}; \\[1mm]
\notag
& \begin{array}{l}
  \tilde{w}\in \tilde{\mbb{W}}-\mi{Func}_{{\mc L}_{\kappa-1}},
  \mi{ar}(\tilde{w})=|\mi{freetermseq}(\forall x\, a),\mi{freetermseq}(b)|, \\
  \gamma=x/\tilde{w}(\mi{freetermseq}(\forall x\, a),\mi{freetermseq}(b))\in \mi{Subst}_{{\mc L}_{\kappa-1}\cup \{\tilde{w}\}}, \\
  \mi{dom}(\gamma)=\{x\}=\mi{vars}(a).
  \end{array}
\end{alignat}
$a\gamma\gle b\vee b\geql \forall x\, a\vee b\gle \forall x\, a$ is a basic order $\forall$-witnessing resolvent of $\forall x\, a$ and $b$.
\begin{alignat}{1}
\ctag{ceq4hr4}{({\it Basic order $\exists$-witnessing rule})} \\[1mm]
\notag
& \dfrac{\begin{array}{l}
         \exists x\, a\in \mi{qatoms}^\exists(S_{\kappa-1}), \\
         b\in \mi{atoms}(S_{\kappa-1})\cup \mi{tcons}(S_{\kappa-1})\cup \mi{qatoms}(S_{\kappa-1})
         \end{array}}
        {b\gle a\gamma\vee \exists x\, a\geql b\vee \exists x\, a\gle b\in S_\kappa}; \\[1mm]
\notag
& \begin{array}{l}
  \tilde{w}\in \tilde{\mbb{W}}-\mi{Func}_{{\mc L}_{\kappa-1}},
  \mi{ar}(\tilde{w})=|\mi{freetermseq}(\exists x\, a),\mi{freetermseq}(b)|, \\
  \gamma=x/\tilde{w}(\mi{freetermseq}(\exists x\, a),\mi{freetermseq}(b))\in \mi{Subst}_{{\mc L}_{\kappa-1}\cup \{\tilde{w}\}}, \\
  \mi{dom}(\gamma)=\{x\}=\mi{vars}(a).
  \end{array}
\end{alignat}
$b\gle a\gamma\vee \exists x\, a\geql b\vee \exists x\, a\gle b$ is a basic order $\exists$-witnessing resolvent of $\exists x\, a$ and $b$.
\begin{alignat}{1}
\ctag{ceq4hr8}{({\it Order $\forall$-witnessing rule})} \\[1mm]
\notag
& \dfrac{\begin{array}{l}
         \forall x\, a\in \mi{qatoms}^\forall(\mi{Vrnt}(S_{\kappa-1})), \\
         b\in \mi{atoms}(\mi{Vrnt}(S_{\kappa-1}))\cup \mi{tcons}(S_{\kappa-1})\cup \mi{qatoms}(\mi{Vrnt}(S_{\kappa-1}))
         \end{array}}
        {a\gamma\gle b\vee b\geql \forall x\, a\vee b\gle \forall x\, a\in S_\kappa}; \\[1mm]
\notag
& \begin{array}{l}
  \mi{freevars}(\forall x\, a)\cap \mi{freevars}(b)=\emptyset, \\
  \tilde{w}\in \tilde{\mbb{W}}-\mi{Func}_{{\mc L}_{\kappa-1}},
  \mi{ar}(\tilde{w})=|\mi{freetermseq}(\forall x\, a),\mi{freetermseq}(b)|, \\
  \gamma=x/\tilde{w}(\mi{freetermseq}(\forall x\, a),\mi{freetermseq}(b))\cup \mi{id}_{\mc L}|_{\mi{freevars}(\forall x\, a)} \\
  \phantom{\gamma}
  \in \mi{Subst}_{{\mc L}_{\kappa-1}\cup \{\tilde{w}\}}, \\
  \mi{dom}(\gamma)=\{x\}\cup \mi{freevars}(\forall x\, a)=\mi{vars}(a).
  \end{array}
\end{alignat}
$a\gamma\gle b\vee b\geql \forall x\, a\vee b\gle \forall x\, a$ is an order $\forall$-witnessing resolvent of $\forall x\, a$ and $b$.
\begin{alignat}{1}
\ctag{ceq4hr9}{({\it Order $\exists$-witnessing rule})} \\[1mm]
\notag
& \dfrac{\begin{array}{l}
         \exists x\, a\in \mi{qatoms}^\exists(\mi{Vrnt}(S_{\kappa-1})), \\
         b\in \mi{atoms}(\mi{Vrnt}(S_{\kappa-1}))\cup \mi{tcons}(S_{\kappa-1})\cup \mi{qatoms}(\mi{Vrnt}(S_{\kappa-1}))
         \end{array}}
        {b\gle a\gamma\vee \exists x\, a\geql b\vee \exists x\, a\gle b\in S_\kappa}; \\[1mm]
\notag
& \begin{array}{l}
  \mi{freevars}(\exists x\, a)\cap \mi{freevars}(b)=\emptyset, \\
  \tilde{w}\in \tilde{\mbb{W}}-\mi{Func}_{{\mc L}_{\kappa-1}},
  \mi{ar}(\tilde{w})=|\mi{freetermseq}(\exists x\, a),\mi{freetermseq}(b)|, \\
  \gamma=x/\tilde{w}(\mi{freetermseq}(\exists x\, a),\mi{freetermseq}(b))\cup \mi{id}_{\mc L}|_{\mi{freevars}(\exists x\, a)} \\
  \phantom{\gamma}
  \in \mi{Subst}_{{\mc L}_{\kappa-1}\cup \{\tilde{w}\}}, \\
  \mi{dom}(\gamma)=\{x\}\cup \mi{freevars}(\exists x\, a)=\mi{vars}(a).
  \end{array}
\end{alignat}
$b\gle a\gamma\vee \exists x\, a\geql b\vee \exists x\, a\gle b$ is an order $\exists$-witnessing resolvent of $\exists x\, a$ and $b$.
\begin{IEEEeqnarray*}{L}
\hline \hline \\[2mm]
\end{IEEEeqnarray*}  
\end{minipage}
\vspace{-6mm}
\end{table}   
The witnessing rules introduce a witness of both the infimum operator $\fwedge$ and the supremum one $\fvee$
for both cases $\forall x\, a\gle b$ and $b\gle \exists x\, a$, respectively,
where $\forall x\, a$, $\exists x\, a$ is a derived quantified atom, and $b$ is another derived atom or truth constant or quantified atom.
The witness is formed as a term with a fresh main function symbol $\tilde{w}$
from $\mi{freetermseq}(\forall x\, a)$, $\mi{freetermseq}(\exists x\, a)$ and $\mi{freetermseq}(b)$.
The witnessing rules together with the quantification ones also induce some total order 
over the derived quantified atom $\forall x\, a$, $\exists x\, a$ and other derived atoms, truth constants, quantified atoms,
which is exploited in the proof of the completeness.

We put ${\mc L}_0={\mc L}$ and $S_0=\emptyset\subseteq \mi{GOrdCl}_{{\mc L}_0}, \mi{OrdCl}_{{\mc L}_0}$.
Let ${\mc D}=C_1,\dots,C_n$, 
$C_\kappa\in \mi{GOrdCl}_{{\mc L}\cup \tilde{\mbb{W}}}, \mi{OrdCl}_{{\mc L}\cup \tilde{\mbb{W}}}$, $n\geq 1$.
${\mc D}$ is a deduction of $C_n$ from $S$ by basic order hyperresolution
iff, for all $1\leq \kappa\leq n$, $C_\kappa\in \mi{ordtcons}(S)\cup \mi{GInst}_{{\mc L}_{\kappa-1}}(S)$, or 
there exist $1\leq j_k\leq \kappa-1$, $k=1,\dots,m$, such that $C_\kappa$ is
a basic order resolvent of $C_{j_1},\dots,C_{j_m}\in S_{\kappa-1}$ 
using Rule~(\cref{ceq4hr0}), (\cref{ceq4hr000}), (\cref{ceq4hr00}), (\cref{ceq4hr1}), (\cref{ceq4hr2}), (\cref{ceq4hr3}), (\cref{ceq4hr4})
with respect to ${\mc L}_{\kappa-1}$ and $S_{\kappa-1}$;
${\mc D}$ is a deduction of $C_n$ from $S$ by order hyperresolution 
iff, for all $1\leq \kappa\leq n$, $C_\kappa\in \mi{ordtcons}(S)\cup S$, or 
there exist $1\leq j_k\leq \kappa-1$, $k=1,\dots,m$, such that $C_\kappa$ is 
an order resolvent of $C_{j_1},\dots,C_{j_m}\in \mi{Vrnt}(S_{\kappa-1})$ 
using Rule~(\cref{ceq4hr5}), (\cref{ceq4hr555}), (\cref{ceq4hr55}), (\cref{ceq4hr6}), (\cref{ceq4hr7}), (\cref{ceq4hr8}), (\cref{ceq4hr9})
with respect to ${\mc L}_{\kappa-1}$ and $S_{\kappa-1}$;
${\mc L}_\kappa$ and $S_\kappa$ are defined by recursion on $1\leq \kappa\leq n$ as follows:
\begin{alignat*}{1}
{\mc L}_\kappa &= \left\{\begin{array}{ll}
                         {\mc L}_{\kappa-1}\cup \{\tilde{w}\} 
                         &\ \text{\it in case of Rule}\ \text{(\cref{ceq4hr3})}, \text{(\cref{ceq4hr4})}; 
                                                        \text{(\cref{ceq4hr8})}, \text{(\cref{ceq4hr9})}; \\[1mm]
                         {\mc L}_{\kappa-1} 
                         &\ \text{\it else};
                         \end{array}
                  \right. \\[2mm]
S_\kappa       &= S_{\kappa-1}\cup \{C_\kappa\}\subseteq \mi{GOrdCl}_{{\mc L}_\kappa}, \mi{OrdCl}_{{\mc L}_\kappa};
\end{alignat*}
for all $\kappa\leq n$, ${\mc L}_\kappa$ is an expansion of ${\mc L}$ and a reduct of ${\mc L}\cup \tilde{\mbb{W}}$,
$S_\kappa=\{C_1,\dots,C_\kappa\}\subseteq \mi{GOrdCl}_{{\mc L}_\kappa}, \mi{OrdCl}_{{\mc L}_\kappa}\subseteq 
                                          \mi{GOrdCl}_{{\mc L}\cup \tilde{\mbb{W}}}, \mi{OrdCl}_{{\mc L}\cup \tilde{\mbb{W}}}$.
${\mc D}$ is a refutation of $S$ iff $C_n=\square$.
We denote
\begin{alignat*}{1}
\mi{clo}^{{\mc B}{\mc H}}(S)        &= \{C \,|\, \text{\it there exists a deduction of}\ C\ \text{\it from}\ S \\
&\phantom{\mbox{}=\{C \,|\, \mbox{}} 
                                       \text{\it by basic order hyperresolution}\} \\ 
                                    &\subseteq 
                                       \mi{GOrdCl}_{{\mc L}\cup \tilde{\mbb{W}}}, \\[1mm]
\mi{clo}^{\mc H}(S)                 &= \{C \,|\, \text{\it there exists a deduction of}\ C\ \text{\it from}\ S \\ 
&\phantom{\mbox{}=\{C \,|\, \mbox{}}
                                       \text{\it by order hyperresolution}\}\subseteq 
                                       \mi{OrdCl}_{{\mc L}\cup \tilde{\mbb{W}}}.
\end{alignat*}

\subsection{Refutational soundness and completeness}
\label{S4.2}

We are in position to prove the refutational soundness and completeness of the order hyperresolution calculus.
At first, we list some auxiliary lemmata.

\begin{lemma}[Lifting Lemma]  
\label{le6}
Let ${\mc L}$ contain at least one constant symbol and $S\subseteq \mi{OrdCl}_{\mc L}$.
Let $C\in \mi{clo}^{{\mc B}{\mc H}}(S)$.
There exists $C^*\in \mi{clo}^{\mc H}(S)$ such that $C$ is an instance of $C^*$.
\end{lemma}

\begin{proof}
Technical, analogous to the standard one.
\qed
\end{proof}

\begin{lemma}[Reduction Lemma]
\label{le7}
Let ${\mc L}$ contain at least one constant symbol and $S\subseteq \mi{OrdCl}_{\mc L}$.
Let $\{\bigvee_{j=0}^{k_i} \varepsilon_j^i\diamond_j^i \upsilon_j^i\vee C_i \,|\, i\leq n\}\subseteq \mi{clo}^{{\mc B}{\mc H}}(S)$ such that 
for all ${\mc S}\in {\mc S}\mi{el}(\{\{j \,|\, j\leq k_i\}_i \,|\, i\leq n\})$, 
there exists a contradiction of $\{\varepsilon_{{\mc S}(i)}^i\diamond_{{\mc S}(i)}^i \upsilon_{{\mc S}(i)}^i \,|\, i\leq n\}$.
There exists $\emptyset\neq I^*\subseteq \{i \,|\, i\leq n\}$ such that $\bigvee_{i\in I^*} C_i\in \mi{clo}^{{\mc B}{\mc H}}(S)$.
\end{lemma}

\begin{proof}
Technical, analogous to the one of Proposition 2, \cite{Guller2009}.
\qed
\end{proof}

\begin{lemma}[Unit Lemma]
\label{le8}
Let ${\mc L}$ contain at least one constant symbol and $S\subseteq \mi{OrdCl}_{\mc L}$.
Let $\square\not\in \mi{clo}^{{\mc B}{\mc H}}(S)=
                    \{\bigvee_{j=0}^{k_\iota} \varepsilon_j^\iota\diamond_j^\iota \upsilon_j^\iota \,|\, \iota<\gamma\}$, $\gamma\leq \omega$.
There exists ${\mc S}^*\in {\mc S}\mi{el}(\{\{j \,|\, j\leq k_\iota\}_\iota \,|\, \iota<\gamma\})$ such that 
there does not exist a contradiction 
of $\{\varepsilon_{{\mc S}^*(\iota)}^\iota\diamond_{{\mc S}^*(\iota)}^\iota \upsilon_{{\mc S}^*(\iota)}^\iota \,|\, \iota<\gamma\}$.
\end{lemma}

\begin{proof}
Technical, a straightforward consequence of K\"{o}nig's Lemma and Lemma \ref{le7}.
\qed
\end{proof}

\begin{theorem}[Refutational Soundness and Completeness]
\label{T3}
Let ${\mc L}$ contain at least one constant symbol; $\overline{C}_{\mc L}$ be finite;
$S\subseteq \mi{OrdCl}_{\mc L}$.
$\square\in \mi{clo}^{\mc H}(S)$ if and only if $S$ is unsatisfiable.
\end{theorem}

\begin{proof}
($\Longrightarrow$)
Let ${\mf A}$ be a model of $S$ for ${\mc L}$ and $C\in \mi{clo}^{\mc H}(S)\subseteq \mi{OrdCl}_{{\mc L}\cup \tilde{\mbb{W}}}$.
Then there exists an expansion ${\mf A}'$ of ${\mf A}$ to ${\mc L}\cup \tilde{\mbb{W}}$ such that ${\mf A}'\models C$.
The proof is by complete induction on the length of a deduction of $C$ from $S$ by order hyperresolution.
Let $\square\in \mi{clo}^{\mc H}(S)$ and $S$ be satisfiable.
Hence, there exists a model ${\mf A}$ of $S$ for ${\mc L}$;
there exists an expansion ${\mf A}'$ of ${\mf A}$ to ${\mc L}\cup \tilde{\mbb{W}}$ such that ${\mf A}'\models \square$, 
which is a contradiction;
$S$ is unsatisfiable.

($\Longleftarrow$)
Let $\square\not\in \mi{clo}^{\mc H}(S)$.
Then, by Lemma \ref{le6} for $\square$, $\square\not\in \mi{clo}^{{\mc B}{\mc H}}(S)$;
we have that ${\mc L}$, $\tilde{\mbb{W}}$ are countable;
$\mi{clo}^{{\mc B}{\mc H}}(S)$ is countable;
there exists $\gamma_1\leq \omega$ and
$\square\not\in \mi{clo}^{{\mc B}{\mc H}}(S)=
                \{\bigvee_{j=0}^{k_\iota} \varepsilon_j^\iota\diamond_j^\iota \upsilon_j^\iota \,|\, \iota<\gamma_1\}$;
by Lemma \ref{le8},
there exists ${\mc S}^*\in {\mc S}\mi{el}(\{\{j \,|\, j\leq k_\iota\}_\iota \,|\, \iota<\gamma_1\})$, and
there does not exist a contradiction 
of $\{\varepsilon_{{\mc S}^*(\iota)}^\iota\diamond_{{\mc S}^*(\iota)}^\iota \upsilon_{{\mc S}^*(\iota)}^\iota \,|\, \iota<\gamma_1\}$.
We put 
$\mbb{S}=\{\varepsilon_{{\mc S}^*(\iota)}^\iota\diamond_{{\mc S}^*(\iota)}^\iota \upsilon_{{\mc S}^*(\iota)}^\iota \,|\, \iota<\gamma_1\}\subseteq
         \mi{GOrdCl}_{{\mc L}\cup \tilde{\mbb{W}}}$.
We have that $\overline{C}_{\mc L}$ is finite.
Then $\mi{ordtcons}(S)\subseteq_{\mc F} \mi{clo}^{{\mc B}{\mc H}}(S)$ is unit;
$\mbb{S}\supseteq \mi{ordtcons}(S)$ is countable and unit, $\mi{tcons}(\mbb{S})\subseteq_{\mc F} \overline{C}_{\mc L}$;
there does not exist a contradiction of $\mbb{S}$.
We have that ${\mc L}$ contains at least one constant symbol.
Hence, there exists $\mi{cn}^*\in \mi{Func}_{\mc L}$ such that $\mi{ar}_{\mc L}(\mi{cn}^*)=0$.
We put $\tilde{\mbb{W}}^*=\mi{funcs}(\mi{clo}^{{\mc B}{\mc H}}(S))\cap \tilde{\mbb{W}}$,
$\tilde{\mbb{W}}^*\cap \mi{Func}_{\mc L}\subseteq \tilde{\mbb{W}}\cap (\mi{Func}_{\mc L}\cup \{\tilde{f}_0\})=\emptyset$,
\begin{alignat*}{1}
{\mc U}_{\mf A} &= \mi{GTerm}_{{\mc L}\cup \tilde{\mbb{W}}^*}, \mi{cn}^*\in {\mc U}_{\mf A}\neq \emptyset, \\
{\mc B}         &= \mi{atoms}(\mbb{S})\cup \mi{tcons}(\mbb{S})\cup \mi{qatoms}(\mbb{S}).
\end{alignat*}
Then $\mi{atoms}(\mbb{S})\cup \mi{qatoms}(\mbb{S})$ is countable;
there exist $\gamma_2\leq \omega$ and a sequence $\delta_2 : \gamma_2\longrightarrow \mi{atoms}(\mbb{S})\cup \mi{qatoms}(\mbb{S})$ 
of $\mi{atoms}(\mbb{S})\cup \mi{qatoms}(\mbb{S})\subseteq {\mc B}$.
Let $\varepsilon_1, \varepsilon_2\in {\mc B}$.
$\varepsilon_1\ceql \varepsilon_2$ iff there exists an equality chain $\varepsilon_1\, \Xi\, \varepsilon_2$ of $\mbb{S}$.
Note that $\ceql$ is a binary symmetric and transitive relation on ${\mc B}$.
$\varepsilon_1\cle \varepsilon_2$ iff there exists an increasing chain $\varepsilon_1\, \Xi\, \varepsilon_2$ of $\mbb{S}$.
Note that $\cle$ is a binary transitive relation on ${\mc B}$.
\begin{equation}
\label{eq7a}
\gz\nceql \gu, \gz\cle \gu, \gu\ncle \gz,\
\text{for all}\ \varepsilon\in {\mc B}, \varepsilon\ncle \gz, \gu\ncle \varepsilon, \varepsilon\ncle \varepsilon.
\end{equation}
The proof is straightforward; we have that there does not exist a contradiction of $\mbb{S}$.
Note that $\cle$ is also irreflexive and a partial, strict order on ${\mc B}$.

Let $\mi{tcons}(\mbb{S})\subseteq X\subseteq {\mc B}$.
A partial valuation ${\mc V}$ is a mapping ${\mc V} : X\longrightarrow [0,1]$ such that 
for all $c\in \mi{tcons}(\mbb{S})$, ${\mc V}(c)=\underline{c}$.
We denote $\mi{dom}({\mc V})=X$, $\mi{tcons}(\mbb{S})\subseteq \mi{dom}({\mc V})\subseteq {\mc B}$.
We define a partial valuation ${\mc V}_\alpha$ by recursion on $\alpha\leq \gamma_2$ as follows:
\begin{alignat*}{2}
& {\mc V}_0          & &= \{(c,\underline{c}) \,|\, c\in \mi{tcons}(\mbb{S})\}; \\[2mm]
& {\mc V}_\alpha     & &= {\mc V}_{\alpha-1}\cup \{(\delta_2(\alpha-1),\lambda_{\alpha-1})\} \\
& & & \hspace{19.8mm} 
                          (1\leq \alpha\leq \gamma_2\ \text{\it is a successor ordinal}), \\[2mm]
& & &
   \begin{alignedat}{2}
   & \mbb{E}_{\alpha-1} & &= \{{\mc V}_{\alpha-1}(a) \,|\, a\in \mi{dom}({\mc V}_{\alpha-1}), a\ceql \delta_2(\alpha-1)\}, \\[1mm]
   & \mbb{D}_{\alpha-1} & &= \{{\mc V}_{\alpha-1}(a) \,|\, a\in \mi{dom}({\mc V}_{\alpha-1}), a\cle \delta_2(\alpha-1)\}, \\[1mm]
   & \mbb{U}_{\alpha-1} & &= \{{\mc V}_{\alpha-1}(a) \,|\, a\in \mi{dom}({\mc V}_{\alpha-1}), \delta_2(\alpha-1)\cle a\}, \\[1mm]
   & \lambda_{\alpha-1} & &= 
     \left\{\begin{array}{ll}
            \dfrac{\bigfvee \mbb{D}_{\alpha-1}+\bigfwedge \mbb{U}_{\alpha-1}}{2}         
                                        &\ \ \text{\it if}\ \mbb{E}_{\alpha-1}=\emptyset, \\[1mm]
            \bigfvee \mbb{E}_{\alpha-1} &\ \ \text{\it else};
            \end{array}
     \right. 
   \end{alignedat} \\[2mm]
& {\mc V}_{\gamma_2} & &= \bigcup_{\alpha<\gamma_2} {\mc V}_\alpha \quad (\gamma_2\ \text{\it is a limit ordinal}). 
\end{alignat*}
Obviously, the valuation ${\mc V}_0$ valuates all the truth constants from $\mi{tcons}(\mbb{S})$ canonically.
In case of $1\leq \alpha\leq \gamma_2$ being a successor ordinal, 
we have to determine a value ${\mc V}_\alpha(\delta_2(\alpha-1))=\lambda_{\alpha-1}\in [0,1]$ so that
all constraints induced by the binary relations $\ceql$ and $\cle$ 
between elements already valuated from $\mi{dom}({\mc V}_{\alpha-1})$, on the one side,
and $\delta_2(\alpha-1)$, on the other side, are fulfilled.
Let $a\in \mi{dom}({\mc V}_{\alpha-1})$.
If $a\ceql \delta_2(\alpha-1)$, then this forces ${\mc V}_\alpha(\delta_2(\alpha-1))={\mc V}_{\alpha-1}(a)$;
analogously, if $a\cle \delta_2(\alpha-1)$, then ${\mc V}_{\alpha-1}(a)<{\mc V}_\alpha(\delta_2(\alpha-1))$;
if $\delta_2(\alpha-1)\cle a$, then ${\mc V}_\alpha(\delta_2(\alpha-1))<{\mc V}_{\alpha-1}(a)$.
To ensure this, we define sets $\mbb{E}_{\alpha-1}$, $\mbb{D}_{\alpha-1}$, $\mbb{U}_{\alpha-1}$ of values ${\mc V}_{\alpha-1}(a)$ 
for those $a\in \mi{dom}({\mc V}_{\alpha-1})$ which are equal to, less than, greater than $\delta_2(\alpha-1)$, respectively, 
under the binary relations $\ceql$ and $\cle$.
We get two cases.

$\mbb{E}_{\alpha-1}=\emptyset$. 
Hence, there is no $a\in \mi{dom}({\mc V}_{\alpha-1})$ equal to $\delta_2(\alpha-1)$.
Then $\lambda_{\alpha-1}$ is defined as the arithmetic mean of $\bigfvee \mbb{D}_{\alpha-1}$ and $\bigfwedge \mbb{U}_{\alpha-1}$.
It can be shown $\bigfvee \mbb{D}_{\alpha-1}<\bigfwedge \mbb{U}_{\alpha-1}$.
So, $\bigfvee \mbb{D}_{\alpha-1}<\lambda_{\alpha-1}<\bigfwedge \mbb{U}_{\alpha-1}$, and for all $a\in \mi{dom}({\mc V}_{\alpha-1})$,
${\mc V}_{\alpha-1}(a)<{\mc V}_\alpha(\delta_2(\alpha-1))$ if $a\cle \delta_2(\alpha-1)$;
${\mc V}_\alpha(\delta_2(\alpha-1))<{\mc V}_{\alpha-1}(a)$ if $\delta_2(\alpha-1)\cle a$.

$\mbb{E}_{\alpha-1}\neq \emptyset$.
Then $\lambda_{\alpha-1}=\bigfvee \mbb{E}_{\alpha-1}$.
It can be shown that for all ${\mc V}_{\alpha-1}(a_1), {\mc V}_{\alpha-1}(a_2)\in \mbb{E}_{\alpha-1}$, $a_1=a_2$ or $a_1\ceql a_2$,
${\mc V}_{\alpha-1}(a_1)={\mc V}_{\alpha-1}(a_2)$.

In case of $\gamma_2$ being a limit ordinal,
we define the valuation ${\mc V}_{\gamma_2}$ as the union of all the valuations ${\mc V}_\alpha$, $\alpha<\gamma_2$.
It can be shown $\mi{dom}({\mc V}_{\gamma_2})={\mc B}$.

\begin{alignat}{1}
\label{eq7b}
& \begin{minipage}[t]{\linewidth-10mm}
  For all $\alpha\leq \alpha'\leq \gamma_2$, ${\mc V}_\alpha$ is a partial valuation, 
  $\mi{dom}({\mc V}_\alpha)=\mi{tcons}(\mbb{S})\cup \delta_2[\alpha]$, ${\mc V}_\alpha\subseteq {\mc V}_{\alpha'}$.
  \end{minipage}
\end{alignat}   
The proof is by induction on $\alpha\leq \gamma_2$.

We list some auxiliary statements without proofs:
\begin{alignat}{1}
\label{eq7h}
& \begin{minipage}[t]{\linewidth-10mm}
  If $\mi{qatoms}(S)=\emptyset$, then $\mi{qatoms}(\mi{clo}^{{\mc B}{\mc H}}(S))=\emptyset$.
  \end{minipage}
\end{alignat}
\begin{equation}
\label{eq7hh}
\mi{tcons}(\mi{clo}^{{\mc B}{\mc H}}(S))=\mi{tcons}(\mbb{S})=\mi{tcons}(S)\cup \{\gz,\gu\}.
\end{equation}   
\begin{alignat}{1}
\label{eq7j}
& \begin{minipage}[t]{\linewidth-10mm}
  For all
  $a, b\in \mi{atoms}(\mi{clo}^{{\mc B}{\mc H}}(S))\cup \mi{tcons}(\mi{clo}^{{\mc B}{\mc H}}(S))\cup \mi{qatoms}(\mi{clo}^{{\mc B}{\mc H}}(S))$,
  there exist a deduction $C_1,\dots,C_n$, $n\geq 1$, from $S$ by basic order hyperresolution,
  associated ${\mc L}_n$, $S_n\subseteq \mi{GOrdCl}_{{\mc L}_n}$ such that 
  $a, b\in \mi{atoms}(S_n)\cup \mi{tcons}(S_n)\cup \mi{qatoms}(S_n)$.
  \end{minipage}
\end{alignat}
\begin{alignat}{1}
\label{eq7kk}
& \begin{minipage}[t]{\linewidth-10mm}
  For all 
  $\emptyset\neq A\subseteq_{\mc F} 
   \mi{atoms}(\mi{clo}^{{\mc B}{\mc H}}(S))\cup \mi{tcons}(\mi{clo}^{{\mc B}{\mc H}}(S))\cup \mi{qatoms}(\mi{clo}^{{\mc B}{\mc H}}(S))$,
  there exist a deduction $C_1,\dots,C_n$, $n\geq 1$, from $S$ by basic order hyperresolution,
  associated ${\mc L}_n$, $S_n\subseteq \mi{GOrdCl}_{{\mc L}_n}$ such that
  $A\subseteq \mi{atoms}(S_n)\cup \mi{tcons}(S_n)\cup \mi{qatoms}(S_n)$.
  \end{minipage}
\end{alignat}
\begin{alignat}{1}
\label{eq7gg}
& \begin{minipage}[t]{\linewidth-10mm}
  For all $a\in \mi{tcons}(\mbb{S})-\{\gz,\gu\}$, $b\in \mi{atoms}(\mbb{S})\cup \mi{qatoms}(\mbb{S})$, 
  either $a\cle b$ or $a\ceql b$ or $b\cle a$.
  \end{minipage}
\end{alignat}   
\begin{alignat}{1}
\label{eq7l}
& \begin{minipage}[t]{\linewidth-10mm}
  Let $\mi{qatoms}(S)\neq \emptyset$.
  For all $a, b\in {\mc B}-\{\gz,\gu\}$, either $a\cle b$ or ($a=b$ or $a\ceql b$) or $b\cle a$.
  \end{minipage}
\end{alignat}
\begin{equation}
\label{eq7c}
\begin{array}[t]{l}
\text{For all}\ \alpha\leq \gamma_2,\ \text{for all}\ a, b\in \mi{dom}({\mc V}_\alpha), \\
\quad \text{if}\ a\ceql b,\ \text{then}\ {\mc V}_\alpha(a)={\mc V}_\alpha(b); \\
\quad \text{if}\ a\cle b,\ \text{then}\ {\mc V}_\alpha(a)<{\mc V}_\alpha(b); \\
\quad \text{if}\ {\mc V}_\alpha(a)=0,\ \text{then}\ a=\gz\ \text{or}\ a\ceql \gz; \\
\quad \text{if}\ {\mc V}_\alpha(a)=1,\ \text{then}\ a=\gu\ \text{or}\ a\ceql \gu.
\end{array} 
\end{equation}    
The proof is by induction on $\alpha\leq \gamma_2$.

We put ${\mc V}={\mc V}_{\gamma_2}$, 
$\mi{dom}({\mc V})\overset{\text{(\ref{eq7b})}}{=\!\!=} \mi{tcons}(\mbb{S})\cup \delta[\gamma_2]=
 \mi{tcons}(\mbb{S})\cup \mi{atoms}(\mbb{S})\cup \mi{qatoms}(\mbb{S})={\mc B}$.
We further list some other auxiliary statements without proofs:
\begin{equation}
\label{eq7k}
\begin{array}[t]{l}
\text{For all}\ a, b\in {\mc B}, \\
\quad \text{if}\ a\ceql b,\ \text{then}\ {\mc V}(a)={\mc V}(b); \\
\quad \text{if}\ a\cle b,\ \text{then}\ {\mc V}(a)<{\mc V}(b).
\end{array}
\end{equation}
\begin{alignat}{1}
\label{eq7i}
& \begin{minipage}[t]{\linewidth-10mm}
  For all $Q x\, a\in \mi{qatoms}(\mi{clo}^{{\mc B}{\mc H}}(S))$ and $u\in {\mc U}_{\mf A}$,
  $a(x/u)\in \mi{atoms}(\mi{clo}^{{\mc B}{\mc H}}(S))$.
  \end{minipage}
\end{alignat}
\begin{equation}
\label{eq7m}
\begin{array}[t]{l}
\text{For all}\ a\in {\mc B}, \\
\quad \text{if}\ a=\forall x b,\ \text{then}\ {\mc V}(a)=\bigfwedge_{u\in {\mc U}_{\mf A}} {\mc V}(b(x/u)); \\
\quad \text{if}\ a=\exists x b,\ \text{then}\ {\mc V}(a)=\bigfvee_{u\in {\mc U}_{\mf A}} {\mc V}(b(x/u)).
\end{array} 
\end{equation}
We put
\begin{alignat*}{1}
f^{\mf A}(u_1,\dots,u_n) &= f(u_1,\dots,u_n), \\[1mm]
                            &\phantom{\mbox{}=\mbox{}}f\in \mi{Func}_{{\mc L}\cup \tilde{\mbb{W}}^*}, u_i\in {\mc U}_{\mf A}; \\[2mm]
p^{\mf A}(u_1,\dots,u_n) &= \left\{\begin{array}{ll}
                                   {\mc V}(p(u_1,\dots,u_n)) &\ \text{\it if}\ p(u_1,\dots,u_n)\in \\
                                                             &\ \phantom{\text{\it if}\ \mbox{}}\mi{atoms}(\mbb{S}), \\[1mm]
                                   0                         &\ \text{\it else},
                                   \end{array}
                            \right. \\[1mm]
                            &\phantom{\mbox{}=\mbox{}}p\in \mi{Pred}_{{\mc L}\cup \tilde{\mbb{W}}^*}, u_i\in {\mc U}_{\mf A};
\end{alignat*}
\begin{alignat*}{1}
& {\mf A}=\big({\mc U}_{\mf A},\{f^{\mf A} \,|\, f\in \mi{Func}_{{\mc L}\cup \tilde{\mbb{W}}^*}\},
                               \{p^{\mf A} \,|\, p\in \mi{Pred}_{{\mc L}\cup \tilde{\mbb{W}}^*}\}\big), \\
&\phantom{{\mf A}=\mbox{}}\hspace{0.38mm}
                               \text{an interpretation for}\ {\mc L}\cup \tilde{\mbb{W}}^*.
\end{alignat*}
\begin{alignat}{1}
\label{eq7p}
& \begin{minipage}[t]{\linewidth-10mm}
  For all $C\in S$ and $e\in {\mc S}_{\mf A}$, $C(e|_{\mi{freevars}(C)})\in \mi{clo}^{{\mc B}{\mc H}}(S)$.
  \end{minipage}
\end{alignat}
The proof is straightforward using (\ref{eq7kk}); 
$\mi{funcs}(C(e|_{\mi{freevars}(C)}))\cap \tilde{\mbb{W}}^*\subseteq_{\mc F} \mi{funcs}(\mi{clo}^{{\mc B}{\mc H}}(S))$,
there exists $A\subseteq_{\mc F} \mi{atoms}(\mi{clo}^{{\mc B}{\mc H}}(S))\cup \mi{qatoms}(\mi{clo}^{{\mc B}{\mc H}}(S))$ and       
$\mi{funcs}(A)\supseteq \mi{funcs}(C(e|_{\mi{freevars}(C)}))\cap \tilde{\mbb{W}}^*$.

It is straightforward to prove that for all $a\in {\mc B}$ and $e\in {\mc S}_{\mf A}$, $\|a\|_e^{\mf A}={\mc V}(a)$;
in case of $a\in \mi{qatoms}(\mbb{S})$, we use (\ref{eq7m}).

Let $\varepsilon_1\geql \varepsilon_2\in \mbb{S}$ and $e\in {\mc S}_{\mf A}$.
Then $\varepsilon_1, \varepsilon_2\in {\mc B}$, $\varepsilon_1\ceql \varepsilon_2$, 
by (\ref{eq7k}) for $\varepsilon_1$, $\varepsilon_2$, ${\mc V}(\varepsilon_1)={\mc V}(\varepsilon_2)$,
$\|\varepsilon_1\geql \varepsilon_2\|_e^{\mf A}=
 \|\varepsilon_1\|_e^{\mf A}\feql \|\varepsilon_2\|_e^{\mf A}=
 {\mc V}(\varepsilon_1)\feql {\mc V}(\varepsilon_2)=1$.

Let $\varepsilon_1\gle \varepsilon_2\in \mbb{S}$.
Then $\varepsilon_1, \varepsilon_2\in {\mc B}$, $\varepsilon_1\cle \varepsilon_2$, 
by (\ref{eq7k}) for $\varepsilon_1$, $\varepsilon_2$, ${\mc V}(\varepsilon_1)<{\mc V}(\varepsilon_2)$,
$\|\varepsilon_1\gle \varepsilon_2\|_e^{\mf A}= 
 \|\varepsilon_1\|_e^{\mf A}\fle \|\varepsilon_2\|_e^{\mf A}=   
 {\mc V}(\varepsilon_1)\fle {\mc V}(\varepsilon_2)=1$.

So, for all $l\in \mbb{S}$ and $e\in {\mc S}_{\mf A}$, 
for both cases $l=\varepsilon_1\geql \varepsilon_2$ and $l=\varepsilon_1\gle \varepsilon_2$, $\|l\|_e^{\mf A}=1$.

Let $C\in S\subseteq \mi{OrdCl}_{\mc L}$ and $e\in {\mc S}_{\mf A}$.
Then $e : \mi{Var}_{\mc L}\longrightarrow {\mc U}_{\mf A}=\mi{GTerm}_{{\mc L}\cup \tilde{\mbb{W}}^*}$,
${\mc L}\cup \tilde{\mbb{W}}^*$ is an expansion of ${\mc L}$;
by (\ref{eq7p}), $C(e|_{\mi{freevars}(C)})\in \mi{clo}^{{\mc B}{\mc H}}(S)$,
there exists $l^*\in C(e|_{\mi{freevars}(C)})$, and $l^*\in \mbb{S}$, $\|l^*\|_e^{\mf A}=1$;
there exists $l^{**}\in C$, and $l^{**}\in \mi{OrdLit}_{\mc L}\subseteq \mi{OrdLit}_{{\mc L}\cup \tilde{\mbb{W}}^*}$,
$l^{**}(e|_{\mi{freevars}(l^{**})})=l^*$;
for all $t\in \mi{Term}_{{\mc L}\cup \tilde{\mbb{W}}^*}$, 
$a\in \mi{Atom}_{{\mc L}\cup \tilde{\mbb{W}}^*}\cup \overline{C}_{\mc L}\cup \mi{QAtom}_{{\mc L}\cup \tilde{\mbb{W}}^*}$,
$l\in \mi{OrdLit}_{{\mc L}\cup \tilde{\mbb{W}}^*}$,
$\|t\|_e^{\mf A}=t(e|_{\mi{vars}(t)})=\|t(e|_{\mi{vars}(t)})\|_e^{\mf A}$, $\|a\|_e^{\mf A}=\|a(e|_{\mi{freevars}(a)})\|_e^{\mf A}$,
$\|l\|_e^{\mf A}=\|l(e|_{\mi{freevars}(l)})\|_e^{\mf A}$; 
the proof is by induction on $t$ and by definition;
$\|l^{**}\|_e^{\mf A}=                                                                                                             \linebreak[4]
 \|l^{**}(e|_{\mi{freevars}(l^{**})})\|_e^{\mf A}=\|l^*\|_e^{\mf A}=1$;
${\mf A}\models_e C$;
for all $C\in S$ and $e\in {\mc S}_{\mf A}$, ${\mf A}\models_e C$;
${\mf A}\models S$, ${\mf A}|_{\mc L}\models S$; 
$S$ is satisfiable.
\qed
\end{proof}

The deduction problem of a formula from a theory can be solved as follows:

\begin{corollary}
\label{cor2}
Let ${\mc L}$ contain at least one constant symbol, and $\overline{C}_{\mc L}$ be finite.
Let $n_0\in \mbb{N}$, $\phi\in \mi{Form}_{\mc L}$, $T\subseteq \mi{Form}_{\mc L}$.
There exist $J_T^\phi\subseteq \{(i,j) \,|\, i\geq n_0\}\subseteq \mbb{I}$ and
$S_T^\phi\subseteq \mi{OrdCl}_{{\mc L}\cup \{\tilde{p}_\mbbm{j} \,|\, \mbbm{j}\in J_T^\phi\}}$ such that
$T\models \phi$ if and only if $\square\in \mi{clo}^{\mc H}(S_T^\phi)$.
\end{corollary}

\begin{proof}
By Theorem \ref{T1}, there exist 
\begin{equation*}
J_T^\phi\subseteq \{(i,j) \,|\, i\geq n_0\}\subseteq \mbb{I},
S_T^\phi\subseteq \mi{OrdCl}_{{\mc L}\cup \{\tilde{p}_\mbbm{j} \,|\, \mbbm{j}\in J_T^\phi\}},
\end{equation*}
and (ii) of Theorem \ref{T1} holds;
we have by Theorem \ref{T1}(ii) that $T\models \phi$ if and only if $S_T^\phi$ is unsatisfiable;
${\mc L}\cup \{\tilde{p}_\mbbm{j} \,|\, \mbbm{j}\in J_T^\phi\}$ contains at least one constant symbol;
we have that $\overline{C}_{\mc L}$ is finite;
by Theorem~\ref{T3} for ${\mc L}\cup \{\tilde{p}_\mbbm{j} \,|\, \mbbm{j}\in J_T^\phi\}$, $S_T^\phi$, 
$S_T^\phi$ is unsatisfiable if and only if $\square\in \mi{clo}^{\mc H}(S_T^\phi)$;
$T\models \phi$ if and only if $\square\in \mi{clo}^{\mc H}(S_T^\phi)$.
\qed
\end{proof}

\section{Multi-step fuzzy inference via hyperresolution}
\label{S5}

In \cite{Guller2023c}, Section IV, we have shown some implementation of the Mamdani-Assilian type of fuzzy rules and inference in G\"{o}del logic.
To retain this section self-contained, we summarise the necessary material from \cite{Guller2023c} and bring further considerations.

Let $\mbb{U}$ be a non-empty set.
We call $\mbb{U}$ the universum.
A fuzzy set $A$ over $\mbb{U}$ is a mapping $A : \mbb{U}\longrightarrow [0,1]$.
We denote the set of all fuzzy sets over $\mbb{U}$ as ${\mc F}_\mbb{U}$.
Let $c\in [0,1]$ and $A_1, A_2\in {\mc F}_\mbb{U}$.
We define the height of $A_1$ as $\mi{height}(A_1)=\bigfvee_{u\in \mbb{U}} A_1(u)\in [0,1]$;
the cut $\mi{cut}(c,A_1)\in {\mc F}_\mbb{U}$ of $A_1$ by $c$ as 
$\mi{cut}(c,A_1) : \mbb{U}\longrightarrow [0,1],\ \mi{cut}(c,A_1)(u)=\mi{min}(c,A_1(u))$;
the union $A_1\cup A_2\in {\mc F}_\mbb{U}$ of $A_1$ and $A_2$ as 
$A_1\cup A_2 : \mbb{U}\longrightarrow [0,1],\ A_1\cup A_2(u)=\mi{max}(A_1(u),A_2(u))$;
the intersection $A_1\cap A_2\in {\mc F}_\mbb{U}$ of $A_1$ and $A_2$ as 
$A_1\cap A_2 : \mbb{U}\longrightarrow [0,1],\ A_1\cap A_2(u)=\mi{min}(A_1(u),A_2(u))$.
Let $\emptyset\neq \mbb{A}\subseteq_{\mc F} {\mc F}_\mbb{U}$.
Let $\mbb{X}$ be a non-empty finite set of fuzzy variables, having values from ${\mc F}_\mbb{U}$.
A fuzzy rule $r$ of the Mamdani-Assilian type is an expression of the form
$\mib{IF}\, X_0\, \mi{is}\, A_0\, \mi{and}\, \dots\, \mi{and}\, X_n\, \mi{is}\, A_n\, \mib{THEN}\, X\, \mi{is}\, A$,
$X_i, X\in \mbb{X}$, $A_i, A\in \mbb{A}$ \cite{MAAS75,Mam76}.
We say that $X_i$ are input fuzzy variables, whereas $X$ is the output fuzzy variable.
We denote $\mi{in}(r)=\{X_i \,|\, i\leq n\}\subseteq \mbb{X}$, $\mi{in}(r)\neq \emptyset$, and $\mi{out}(r)=X\in \mbb{X}$.
A fuzzy rule base is a non-empty finite set of fuzzy rules.
A fuzzy variable assignment is a mapping $\mbb{X}\longrightarrow {\mc F}_\mbb{U}$.
We denote the set of all fuzzy variable assignments as ${\mc S}_\mbb{U}$.
Let $e\in {\mc S}_\mbb{U}$. 
We define the value of $X$ with respect to $e$ and $r$ as 
$\|X\|_e^r=\mi{cut}(\bigfwedge_{i=0}^n \mi{height}(e(X_i)\cap A_i),A)\in {\mc F}_\mbb{U}$.
Let $B$ be a fuzzy rule base and $X\in \mbb{X}$.
We define the value of $X$ with respect to $e$ and $B$ as $\|X\|_e^B=\bigcup \{\|X\|_e^r \,|\, r\in B, \mi{out}(r)=X\}\in {\mc F}_\mbb{U}$.
Let ${\mc D}=e_0,\dots,e_\eta$, $e_\kappa\in {\mc S}_\mbb{U}$.
${\mc D}$ is a fuzzy derivation of $e_\eta$ from $e_0$ using $B$ iff, 
for all $1\leq \kappa\leq \eta$, $e_\kappa=\{(X,\|X\|_{e_{\kappa-1}}^B) \,|\, X\in \mbb{X}\}$.

Fuzzy rules can be translated to formulae of G\"{o}del logic on a reasonable assumption that the universum $\mbb{U}$ is countable.
At first, we represent natural and rational numbers as respective numerals.
We shall assume a fresh constant symbol $\tilde{z}$, a fresh unary function symbol $\tilde{s}$, and 
two fresh binary function symbols $\mi{frac}$, $\mi{-frac}$.
We denote $\tilde{\mbb{Z}}=\{\tilde{z},\tilde{s},\mi{frac},\mi{-frac}\}$.
Natural and rational numerals are defined as follows.
Let $t\in \mi{GTerm}_{\tilde{\mbb{Z}}}$.
$t$ is a natural numeral iff $t=\tilde{s}^n(\tilde{z})$.
$t$ is a rational numeral iff either $t=\mi{frac}(\tilde{s}^m(\tilde{z}),\tilde{s}^n(\tilde{z}))$, $n>0$, or
$t=\mi{-frac}(\tilde{s}^m(\tilde{z}),\tilde{s}^n(\tilde{z}))$, $m, n>0$.
We shall assume a set of four fresh unary predicate symbols $\tilde{\mbb{D}}=\{\mi{nat},\mi{rat},\mi{time},\mi{uni}\}$.
These symbols will be used for axiomatisation of certain domain properties of natural and rational numbers, of time and the universum $\mbb{U}$.
We shall assume a non-empty finite set of fresh unary predicate symbols 
$\tilde{\mbb{G}}=\{\tilde{G}_A \,|\, A\in \mbb{A}\}$ such that $\tilde{\mbb{G}}\cap \tilde{\mbb{D}}=\emptyset$.
These symbols will be used for axiomatisation of fuzzy sets appearing in fuzzy rules of the fuzzy rule base $B$.
We shall assume a non-empty finite set of fresh binary predicate symbols 
$\tilde{\mbb{H}}=\{\tilde{H}_X^r \,|\, r\in B, X\in \mbb{X}, \mi{out}(r)=X\}\cup \{\tilde{H}_X \,|\, X\in \mbb{X}\}$ such that 
$\tilde{\mbb{H}}\cap (\tilde{\mbb{D}}\cup \tilde{\mbb{G}})=\emptyset$.
These symbols will be used for axiomatisation of values of fuzzy variables (fuzzy sets) at individual time points.
$\tilde{H}_X^r$ are exploited for output fuzzy variables related to individual fuzzy rules of the fuzzy rule base $B$, 
while $\tilde{H}_X$ are for aggregated values with respect to the whole fuzzy rule base $B$.
We put ${\mc L}=\tilde{\mbb{Z}}\cup \tilde{\mbb{D}}\cup \tilde{\mbb{G}}\cup \tilde{\mbb{H}}$.
Let $e_0\in {\mc S}_\mbb{U}$ be the initial fuzzy variable assignment of a fuzzy derivation using the fuzzy rule base $B$.
We put $C_{\mc L}=\{0,1\}\cup \bigcup \{A[\mbb{U}] \,|\, A\in \mbb{A}\}\cup \bigcup \{e_0(X)[\mbb{U}] \,|\, X\in \mbb{X}\}$; 
$\{0,1\}\subseteq C_{\mc L}\subseteq [0,1]$ is countable.

In many practical cases, fuzzy sets appearing in the fuzzy rule base and as values for fuzzy variables have "reasonable" shapes,
which allows us to assume that the universum $\mbb{U}$, and subsequently, 
the set $C_{\mc L}$, the set of truth constants $\overline{C}_{\mc L}$ are even finite.  
For instance, assume that the universum $\mbb{U}$ is a closed interval $[u_b,u_e]$ of real numbers, $u_b<u_e$, $u_b, u_e\in \mbb{Q}$.
Our fuzzy inference system is in some sense finitary, that is $\mbb{A}$, $\mbb{X}$, $B$, $e_0$ are finite; 
only the universum $\mbb{U}$ and fuzzy sets in question are infinite.
Since the operators of union and intersection over fuzzy sets are idempotent, commutative, associative, 
only a finite number of fuzzy sets can be derived.
For every initial fuzzy variable assignment $e_0$, there exists $\eta^*$ such that
the fuzzy derivation ${\mc D}=e_0,\dots,e_\eta^*$, $e_\kappa\in {\mc S}_\mbb{U}$, contains all the derivable fuzzy variable assignments
from $e_0$ using the fuzzy rule base $B$; 
other inference steps after $e_\eta^*$ will derive only the fuzzy variable assignments already appearing in ${\mc D}$.
This gives rise to some finite approximation of the universum $\mbb{U}=[u_b,u_e]$.
Suppose that the original fuzzy sets appearing in $\mbb{A}$ and $e_0$ are semi-differentiable 
at every point of $(u_b,u_e)$, and right-differentiable at $u_b$, left-differentiable at $u_e$.
In addition, there exist $\mi{left}_\mi{min}, \mi{left}_\mi{max}, \mi{right}_\mi{min}, \mi{right}_\mi{max}\in \mbb{R}$,
$\mi{left}_\mi{min}<=\mi{left}_\mi{max}$, $\mi{right}_\mi{min}<=\mi{right}_\mi{max}$, such that 
$\mi{left}_\mi{min}$ and $\mi{left}_\mi{max}$ are lower and upper bounds, respectively, on every left derivative, 
and analogously for right derivatives.
It is straightforward to see that this assumption is preserved for all the derived fuzzy sets using the fuzzy rule base $B$ as well. 
If two fuzzy sets $A$ and $B$ are different, then they may have different values in the endpoint $u_b$ or $u_e$, or
there may exist some point $u\in (u_b,u_e)$ where they have different values, and moreover,
there exists a non-empty open subinterval $u\in (u-\delta,u+\delta)\subseteq (u_b,u_e)$, $0<\delta\in \mbb{Q}$, such that
$A$ and $B$ have different values at every point of $(u-\delta,u+\delta)$.
Recall that we have only a finite number of the fuzzy sets (original/derived).
Then the universum $\mbb{U}=[u_b,u_e]$ can be replaced with a finite set of witness points $\mbb{U}^*=\{u_b<u_1<\cdots<u_\lambda<u_e\}$
which splits the closed interval $[u_b,u_e]$ into closed subintervals (determined by two subsequent points) 
of the same length (an equilength split) so that
if two fuzzy sets $A$ and $B$ are different, then there exists $u^*\in \mbb{U}^*$ witnessing $A(u^*)\neq B(u^*)$.
Such a finite approximation $\mbb{U}^*$ is sufficient for a broad class of problems, 
covering the reachability, stability, and the existence of a $k$-cycle problems from \cite{Guller2023c}.
A lower bound on the equal length of closed subintervals of the split will be a subject of further research.

Let $\tilde{\mbb{P}}^*\subseteq \tilde{\mbb{P}}$ and $S\subseteq \mi{OrdCl}_{{\mc L}\cup \tilde{\mbb{P}}^*}$.
Notice that ${\mc L}\cup \tilde{\mbb{P}}^*$ contains the constant symbol $\tilde{z}\in \tilde{\mbb{Z}}$, and 
$\overline{C}_{\mc L}$ is assumed to be finite.
From the proof of Theorem \ref{T3} for ${\mc L}\cup \tilde{\mbb{P}}^*$, 
it follows that if $\square\not\in \mi{clo}^{\mc H}(S)$,
we can construct a model ${\mf A}$ of $S$ for ${\mc L}\cup \tilde{\mbb{P}}^*\cup \tilde{\mbb{W}}^*$
with the universum ${\mc U}_{\mf A}=\mi{GTerm}_{\tilde{\mbb{Z}}\cup \tilde{\mbb{W}}^*}\neq \emptyset$, $\tilde{z}\in {\mc U}_{\mf A}$, 
for some $\tilde{\mbb{W}}^*\subseteq \tilde{\mbb{W}}$,
where the function symbols from $\tilde{\mbb{Z}}\cup \tilde{\mbb{W}}^*$ are interpreted as constructors.
We denote 
${\mc K}=\{{\mc I} \,|\, {\mc I}\ \text{\it is an interpretation for}\ {\mc L}\cup \tilde{\mbb{P}}^*\cup \tilde{\mbb{W}}^*,
                         \tilde{\mbb{P}}^*\subseteq \tilde{\mbb{P}}, \tilde{\mbb{W}}^*\subseteq \tilde{\mbb{W}},
                         \mbox{${\mc U}_{\mc I}=\mi{GTerm}_{\tilde{\mbb{Z}}\cup \tilde{\mbb{W}}^*}$},\
                         \text{\it the function symbols from}\ \mbox{$\tilde{\mbb{Z}}\cup \tilde{\mbb{W}}^*$}                      \linebreak[4]  
                         \text{\it are interpreted as constructors}\}$. 
This yields that for purposes of the implementation, we may confine ourselves to interpretations from ${\mc K}$.
$S$ is satisfiable with respect to ${\mc K}$ iff there exists a model of $S$ from ${\mc K}$.
So, using Theorem \ref{T3}, 
$\square\in \mi{clo}^{\mc H}(S)$ if and only if $S$ is unsatisfiable if and only if $S$ is unsatisfiable with respect to ${\mc K}$.
Let $n_0\in \mbb{N}$, $\phi\in \mi{Form}_{\mc L}$, $T\subseteq \mi{Form}_{\mc L}$.
$\phi$ is a logical consequence of $T$ with respect to ${\mc K}$, in symbols $T\models_{\mc K} \phi$,
iff, for every interpretation ${\mc I}\in {\mc K}$, if ${\mc I}\models T$, then ${\mc I}\models \phi$.
We conclude by Theorem \ref{T1} that
there exist $J_T^\phi\subseteq \{(i,j) \,|\, i\geq n_0\}\subseteq \mbb{I}$,
$\tilde{\mbb{P}}^*=\{\tilde{p}_\mbbm{j} \,|\, \mbbm{j}\in J_T^\phi\}\subseteq \tilde{\mbb{P}}$,
$S_T^\phi\subseteq \mi{OrdCl}_{{\mc L}\cup \tilde{\mbb{P}}^*}$, and (i,ii) of Theorem \ref{T1} hold;
there exists an interpretation ${\mf A}$ for ${\mc L}$, and ${\mf A}\models T$, ${\mf A}\not\models \phi$, if and only if
there exists an interpretation ${\mf A}'$ for ${\mc L}\cup \tilde{\mbb{P}}^*$ and ${\mf A}'\models S_T^\phi$,
satisfying ${\mf A}={\mf A}'|_{\mc L}$;
using Theorem \ref{T3} for $S_T^\phi$, 
there exists an interpretation ${\mf A}$ for ${\mc L}$, and ${\mf A}\models T$, ${\mf A}\not\models \phi$, if and only if
there exists an interpretation ${\mf A}'\in {\mc K}$ for ${\mc L}\cup \tilde{\mbb{P}}^*\cup \tilde{\mbb{W}}^*$  
for some $\tilde{\mbb{W}}^*\subseteq \tilde{\mbb{W}}$ and ${\mf A}'\models S_T^\phi$;
$T\models \phi$ if and only if $S_T^\phi$ is unsatisfiable if and only if $S_T^\phi$ is unsatisfiable with respect to ${\mc K}$ if and only if
$T\models_{\mc K} \phi$.

The domains of natural, rational numbers, 
a domain of time (a kind of linear discrete time with the starting point $0$, represented as $\tilde{z}$, and without an endpoint), and 
the universum $\mbb{U}$ can be axiomatised as follows.
We define 
$T_D=\{\mi{nat}(\tilde{z}), \mi{nat}(\tilde{s}(x))\leftrightarrow \mi{nat}(x), 
       \mi{rat}(\mi{frac}(x,\tilde{s}(y)))\leftrightarrow \mi{nat}(x)\wedge \mi{nat}(y),
       \mi{rat}(\mi{-frac}(\tilde{s}(x),\tilde{s}(y)))\leftrightarrow \mi{nat}(x)\wedge \mi{nat}(y)\}\cup
       \{\mi{nat}(\mi{frac}(x,y))\geql \gz,\mi{nat}(\mi{-frac}(x,y))\geql \gz,\mi{rat}(\tilde{z})\geql \gz,\mi{rat}(\tilde{s}(x))\geql \gz\}\cup
       \{\mi{time}(x)\leftrightarrow \mi{nat}(x),\mi{uni}(x)\rightarrow \mi{rat}(x)\}\subseteq_{\mc F} \mi{Form}_{\mc L}$. 
The domains of natural and rational numbers are axiomatised by the first four formulae.
The domains of natural and rational numbers contain only respective natural and rational numerals in every ${\mc I}\in {\mc K}$, and no other
ground terms of ${\mc L}$, which is ensured by the fifth, sixth, seventh, and eighth axioms.
The domain of time is axiomatised as the domain of natural numbers (the ninth axiom), and 
the universum $\mbb{U}$ (finite/countable) as a subdomain of rational numbers (the tenth axiom).
Hence, there exists an injection $\gamma : \mbb{U}\longrightarrow \mbb{Q}$.
Let $t\in \mi{GTerm}_{\tilde{\mbb{Z}}}$ be a rational numeral.
If $t=\mi{frac}(\tilde{s}^m(\tilde{z}),\tilde{s}^n(\tilde{z}))$, $n>0$, 
then we define the value of $t$ as $\|t\|=\dfrac{m}{n}\in \mbb{Q}$.
If $t=\mi{-frac}(\tilde{s}^m(\tilde{z}),\tilde{s}^n(\tilde{z}))$, $m, n>0$,
then we define the value of $t$ as $\|t\|=-\dfrac{m}{n}\in \mbb{Q}$.
We axiomatise the universum $\mbb{U}$ as follows.
Let 
$\tilde{\mbb{U}}\subseteq \{t \,|\, t\in \mi{GTerm}_{\tilde{\mbb{Z}}}\ \text{\it is a rational numeral},\ \|t\|\in \gamma[\mbb{U}]\}$ such that 
$\{\|\tilde{u}\| \,|\, \tilde{u}\in \tilde{\mbb{U}}\}=\gamma[\mbb{U}]$.
We define 
$S_U^+=\{\mi{uni}(\tilde{u})\geql \gu \,|\, \tilde{u}\in \tilde{\mbb{U}}\}\subseteq \mi{OrdCl}_{\mc L}$,
$S_U^-=\{\mi{uni}(t)\geql \gz \,|\, t\in \mi{GTerm}_{\tilde{\mbb{Z}}}\ \text{\it is a rational numeral},\ t\not\in \tilde{\mbb{U}}\}\subseteq 
       \mi{OrdCl}_{\mc L}$,
$S_U=S_U^+\cup S_U^-\subseteq \mi{OrdCl}_{\mc L}$.
The clausal theory $S_U$ ensures that the universum $\mbb{U}$ is interpreted exactly as $\tilde{\mbb{U}}$ in every ${\mc I}\in {\mc K}$, 
i.e. it does not contain any other rational numerals or ground terms of ${\mc L}$ (the tenth axiom).
Let $\tilde{u}\in \tilde{\mbb{U}}$.
We denote $\langle\tilde{u}\rangle=\gamma^{-1}(\|\tilde{u}\|)\in \mbb{U}$; 
$\langle\tilde{u}\rangle$ denotes a unique element $u\in \mbb{U}$ such that $\gamma(u)=\|\tilde{u}\|\in \mbb{Q}$.
The translation of $\mbb{A}$ is defined as 
$S_\mbb{A}=\{\tilde{G}_A(\tilde{u})\geql \overline{A(\langle\tilde{u}\rangle)} \,|\, A\in \mbb{A}, \tilde{u}\in \tilde{\mbb{U}}\}\subseteq 
           \mi{OrdCl}_{\mc L}$. 
The translation of $e$ is defined as
$S_e=\{\tilde{H}_X(\tau,\tilde{u})\geql \overline{e(X)(\langle\tilde{u}\rangle)} \,|\, X\in \mbb{X}, \tilde{u}\in \tilde{\mbb{U}}\}\subseteq
     \mi{OrdCl}_{\mc L}$.
The translation of $r$ is defined as
$\phi_r(\tau,y)=\big(\mi{time}(\tau)\wedge \mi{uni}(y)\rightarrow
                     \big(\tilde{H}_X^r(\tilde{s}(\tau),y)\geql 
                          ((\bigwedge_{i=0}^n \exists x\, (\mi{uni}(x)\wedge \tilde{H}_{X_i}(\tau,x)\wedge \tilde{G}_{A_i}(x)))\wedge 
                           \tilde{G}_A(y))\big)\big)\in
                \mi{Form}_{\mc L}$.
The translation of $B$ is defined as 
$T_B=\{\phi_r(\tau,y) \,|\, r\in B\}\cup 
     \big\{\mi{time}(\tau)\wedge \mi{uni}(y)\rightarrow \big(\tilde{H}_X(\tilde{s}(\tau),y)\geql 
                                                             \bigvee_{r\in B, \mi{out}(r)=X} \tilde{H}_X^r(\tilde{s}(\tau),y)\big) \,|\, 
           X\in \mbb{X}\big\}\subseteq_{\mc F} 
     \mi{Form}_{\mc L}$.

Notice that $T_D$ and $T_B$ are finite theories.
In case of the universum $\mbb{U}$ being finite, $S_U^+$, $S_\mbb{A}$, $S_e$ are finite clausal theories, 
while $S_U^-$, $S_U$ are countably infinite clausal theories. 
If the universum $\mbb{U}$ is countably infinite, so are $S_U^+$, $S_U$, $S_\mbb{A}$, $S_e$; $S_U^-$ is countable and may even be finite.

\begin{lemma}[\mbox{\rm Lemma 6, Section IV, \cite{Guller2023c}}]
\label{le2}
Let ${\mc D}=e_0,\dots,e_\eta$ be a fuzzy derivation.
$T_D\cup S_U\cup S_\mbb{A}\cup T_B\cup S_{e_0}(\tau/\tilde{z})\models_{\mc K} S_{e_\eta}(\tau/\tilde{s}^\eta(\tilde{z}))$.
\end{lemma}

\begin{proof}
By straightforward induction on $\eta$.
Cf. Lemma 6, \cite{Guller2023c}.
\qed
\end{proof}

In Table \ref{tab44}, we introduce some admissible general order rules, which are useful to get shorter and more readable derivations, 
but superfluous for the refutational completeness argument.
\begin{table}[t]
\caption{Admissible order rules}\label{tab44}
\vspace{-6mm}
\centering
\begin{minipage}[t]{\linewidth}
\footnotesize
\begin{IEEEeqnarray*}{L}
\hline \hline \\[-6mm]
\end{IEEEeqnarray*}
\begin{alignat}{1}
\ctag{ceq4hr11}{({\it Order $\gz$-dichotomy rule})} \\[1mm]
\notag
& \dfrac{a\in \mi{atoms}(\mi{Vrnt}(S_{\kappa-1}))\cup \mi{qatoms}(\mi{Vrnt}(S_{\kappa-1}))}
        {a\geql \gz\vee \gz\gle a\in S_\kappa}. 
\end{alignat}
$a\geql \gz\vee \gz\gle a$ is an order $\gz$-dichotomy resolvent of $a$.
\begin{alignat}{1}
\ctag{ceq4hr10}{({\it Order $\gu$-dichotomy rule})} \\[1mm]
\notag
& \dfrac{a\in \mi{atoms}(\mi{Vrnt}(S_{\kappa-1}))\cup \mi{qatoms}(\mi{Vrnt}(S_{\kappa-1}))}
        {a\gle \gu\vee a\geql \gu\in S_\kappa}.                                 
\end{alignat}
$a\gle \gu\vee a\geql \gu$ is an order $\gu$-dichotomy resolvent of $a$.
\begin{alignat}{1}
\ctag{ceq4hr12}{({\it Order $\wedge$-commutativity rule})} \\[1mm]
\notag
& \dfrac{\begin{array}{l}
         \tilde{p}_{\mbbm{i}_1}(\bar{x})\gle \tilde{p}_{\mbbm{i}_2}(\bar{x})\vee 
         \tilde{p}_{\mbbm{i}_1}(\bar{x})\geql \tilde{p}_{\mbbm{i}_2}(\bar{x})\vee 
         \tilde{p}_\mbbm{i}(\bar{x})\geql \tilde{p}_{\mbbm{i}_2}(\bar{x}), \\
         \tilde{p}_{\mbbm{i}_2}(\bar{x})\gle \tilde{p}_{\mbbm{i}_1}(\bar{x})\vee 
         \tilde{p}_\mbbm{i}(\bar{x})\geql \tilde{p}_{\mbbm{i}_1}(\bar{x})\in \mi{Vrnt}(S_{\kappa-1})               
         \end{array}} 
        {\tilde{p}_{\mbbm{i}_1}(\bar{x})\gle \tilde{p}_{\mbbm{i}_2}(\bar{x})\vee 
         \tilde{p}_\mbbm{i}(\bar{x})\geql \tilde{p}_{\mbbm{i}_2}(\bar{x})\in S_\kappa}; \\[1mm]
\notag
& \tilde{p}_\mbbm{i}, \tilde{p}_{\mbbm{i}_1}, \tilde{p}_{\mbbm{i}_2}\in \tilde{\mbb{P}},\
  \bar{x}\ \text{\it is a sequence of variables of}\ {\mc L}.
\end{alignat}
$\tilde{p}_{\mbbm{i}_1}(\bar{x})\gle \tilde{p}_{\mbbm{i}_2}(\bar{x})\vee 
 \tilde{p}_\mbbm{i}(\bar{x})\geql \tilde{p}_{\mbbm{i}_2}(\bar{x})$ is an order $\wedge$-commutativity resolvent of  
$\tilde{p}_{\mbbm{i}_1}(\bar{x})\gle \tilde{p}_{\mbbm{i}_2}(\bar{x})\vee 
 \tilde{p}_{\mbbm{i}_1}(\bar{x})\geql \tilde{p}_{\mbbm{i}_2}(\bar{x})\vee 
 \tilde{p}_\mbbm{i}(\bar{x})\geql \tilde{p}_{\mbbm{i}_2}(\bar{x})$ and
$\tilde{p}_{\mbbm{i}_2}(\bar{x})\gle \tilde{p}_{\mbbm{i}_1}(\bar{x})\vee
 \tilde{p}_\mbbm{i}(\bar{x})\geql \tilde{p}_{\mbbm{i}_1}(\bar{x})$.
\begin{alignat}{1}
\ctag{ceq4hr14}{({\it Order $\tilde{\mbb{U}}$-bounded $\forall$-quantification rule})} \\[1mm]
\notag
& \dfrac{\forall \tilde{u}\in \tilde{\mbb{U}}\ C(x/\tilde{u})\in \mi{Vrnt}(S_{\kappa-1})}
        {\mi{uni}(x)\geql \gz\vee C(x)\in S_\kappa}; \\[1mm]
\notag
& C(x)\in \mi{OrdCl}_{{\mc L}_{\kappa-1}}.
\end{alignat}
$\mi{uni}(x)\geql \gz\vee C(x)$ is an order $\tilde{\mbb{U}}$-bounded $\forall$-quantification resolvent 
of $C(x/\tilde{u})$, $\tilde{u}\in \tilde{\mbb{U}}$.
\begin{IEEEeqnarray*}{L}
\hline \hline \\[2mm]
\end{IEEEeqnarray*}
\end{minipage}
\vspace{-6mm}
\end{table}

\section{An example}
\label{S6}

In \cite{Guller2023c}, at the end of Section IV, 
we have illustrated the implementation of the Mamdani-Assilian type of fuzzy rules and inference in G\"{o}del logic 
by an example on a fuzzy inference system modelling a simple thermodynamic system.
We have modelled an engine with inner combustion and cooling medium.
We have considered three physical quantities: the temperature (t), density (d) of the cooling medium, and the rotation (r) of the engine
together with their first derivatives.
For simplicity, the universum $\mbb{U}=\{0,1,2,3,4\}$ has been finite.
For every physical quantity $k$, we have defined fuzzy sets $\mi{low}_k$, $\mi{medium}_k$, $\mi{high}_k$, and
for its derivative $\dot{k}$, fuzzy sets $\mi{negative}_{\dot{k}}$, $\mi{zero}_{\dot{k}}$, $\mi{positive}_{\dot{k}}$, Table~\ref{tab10}.
The set of all fuzzy sets appearing in the underlying fuzzy rule base $B$, Table IV, \cite{Guller2023c}, has been defined as 
$\mbb{A}=\bigcup_{k\in \{t,d,r\}} \{\mi{low}_k,\mi{medium}_k,\mi{high}_k,\mi{negative}_{\dot{k}},\mi{zero}_{\dot{k}},              \linebreak[4]
                                    \mi{positive}_{\dot{k}}\}$.
The corresponding clausal theory $S_\mbb{A}$ is given in Table~\ref{tab15}, Appendix, $[47k-76k]$, $k\in \{t,d,r\}$.
The set of all fuzzy variables appearing in the fuzzy rule base $B$ has been defined as
$\mbb{X}=\{X_i \,|\, i\leq 5\}$ where the variables $X_0$ and $X_3$ correspond to the temperature and its derivative,
                                      $X_1$ and $X_4$ to the density and its derivative,
                                      $X_2$ and $X_5$ to the rotation and its derivative, respectively.
The universum $\mbb{U}$ has been represented as
$\tilde{\mbb{U}}=\{\mi{frac}(\tilde{z},\tilde{s}(\tilde{z})),
                   \mi{frac}(\tilde{s}(\tilde{z}),\tilde{s}(\tilde{z})),
                   \mi{frac}(\tilde{s}^2(\tilde{z}),\tilde{s}(\tilde{z})),                                                         \linebreak[4]                                                 
                   \mi{frac}(\tilde{s}^3(\tilde{z}),\tilde{s}(\tilde{z})),
                   \mi{frac}(\tilde{s}^4(\tilde{z}),\tilde{s}(\tilde{z}))\}$ 
(for simplicity, $\mi{frac}(\tilde{s}^n(\tilde{z}),\tilde{s}(\tilde{z}))$ is abbreviated as $\tilde{n}$), and
the corresponding clausal theory $S_U^+$ is given in Table \ref{tab14}, Appendix, $[42-46]$,
together with a clausal translation $S_D$ of the theory $T_D$, $[1-41]$.
\begin{table*}[t]
\vspace{-6mm}
\caption{Simplified fuzzy inference system}\label{tab10}
\vspace{-6mm}
\centering   
\begin{minipage}[t]{\linewidth-15mm}
\footnotesize
\begin{IEEEeqnarray*}{LLLL}
\hline \hline \\[1mm]
\IEEEeqnarraymulticol{4}{c}{\text{\bf \normalsize{Fuzzy sets}}} \\[2mm]
\text{\bf \small{Quantity:}}\quad & \mi{low}_k=\Big\{\frac{1}{0},\frac{0.5}{1},\frac{0}{2},\frac{0}{3},\frac{0}{4}\Big\} \quad
                                  & \mi{medium}_k=\Big\{\frac{0}{0},\frac{0.5}{1},\frac{1}{2},\frac{0.5}{3},\frac{0}{4}\Big\} \quad
                                  & \mi{high}_k=\Big\{\frac{0}{0},\frac{0}{1},\frac{0}{2},\frac{0.5}{3},\frac{1}{4}\Big\} \\[1mm]
\IEEEeqnarraymulticol{4}{l}{\text{$k$ stands for $t$ -- temperature, $d$ -- density, $r$ -- rotation}} \\[2mm]
\text{\bf \small{Derivative:}}\quad & \mi{negative}_{\dot{k}}=\Big\{\frac{1}{0},\frac{0.5}{1},\frac{0}{2},\frac{0}{3},\frac{0}{4}\Big\} \quad
                                    & \mi{zero}_{\dot{k}}=\Big\{\frac{0}{0},\frac{0.5}{1},\frac{1}{2},\frac{0.5}{3},\frac{0}{4}\Big\} \quad
                                    & \mi{positive}_{\dot{k}}=\Big\{\frac{0}{0},\frac{0}{1},\frac{0}{2},\frac{0.5}{3},\frac{1}{4}\Big\} \\[1mm]
\IEEEeqnarraymulticol{4}{l}{\text{$\dot{k}$ stands for $\dot{t}$ -- derivative of temperature,
                                                       $\dot{d}$ -- derivative of density,
                                                       $\dot{r}$ -- derivative of rotation}} 
\end{IEEEeqnarray*}
\vspace{-2mm}
\begin{IEEEeqnarray*}{LL}
\IEEEeqnarraymulticol{2}{c}{\text{\bf \normalsize{Simplified fuzzy rule base $\mbi{B^*}$}}} \\[2mm]
R_1:    & 
\mib{IF}\, X_0\, \mi{is}\, \mi{low}_t\, \mib{THEN}\, X_1\, \mi{is}\, \mi{high}_d \\[0mm]
R_6:    &
\mib{IF}\, X_2\, \mi{is}\, \mi{high}_r\, \mib{THEN}\, X_3\, \mi{is}\, \mi{positive}_{\dot{t}} \\[0mm]
R_8:    & 
\mib{IF}\, X_1\, \mi{is}\, \mi{high}_d\, \mi{and}\, X_2\, \mi{is}\, \mi{high}_r\, \mib{THEN}\, X_5\, \mi{is}\, \mi{negative}_{\dot{r}} \\[0mm]
R_{29}: &
\mib{IF}\, X_2\, \mi{is}\, \mi{medium}_r\, \mi{and}\, X_5\, \mi{is}\, \mi{positive}_{\dot{r}}\,
\mib{THEN}\, X_2\, \mi{is}\, \mi{high}_r 
\end{IEEEeqnarray*}
\vspace{-2mm}
%
%
\begin{IEEEeqnarray*}{LL}
\IEEEeqnarraymulticol{2}{c}{\text{\bf \normalsize{Translation of the simplified fuzzy rule base $\mbi{B^*}$}}} \\[2mm]
\phi_1(\tau,y) &= \big(\mi{time}(\tau)\wedge \mi{uni}(y)\rightarrow 
                       \big(\tilde{H}_{X_1}(\tilde{s}(\tau),y)\geql
                            (\exists x\, (\mi{uni}(x)\wedge \tilde{H}_{X_0}(\tau,x)\wedge \tilde{G}_{\mi{low}_t}(x))\wedge \tilde{G}_{\mi{high}_d}(y))\big)\big) \\[0mm]
\phi_6(\tau,y) &= \big(\mi{time}(\tau)\wedge \mi{uni}(y)\rightarrow 
                       \big(\tilde{H}_{X_3}(\tilde{s}(\tau),y)\geql
                            (\exists x\, (\mi{uni}(x)\wedge \tilde{H}_{X_2}(\tau,x)\wedge \tilde{G}_{\mi{high}_r}(x))\wedge \tilde{G}_{\mi{positive}_{\dot{t}}}(y))\big)\big) \\[0mm]
\phi_8(\tau,y) &= \big(\mi{time}(\tau)\wedge \mi{uni}(y)\rightarrow \\[0mm]
&\phantom{\mbox{}=\big(}
                       \big(\tilde{H}_{X_5}(\tilde{s}(\tau),y)\geql
                            (\exists x\, (\mi{uni}(x)\wedge \tilde{H}_{X_1}(\tau,x)\wedge \tilde{G}_{\mi{high}_d}(x))\wedge
                             \exists x\, (\mi{uni}(x)\wedge \tilde{H}_{X_2}(\tau,x)\wedge \tilde{G}_{\mi{high}_r}(x))\wedge
                             \tilde{G}_{\mi{negative}_{\dot{r}}}(y))\big)\big) \\[0mm]
\phi_{29}(\tau,y) &= \big(\mi{time}(\tau)\wedge \mi{uni}(y)\rightarrow \\[0mm]
&\phantom{\mbox{}=\big(}
                          \big(\tilde{H}_{X_2}(\tilde{s}(\tau),y)\geql
                               (\exists x\, (\tilde{H}_{X_2}(\tau,x)\wedge \tilde{G}_{\mi{medium}_r}(x))\wedge
                                \exists x\, (\tilde{H}_{X_5}(\tau,x)\wedge \tilde{G}_{\mi{positive}_{\dot{r}}}(x))\wedge
                                \tilde{G}_{\mi{high}_r}(y))\big)\big) 
\end{IEEEeqnarray*}
\vspace{-2mm}
%
%
\begin{IEEEeqnarray*}{LL}
\IEEEeqnarraymulticol{2}{c}{\text{\bf \normalsize{Translation of the reachability problem}}} \\[2mm]
\phi_r &= \exists \tau\, (\mi{time}(\tau)\wedge
                          \forall x\, (\mi{uni}(x)\rightarrow \tilde{H}_{X_3}(\tau,x)\geql \tilde{G}_{\mi{positive}_{\dot{t}}}(x))\wedge
                          \forall x\, (\mi{uni}(x)\rightarrow \tilde{H}_{X_5}(\tau,x)\geql \tilde{G}_{\mi{negative}_{\dot{r}}}(x))) \\[1mm]
\hline \hline 
\end{IEEEeqnarray*}
\end{minipage}  
\vspace{-2mm}
\end{table*}
As a continuation of this example, we show how can be solved a reachability problem, \cite{Guller2023c}, Section~IV, using hyperresolution.
To retain a reasonable length of this illustration, we shall consider a simplified fuzzy rule base $B^*$ containing only the four fuzzy rules:
$R_1$, $R_6$, $R_8$, $R_{29}$ of the original fuzzy rule base $B$, Table \ref{tab10}.
The initial state of a fuzzy derivation (at the starting time point $0$, represented as $\tilde{z}$) is determined as follows:
the temperature is low; the density is high; the rotation is low;
the first derivative of the temperature is zero;
the first derivative of the density is zero;
the first derivative of the rotation is positive;
hence, the initial variable assignment 
$e_0=\{(X_0,\mi{low}_t),(X_1,\mi{high}_d),(X_2,\mi{low}_r),(X_3,\mi{zero}_{\dot{t}}),                                              \linebreak[4]    
       (X_4,\mi{zero}_{\dot{d}}),(X_5,\mi{positive}_{\dot{r}})\}$, and 
the part of the corresponding clausal theory $S_{e_0}(\tau/\tilde{z})$ 
for the fuzzy variables $X_0$, $X_2$, $X_5$ is given in Table \ref{tab18}, Appendix, $[181-195]$.
Let ${\mc D}=e_0,\dots,e_\eta$ be a fuzzy derivation of $e_\eta$ from $e_0$ using $B^*$.
We show how can be solved a reachability problem of the form: 
there exists $\kappa\leq \eta$ such that $e_\kappa(X_3)=\mi{positive}_{\dot{t}}$ and $e_\kappa(X_5)=\mi{negative}_{\dot{r}}$.
In Table \ref{tab10}, the fuzzy rule base $B^*$ and the reachability problem are translated to formulae of G\"{o}del logic:
$T_{B^*}=\{\phi_1(\tau,y),\phi_6(\tau,y),\phi_8(\tau,y),\phi_{29}(\tau,y)\}$ and $\phi_r$, respectively.
Subsequently, the reachability problem formula $\phi_r$ can be translated to a clausal theory $S_{\phi_r}$, 
Tables \ref{tab11}, \ref{tab12}, \ref{tab18}, Appendix, $[196-218]$.
A clausal theory $S_{B^*}$ is given in Tables \ref{tab16}, \ref{tab17}, \ref{tab17b}, Appendix, $[77-180]$;
we have reused and simply cloned the clausal translations of $\phi_1(\tau,y)$ (one input fuzzy variable, Table~X, \cite{Guller2023c}) and 
$\phi_7(\tau,y)$ (two input fuzzy variables, Table~XIV, \cite{Guller2023c}). 
So, by Theorem \ref{T1} (for $n_0=0$), the deduction problem $T_D\cup S_U\cup S_\mbb{A}\cup T_{B^*}\cup S_{e_0}(\tau/\tilde{z})\models \phi_r$
has been reduced to an unsatisfiability problem for the clausal theory 
$S_D\cup S_U\cup S_\mbb{A}\cup S_{B^*}\cup S_{e_0}(\tau/\tilde{z})\cup S_{\phi_r}$, $[1-218]$, Tables \ref{tab14}--\ref{tab18}, Appendix,
which can be solved using hyperresolution, Theorem \ref{T3}.
There exists a refutation of $S_D\cup S_U\cup S_\mbb{A}\cup S_{B^*}\cup S_{e_0}(\tau/\tilde{z})\cup S_{\phi_r}$. 
In Table~\ref{tab10b}, we outline a digest of this refutation 
(the full refutation is given in Tables \ref{tab19}--\ref{tab24}, Appendix, $[219-389]$).
\begin{table}[t]
\vspace{-6mm}
\caption{A digest of the refutation of $S_D\cup S_U\cup S_\mbb{A}\cup S_B\cup S_{e_0}(\tau/\tilde{z})\cup S_{\phi_r}$}\label{tab10b}
\vspace{-6mm}
\centering   
\begin{minipage}[t]{\linewidth}
\scriptsize
\begin{IEEEeqnarray*}{LR}
\hline \hline \\[2mm]
\mi{nat}(\tilde{z})\geql \gu 
& [220] \\
\mi{nat}(\tilde{s}(\tilde{z}))\geql \gu
& [223] \\
\mi{nat}(\tilde{s}(\tilde{s}(\tilde{z})))\geql \gu
& [226] \\
\mi{time}(\tilde{z})\geql \gu
& [229] \\
\mi{time}(\tilde{s}(\tilde{z}))\geql \gu
& [232] \\
\mi{time}(\tilde{s}(\tilde{s}(\tilde{z})))\geql \gu
& [235] \\
\tilde{H}_{X_1}(\tilde{s}(\tilde{z}),\tilde{u})\geql \tilde{G}_{\mi{high}_d}(\tilde{u}), \tilde{u}\in \tilde{\mbb{U}}
& [258-262] \\
\tilde{H}_{X_2}(\tilde{s}(\tilde{z}),\tilde{u})\geql \tilde{G}_{\mi{high}_r}(\tilde{u}), \tilde{u}\in \tilde{\mbb{U}}
& [291-295] \\
\tilde{H}_{X_3}(\tilde{s}(\tilde{s}(\tilde{z})),\tilde{u})\geql \tilde{G}_{\mi{positive}_{\dot{t}}}(\tilde{u}), \tilde{u}\in \tilde{\mbb{U}}
& [318-322] \\
\tilde{H}_{X_5}(\tilde{s}(\tilde{s}(\tilde{z})),\tilde{u})\geql \tilde{G}_{\mi{negative}_{\dot{r}}}(\tilde{u}), \tilde{u}\in \tilde{\mbb{U}}
& [351-355] \\
\framebox{$\forall x\, \tilde{p}_{0,6}(\tilde{s}(\tilde{s}(\tilde{z})),x)\geql \gu$}
& [370] \\
\framebox{$\forall x\, \tilde{p}_{0,7}(\tilde{s}(\tilde{s}(\tilde{z})),x)\geql \gu$}
& [385] \\
\text{\bf Rule (\cref{ceq4hr5})} : [202] [211]; \tau/\tilde{s}(\tilde{s}(\tilde{z})) : [385] :
& \\
\framebox{$\tilde{p}_{0,3}(\tilde{s}(\tilde{s}(\tilde{z})),x)\geql \tilde{p}_{0,4}(\tilde{s}(\tilde{s}(\tilde{z})),x)$}
& [386] \\
\text{\bf Rule (\cref{ceq4hr5})} : [199] [203]; \tau/\tilde{s}(\tilde{s}(\tilde{z})) : [370] [386] :
& \\
\framebox{$\tilde{p}_{0,1}(\tilde{s}(\tilde{s}(\tilde{z})),x)\geql \tilde{p}_{0,2}(\tilde{s}(\tilde{s}(\tilde{z})),x)$}
& [387] \\
\text{\bf Rule (\cref{ceq4hr7})} : \exists \tau\, \tilde{p}_{0,1}(\tau,x) :
& \\
\framebox{$\tilde{p}_{0,1}(\tau,x)\gle \exists \tau\, \tilde{p}_{0,1}(\tau,x)\vee
           \tilde{p}_{0,1}(\tau,x)\geql \exists \tau\, \tilde{p}_{0,1}(\tau,x)$} 
& [388] \\  
\text{repeatedly \bf Rule (\cref{ceq4hr5})} : [196] [197] [235] [387] : [200] [388]; \tau/\tilde{s}(\tilde{s}(\tilde{z})) : \hspace{-1mm}
& \\
\square
& [389] \\[1mm]
\hline \hline  
\end{IEEEeqnarray*}
\end{minipage}
\vspace{-2mm} 
\end{table}  
We conclude by Corollary \ref{cor2} that the reachability problem has been solved in an affirmative way:
there exists $\kappa=2$ (represented as $\tilde{s}(\tilde{s}(\tilde{z}))$)
satisfying $e_\kappa(X_3)=\mi{positive}_{\dot{t}}$, $[318-322]$, Table~\ref{tab22}, Appendix, and 
$e_\kappa(X_5)=\mi{negative}_{\dot{r}}$, $[351-355]$, Table~\ref{tab23}, Appendix.%
\footnote{We have devised a rule-based system (download link: www.dai.fmph.uniba.sk/$\sim$guller/tfs18A.clp)
          for simulation of fuzzy inference using the fuzzy rule base $B$
          in the language (IDE) CLIPS \cite{GIRI98}.}

\section{Conclusions}
\label{S7}

In \cite{Guller2018a}, we have started our research programme with the main aim to develop 
the logical and computational foundations of multi-step fuzzy inference from the perspective of many-valued logics and automated reasoning.
As a first step, we have extent the Davis-Putnam-Logemann-Loveland procedure ({\it DPLL}) to the propositional G\"{o}del logic and proved
its refutational soundness and finite completeness.
The {\it DPLL} procedure infer over finite order clausal theories, 
where an order clause is a finite set of order literals of the form $\varepsilon_1\diamond \varepsilon_2$;
$\varepsilon_i$ is a propositional atom or a truth constant $\gz$ or $\gu$, and $\diamond$ is a connective $\geql$ or $\gle$.
As a second step, we have generalised the clausal form and translation to the first-order case in \cite{Guller2023c}.
Moreover, we have added to the first-order G\"{o}del logic intermediate truth constants, which is necessary when dealing with fuzzy sets.
In this logical calculus, we have proposed implementation of the Mamdani-Assilian type of fuzzy rules and inference.
We have outlined a class of problems concerning general properties of multi-step fuzzy inference 
(with three fundamental instances: reachability, stability, and the existence of a $k$-cycle in fuzzy inference). 
This class of problems can be reduced to a class of deduction/unsatisfiability problems.
Finally, as a third step, we have modified the hyperresolution calculus from \cite{Guller2019b}, inferring over clausal theories,  
which remains refutation sound and complete on condition that the universum of fuzzy sets is finite, and 
the underlying clausal theory is countable.
If the fuzzy sets in question have "reasonable" shapes, such a finite approximation of the universum is feasible.
The implementation in all its steps has been illustrated with an example on a fuzzy inference system 
modelling an engine with inner combustion and cooling medium (introduced in \cite{Guller2023c}).

\bibliographystyle{IEEEtran}  
\bibliography{IEEEabrv,tfs19}

\newpage
\section{Appendix}
\label{S10}

\subsection{Substitutions}
\label{S10.1}

Let $X=\{x_i \,|\, 1\leq i\leq n\}\subseteq \mi{Var}_{\mc L}$.
A substitution $\vartheta$ of ${\mc L}$ is a mapping $\vartheta : X\longrightarrow \mi{Term}_{\mc L}$.
$\vartheta$ is commonly written in the form $x_1/\vartheta(x_1),\dots,x_n/\vartheta(x_n)$.
We denote $\mi{dom}(\vartheta)=X\subseteq_{\mc F} \mi{Var}_{\mc L}$ and 
$\mi{range}(\vartheta)=\bigcup_{x\in X} \mi{vars}(\vartheta(x))\subseteq_{\mc F} \mi{Var}_{\mc L}$.
The set of all substitutions of ${\mc L}$ is designated as $\mi{Subst}_{\mc L}$.
$\vartheta$ is a variable renaming iff $\vartheta : \mi{dom}(\vartheta)\longrightarrow \mi{Var}_{\mc L}$, 
for all $x, x'\in \mi{dom}(\vartheta)$ and $x\neq x'$, $\vartheta(x)\neq \vartheta(x')$.
We define $\mi{id}_{\mc L} : \mi{Var}_{\mc L}\longrightarrow \mi{Var}_{\mc L},\ \mi{id}_{\mc L}(x)=x$.
Let $t\in \mi{Term}_{\mc L}$ and $\vartheta'\in \mi{Subst}_{\mc L}$. 
$\vartheta$ is applicable to $t$ iff $\mi{dom}(\vartheta)\supseteq \mi{vars}(t)$.
Let $\vartheta$ be applicable to $t$.
We define the application $t\vartheta\in \mi{Term}_{\mc L}$ of $\vartheta$ to $t$ by recursion on the structure of $t$ as follows:
\begin{equation*}
t\vartheta=\left\{\begin{array}{ll}
                  \vartheta(t)                       &\ \text{\it if}\ t\in \mi{Var}_{\mc L}, \\[1mm]
                  f(t_1\vartheta,\dots,t_n\vartheta) &\ \text{\it if}\ t=f(t_1,\dots,t_n).
                  \end{array}
           \right.
\end{equation*}
Let $\mi{range}(\vartheta)\subseteq \mi{dom}(\vartheta')$.
We define the composition $\vartheta\circ \vartheta'\in \mi{Subst}_{\mc L}$ of $\vartheta$ and $\vartheta'$ 
as $\vartheta\circ \vartheta' : \mi{dom}(\vartheta)\longrightarrow \mi{Term}_{\mc L},\ 
    \vartheta\circ \vartheta'(x)=\vartheta(x)\vartheta'$, 
$\mi{dom}(\vartheta\circ \vartheta')=\mi{dom}(\vartheta)$,
$\mi{range}(\vartheta\circ \vartheta')=\mi{range}(\vartheta'|_{\mi{range}(\vartheta)})$.
Note that composition of substitutions is associative.
%
%
%
%
%
Let $a\in \mi{Atom}_{\mc L}$.
$\vartheta$ is applicable to $a$ iff $\mi{dom}(\vartheta)\supseteq \mi{vars}(a)$.
Let $\vartheta$ be applicable to $a$ and $a=p(t_1,\dots,t_n)$.
We define the application $a\vartheta\in \mi{Atom}_{\mc L}$ of $\vartheta$ to $a$ as $a\vartheta=p(t_1\vartheta,\dots,t_n\vartheta)$.
Let $c\in \overline{C}_{\mc L}$.
$\vartheta$ is trivially applicable to $c$. 
We define the application $c\vartheta\in \overline{C}_{\mc L}$ of $\vartheta$ to $c$ as $c\vartheta=c$. 
Let $Q x\, a\in \mi{QAtom}_{\mc L}$.
$\vartheta$ is applicable to $Q x\, a$ 
iff $\mi{dom}(\vartheta)\supseteq \mi{freevars}(Q x\, a)$ and $x\not\in \mi{range}(\vartheta|_{\mi{freevars}(Q x\, a)})$.
Let $\vartheta$ be applicable to $Q x\, a$.
We define the application $(Q x\, a)\vartheta\in \mi{QAtom}_{\mc L}$ of $\vartheta$ to $Q x\, a$ 
as $(Q x\, a)\vartheta=Q x\, a(\vartheta|_{\mi{freevars}(Q x\, a)}\cup x/x)$,
$\mi{vars}(a)=\mi{freevars}(Q x\, a)\cup \{x\}$,
$x\not\in \mi{freevars}(Q x\, a)$,
$\mi{dom}(\vartheta|_{\mi{freevars}(Q x\, a)}\cup x/x)=\mi{freevars}(Q x\, a)\cup \{x\}=\mi{vars}(a)$,
$\mi{vars}(a(\vartheta|_{\mi{freevars}(Q x\, a)}\cup x/x))=\mi{range}(\vartheta|_{\mi{freevars}(Q x\, a)})\cup \{x\}$,
$\mi{freevars}((Q x\, a)\vartheta)=\mi{range}(\vartheta|_{\mi{freevars}(Q x\, a)})$,
$\mi{boundvars}((Q x\, a)\vartheta)=\{x\}$.
Let $\varepsilon_1\diamond \varepsilon_2\in \mi{OrdLit}_{\mc L}$.
$\vartheta$ is applicable to $\varepsilon_1\diamond \varepsilon_2$ iff, for both $i$, $\vartheta$ is applicable to $\varepsilon_i$.
Let $\vartheta$ be applicable to $\varepsilon_1\diamond \varepsilon_2$.
We define the application $(\varepsilon_1\diamond \varepsilon_2)\vartheta\in \mi{OrdLit}_{\mc L}$ 
of $\vartheta$ to $\varepsilon_1\diamond \varepsilon_2$ 
as $(\varepsilon_1\diamond \varepsilon_2)\vartheta=\varepsilon_1\vartheta\diamond \varepsilon_2\vartheta$. 
Let $E\subseteq \mbb{E}$, 
either $\mbb{E}=\mi{Term}_{\mc L}$ or $\mbb{E}=\mi{Atom}_{\mc L}$ or $\mbb{E}=\overline{C}_{\mc L}$ or 
$\mbb{E}=\mi{QAtom}_{\mc L}$ or $\mbb{E}=\mi{OrdLit}_{\mc L}$ or $\mbb{E}=\mi{OrdCl}_{\mc L}$.
$\vartheta$ is applicable to $E$ iff, for all $\varepsilon\in E$, $\vartheta$ is applicable to $\varepsilon$.
Let $\vartheta$ be applicable to $E$.
We define the application $E\vartheta\subseteq \mbb{E}$ of $\vartheta$ to $E$ as $E\vartheta=\{\varepsilon\vartheta \,|\, \varepsilon\in E\}$.
Let $\varepsilon, \varepsilon'\in \mbb{E}$. 
$\varepsilon'$ is an instance of $\varepsilon$ iff there exists $\vartheta^*\in \mi{Subst}_{\mc L}$ such that
$\varepsilon'=\varepsilon\vartheta^*$.
$\varepsilon'$ is a variant of $\varepsilon$ iff there exists a variable renaming $\rho^*\in \mi{Subst}_{\mc L}$ such that 
$\varepsilon'=\varepsilon\rho^*$.
Let $C\in \mi{OrdCl}_{\mc L}$ and $S\subseteq \mi{OrdCl}_{\mc L}$.
$C$ is an instance, a variant of $S$ iff there exists $C^*\in S$ such that $C$ is an instance, a variant of $C^*$.
We denote $\mi{Inst}(S)=\{C \,|\, C\in \mi{OrdCl}_{\mc L}\ \text{\it is an instance of}\ S\}$ and
$\mi{Vrnt}(S)=\{C \,|\, C\in \mi{OrdCl}_{\mc L}\ \text{\it is a variant of}\ S\}$.

$\vartheta$ is a unifier for $E$ iff $E\vartheta$ is a singleton set.
Note that there does not exist a unifier for $\emptyset$.
Let $\theta\in \mi{Subst}_{\mc L}$. 
$\theta$ is a most general unifier for $E$ iff $\theta$ is a unifier for $E$, and
for every unifier $\vartheta$ for $E$, there exists $\gamma^*\in \mi{Subst}_{\mc L}$ such that 
$\vartheta|_{\mi{freevars}(E)}=\theta|_{\mi{freevars}(E)}\circ \gamma^*$.
By $\mi{mgu}(E)\subseteq \mi{Subst}_{\mc L}$ we denote the set of all most general unifiers for $E$.
Let $\overline{E}=E_0,\dots,E_n$, $E_i\subseteq \mbb{E}_i$,
either $\mbb{E}_i=\mi{Term}_{\mc L}$ or $\mbb{E}_i=\mi{Atom}_{\mc L}$ or $\mbb{E}_i=\overline{C}_{\mc L}$ or
$\mbb{E}_i=\mi{QAtom}_{\mc L}$ or $\mbb{E}_i=\mi{OrdLit}_{\mc L}$.
$\vartheta$ is applicable to $\overline{E}$ iff, for all $i\leq n$, $\vartheta$ is applicable to $E_i$.
Let $\vartheta$ be applicable to $\overline{E}$.
We define the application $\overline{E}\vartheta$ of $\vartheta$ to $\overline{E}$ as $\overline{E}\vartheta=E_0\vartheta,\dots,E_n\vartheta$, 
$E_i\vartheta\subseteq \mbb{E}_i$.
Note that if $\vartheta$ is applicable to $\overline{E}$, then $\mi{dom}(\vartheta)\supseteq \mi{freevars}(\overline{E})$.
$\vartheta$ is a unifier for $\overline{E}$ iff, for all $i\leq n$, $\vartheta$ is a unifier for $E_i$.
Note that if there exists $i^*\leq n$ and $E_{i^*}=\emptyset$, then there does not exist a unifier for $\overline{E}$.
$\theta$ is a most general unifier for $\overline{E}$ iff $\theta$ is a unifier for $\overline{E}$, and
for every unifier $\vartheta$ for $\overline{E}$, there exists $\gamma^*\in \mi{Subst}_{\mc L}$ such that
$\vartheta|_{\mi{freevars}(\overline{E})}=\theta|_{\mi{freevars}(\overline{E})}\circ \gamma^*$.
By $\mi{mgu}(\overline{E})\subseteq \mi{Subst}_{\mc L}$ we denote the set of all most general unifiers for $\overline{E}$.

\begin{theorem}[Unification Theorem]
\label{T0}
Let $\overline{E}=E_0,\dots,E_n$, $E_i\subseteq_{\mc F} \mbb{E}_i$,
either $\mbb{E}_i=\mi{Term}_{\mc L}$ or $\mbb{E}_i=\mi{Atom}_{\mc L}$ or $\mbb{E}_i=\overline{C}_{\mc L}$.
If there exists a unifier for $\overline{E}$, then there exists $\theta^*\in \mi{mgu}(\overline{E})$ such that
$\mi{range}(\theta^*|_{\mi{vars}(\overline{E})})\subseteq \mi{vars}(\overline{E})$.
\end{theorem}

\begin{proof} 
By induction on $\|\mi{vars}(\overline{E})\|$;
a modification of the proof of Theorem 2.3 (Unification Theorem) in \cite{Apt86}.
\qed
\end{proof}

\begin{theorem}[Extended Unification Theorem]
\label{T00}
Let $\overline{E}=E_0,\dots,E_n$, $E_i\subseteq_{\mc F} \mbb{E}_i$,
either $\mbb{E}_i=\mi{Term}_{\mc L}$ or $\mbb{E}_i=\mi{Atom}_{\mc L}$ or $\mbb{E}_i=\overline{C}_{\mc L}$ or
$\mbb{E}_i=\mi{QAtom}_{\mc L}$ or $\mbb{E}_i=\mi{OrdLit}_{\mc L}$, and 
$\mi{boundvars}(\overline{E})\subseteq V\subseteq_{\mc F} \mi{Var}_{\mc L}$.
If there exists a unifier for $\overline{E}$, then there exists $\theta^*\in \mi{mgu}(\overline{E})$ such that
$\mi{range}(\theta^*|_{\mi{freevars}(\overline{E})})\cap V=\emptyset$.
\end{theorem}

\begin{proof}
A straightforward generalisation of Theorem \ref{T0}.
\qed
\end{proof}

\begin{table*}[p]
\vspace{-6mm}
\caption{Translation of the formula $\phi_r$ to clausal form}\label{tab11}
\vspace{-6mm}
\centering
\begin{minipage}[t]{\linewidth-10mm}
\footnotesize
\begin{IEEEeqnarray*}{LR}
\hline \hline \\[2mm]
\Big\{
  \tilde{p}_{0,0}(\tau,x)\gle \gu,
& \\
\phantom{\Big\{}
  \tilde{p}_{0,0}(\tau,x)\leftrightarrow
  \exists \tau\, (\underbrace{\mi{time}(\tau)\wedge
                              \forall x\, (\mi{uni}(x)\rightarrow \tilde{H}_{X_3}(\tau,x)\geql \tilde{G}_{\mi{positive}_{\dot{t}}}(x))\wedge
                              \forall x\, (\mi{uni}(x)\rightarrow \tilde{H}_{X_5}(\tau,x)\geql \tilde{G}_{\mi{negative}_{\dot{r}}}(x))}_{\tilde{p}_{0,1}(\tau,x)})\Big\}
& \quad (\ref{eq0rr6+}) \\
\Big\{
  \tilde{p}_{0,0}(\tau,x)\gle \gu,
  \tilde{p}_{0,0}(\tau,x)\geql \exists \tau\, \tilde{p}_{0,1}(\tau,x),
& \\
\phantom{\Big\{}  
  \tilde{p}_{0,1}(\tau,x)\leftrightarrow \big(\underbrace{\mi{time}(\tau)}_{\tilde{p}_{0,2}(\tau,x)}\wedge
                                              \big(\underbrace{\forall x\, (\mi{uni}(x)\rightarrow \tilde{H}_{X_3}(\tau,x)\geql \tilde{G}_{\mi{positive}_{\dot{t}}}(x))\wedge
                                                               \forall x\, (\mi{uni}(x)\rightarrow \tilde{H}_{X_5}(\tau,x)\geql \tilde{G}_{\mi{negative}_{\dot{r}}}(x))}_{\tilde{p}_{0,3}(\tau,x)}\big)\big)\Big\}
& \quad (\ref{eq0rr1+}) \\
\Big\{
  \tilde{p}_{0,0}(\tau,x)\gle \gu,
  \tilde{p}_{0,0}(\tau,x)\geql \exists \tau\, \tilde{p}_{0,1}(\tau,x),
& \\
\phantom{\Big\{}
  \tilde{p}_{0,2}(\tau,x)\gle \tilde{p}_{0,3}(\tau,x)\vee \tilde{p}_{0,2}(\tau,x)\geql \tilde{p}_{0,3}(\tau,x)\vee \tilde{p}_{0,1}(\tau,x)\geql \tilde{p}_{0,3}(\tau,x),
  \tilde{p}_{0,3}(\tau,x)\gle \tilde{p}_{0,2}(\tau,x)\vee \tilde{p}_{0,1}(\tau,x)\geql \tilde{p}_{0,2}(\tau,x),
& \\
\phantom{\Big\{}
  \tilde{p}_{0,2}(\tau,x)\geql \mi{time}(\tau),
  \tilde{p}_{0,3}(\tau,x)\leftrightarrow \big(\underbrace{\forall x\, (\mi{uni}(x)\rightarrow \tilde{H}_{X_3}(\tau,x)\geql \tilde{G}_{\mi{positive}_{\dot{t}}}(x))}_{\tilde{p}_{0,4}(\tau,x)}\wedge
                                              \underbrace{\forall x\, (\mi{uni}(x)\rightarrow \tilde{H}_{X_5}(\tau,x)\geql \tilde{G}_{\mi{negative}_{\dot{r}}}(x))}_{\tilde{p}_{0,5}(\tau,x)}\big)\Big\}                                              
& \quad (\ref{eq0rr1+}) \\
\Big\{
  \tilde{p}_{0,0}(\tau,x)\gle \gu,
  \tilde{p}_{0,0}(\tau,x)\geql \exists \tau\, \tilde{p}_{0,1}(\tau,x),
& \\
\phantom{\Big\{}
  \tilde{p}_{0,2}(\tau,x)\gle \tilde{p}_{0,3}(\tau,x)\vee \tilde{p}_{0,2}(\tau,x)\geql \tilde{p}_{0,3}(\tau,x)\vee \tilde{p}_{0,1}(\tau,x)\geql \tilde{p}_{0,3}(\tau,x),
  \tilde{p}_{0,3}(\tau,x)\gle \tilde{p}_{0,2}(\tau,x)\vee \tilde{p}_{0,1}(\tau,x)\geql \tilde{p}_{0,2}(\tau,x),
& \\
\phantom{\Big\{}
  \tilde{p}_{0,2}(\tau,x)\geql \mi{time}(\tau),
& \\
\phantom{\Big\{}
  \tilde{p}_{0,4}(\tau,x)\gle \tilde{p}_{0,5}(\tau,x)\vee \tilde{p}_{0,4}(\tau,x)\geql \tilde{p}_{0,5}(\tau,x)\vee \tilde{p}_{0,3}(\tau,x)\geql \tilde{p}_{0,5}(\tau,x),
  \tilde{p}_{0,5}(\tau,x)\gle \tilde{p}_{0,4}(\tau,x)\vee \tilde{p}_{0,3}(\tau,x)\geql \tilde{p}_{0,4}(\tau,x),
& \\
\phantom{\Big\{}
  \tilde{p}_{0,4}(\tau,x)\leftrightarrow \forall x\, (\underbrace{\mi{uni}(x)\rightarrow \tilde{H}_{X_3}(\tau,x)\geql \tilde{G}_{\mi{positive}_{\dot{t}}}(x)}_{\tilde{p}_{0,6}(\tau,x)}),
  \tilde{p}_{0,5}(\tau,x)\leftrightarrow \forall x\, (\underbrace{\mi{uni}(x)\rightarrow \tilde{H}_{X_5}(\tau,x)\geql \tilde{G}_{\mi{negative}_{\dot{r}}}(x)}_{\tilde{p}_{0,7}(\tau,x)})\Big\}
& \quad (\ref{eq0rr5+}) \\
\Big\{
  \tilde{p}_{0,0}(\tau,x)\gle \gu,
  \tilde{p}_{0,0}(\tau,x)\geql \exists \tau\, \tilde{p}_{0,1}(\tau,x),
& \\
\phantom{\Big\{}
  \tilde{p}_{0,2}(\tau,x)\gle \tilde{p}_{0,3}(\tau,x)\vee \tilde{p}_{0,2}(\tau,x)\geql \tilde{p}_{0,3}(\tau,x)\vee \tilde{p}_{0,1}(\tau,x)\geql \tilde{p}_{0,3}(\tau,x),
  \tilde{p}_{0,3}(\tau,x)\gle \tilde{p}_{0,2}(\tau,x)\vee \tilde{p}_{0,1}(\tau,x)\geql \tilde{p}_{0,2}(\tau,x),
& \\
\phantom{\Big\{}
  \tilde{p}_{0,2}(\tau,x)\geql \mi{time}(\tau),
& \\
\phantom{\Big\{}
  \tilde{p}_{0,4}(\tau,x)\gle \tilde{p}_{0,5}(\tau,x)\vee \tilde{p}_{0,4}(\tau,x)\geql \tilde{p}_{0,5}(\tau,x)\vee \tilde{p}_{0,3}(\tau,x)\geql \tilde{p}_{0,5}(\tau,x),
  \tilde{p}_{0,5}(\tau,x)\gle \tilde{p}_{0,4}(\tau,x)\vee \tilde{p}_{0,3}(\tau,x)\geql \tilde{p}_{0,4}(\tau,x),
& \\
\phantom{\Big\{}
  \tilde{p}_{0,4}(\tau,x)\geql \forall x\, \tilde{p}_{0,6}(\tau,x),
  \tilde{p}_{0,6}(\tau,x)\leftrightarrow \big(\underbrace{\mi{uni}(x)}_{\tilde{p}_{0,8}(\tau,x)}\rightarrow \underbrace{\tilde{H}_{X_3}(\tau,x)\geql \tilde{G}_{\mi{positive}_{\dot{t}}}(x)}_{\tilde{p}_{0,9}(\tau,x)}\big),
& \\
\phantom{\Big\{}
  \tilde{p}_{0,5}(\tau,x)\geql \forall x\, \tilde{p}_{0,7}(\tau,x),
  \tilde{p}_{0,7}(\tau,x)\leftrightarrow \big(\underbrace{\mi{uni}(x)}_{\tilde{p}_{0,10}(\tau,x)}\rightarrow \underbrace{\tilde{H}_{X_5}(\tau,x)\geql \tilde{G}_{\mi{negative}_{\dot{r}}}(x)}_{\tilde{p}_{0,11}(\tau,x)})\Big\}
& \quad (\ref{eq0rr3+}) \\
\Big\{
  \tilde{p}_{0,0}(\tau,x)\gle \gu,
  \tilde{p}_{0,0}(\tau,x)\geql \exists \tau\, \tilde{p}_{0,1}(\tau,x),
& \\
\phantom{\Big\{}
  \tilde{p}_{0,2}(\tau,x)\gle \tilde{p}_{0,3}(\tau,x)\vee \tilde{p}_{0,2}(\tau,x)\geql \tilde{p}_{0,3}(\tau,x)\vee \tilde{p}_{0,1}(\tau,x)\geql \tilde{p}_{0,3}(\tau,x),
  \tilde{p}_{0,3}(\tau,x)\gle \tilde{p}_{0,2}(\tau,x)\vee \tilde{p}_{0,1}(\tau,x)\geql \tilde{p}_{0,2}(\tau,x),
& \\
\phantom{\Big\{}
  \tilde{p}_{0,2}(\tau,x)\geql \mi{time}(\tau),
& \\
\phantom{\Big\{}
  \tilde{p}_{0,4}(\tau,x)\gle \tilde{p}_{0,5}(\tau,x)\vee \tilde{p}_{0,4}(\tau,x)\geql \tilde{p}_{0,5}(\tau,x)\vee \tilde{p}_{0,3}(\tau,x)\geql \tilde{p}_{0,5}(\tau,x),
  \tilde{p}_{0,5}(\tau,x)\gle \tilde{p}_{0,4}(\tau,x)\vee \tilde{p}_{0,3}(\tau,x)\geql \tilde{p}_{0,4}(\tau,x),
& \\
\phantom{\Big\{}
  \tilde{p}_{0,4}(\tau,x)\geql \forall x\, \tilde{p}_{0,6}(\tau,x),
& \\
\phantom{\Big\{}
  \tilde{p}_{0,8}(\tau,x)\gle \tilde{p}_{0,9}(\tau,x)\vee \tilde{p}_{0,8}(\tau,x)\geql \tilde{p}_{0,9}(\tau,x)\vee \tilde{p}_{0,6}(\tau,x)\geql \tilde{p}_{0,9}(\tau,x),  
  \tilde{p}_{0,9}(\tau,x)\gle \tilde{p}_{0,8}(\tau,x)\vee \tilde{p}_{0,6}(\tau,x)\geql \gu,
& \\
\phantom{\Big\{}
  \tilde{p}_{0,8}(\tau,x)\geql \mi{uni}(x),
  \tilde{p}_{0,9}(\tau,x)\leftrightarrow \underbrace{\tilde{H}_{X_3}(\tau,x)}_{\tilde{p}_{0,12}(\tau,x)}\geql \underbrace{\tilde{G}_{\mi{positive}_{\dot{t}}}(x)}_{\tilde{p}_{0,13}(\tau,x)},
& \\
\phantom{\Big\{}
  \tilde{p}_{0,5}(\tau,x)\geql \forall x\, \tilde{p}_{0,7}(\tau,x),
& \\
\phantom{\Big\{}
  \tilde{p}_{0,10}(\tau,x)\gle \tilde{p}_{0,11}(\tau,x)\vee \tilde{p}_{0,10}(\tau,x)\geql \tilde{p}_{0,11}(\tau,x)\vee \tilde{p}_{0,7}(\tau,x)\geql \tilde{p}_{0,11}(\tau,x),  
  \tilde{p}_{0,11}(\tau,x)\gle \tilde{p}_{0,10}(\tau,x)\vee \tilde{p}_{0,7}(\tau,x)\geql \gu,
& \\
\phantom{\Big\{}
  \tilde{p}_{0,10}(\tau,x)\geql \mi{uni}(x),
  \tilde{p}_{0,11}(\tau,x)\leftrightarrow \underbrace{\tilde{H}_{X_5}(\tau,x)}_{\tilde{p}_{0,14}(\tau,x)}\geql \underbrace{\tilde{G}_{\mi{negative}_{\dot{r}}}(x)}_{\tilde{p}_{0,15}(\tau,x)}\Big\}
& \quad (\ref{eq0rr7+}) \\[1mm]
\hline \hline
\end{IEEEeqnarray*}
\end{minipage}     
\vspace{-2mm}      
\end{table*}       
\begin{table*}[p]
\vspace{-6mm}
\caption{Translation of the formula $\phi_r$ to clausal form}\label{tab12}
\vspace{-6mm}
\centering
\begin{minipage}[t]{\linewidth-10mm}
\footnotesize
\begin{IEEEeqnarray*}{LR}
\hline \hline \\[2mm]
\Big\{
  \tilde{p}_{0,0}(\tau,x)\gle \gu,
  \tilde{p}_{0,0}(\tau,x)\geql \exists \tau\, \tilde{p}_{0,1}(\tau,x),
& \\
\phantom{\Big\{}
  \tilde{p}_{0,2}(\tau,x)\gle \tilde{p}_{0,3}(\tau,x)\vee \tilde{p}_{0,2}(\tau,x)\geql \tilde{p}_{0,3}(\tau,x)\vee \tilde{p}_{0,1}(\tau,x)\geql \tilde{p}_{0,3}(\tau,x),
  \tilde{p}_{0,3}(\tau,x)\gle \tilde{p}_{0,2}(\tau,x)\vee \tilde{p}_{0,1}(\tau,x)\geql \tilde{p}_{0,2}(\tau,x),
& \\
\phantom{\Big\{}
  \tilde{p}_{0,2}(\tau,x)\geql \mi{time}(\tau),
& \\
\phantom{\Big\{}
  \tilde{p}_{0,4}(\tau,x)\gle \tilde{p}_{0,5}(\tau,x)\vee \tilde{p}_{0,4}(\tau,x)\geql \tilde{p}_{0,5}(\tau,x)\vee \tilde{p}_{0,3}(\tau,x)\geql \tilde{p}_{0,5}(\tau,x),
  \tilde{p}_{0,5}(\tau,x)\gle \tilde{p}_{0,4}(\tau,x)\vee \tilde{p}_{0,3}(\tau,x)\geql \tilde{p}_{0,4}(\tau,x),
& \\
\phantom{\Big\{}
  \tilde{p}_{0,4}(\tau,x)\geql \forall x\, \tilde{p}_{0,6}(\tau,x),
& \\
\phantom{\Big\{}
  \tilde{p}_{0,8}(\tau,x)\gle \tilde{p}_{0,9}(\tau,x)\vee \tilde{p}_{0,8}(\tau,x)\geql \tilde{p}_{0,9}(\tau,x)\vee \tilde{p}_{0,6}(\tau,x)\geql \tilde{p}_{0,9}(\tau,x),  
  \tilde{p}_{0,9}(\tau,x)\gle \tilde{p}_{0,8}(\tau,x)\vee \tilde{p}_{0,6}(\tau,x)\geql \gu,
& \\
\phantom{\Big\{}
  \tilde{p}_{0,8}(\tau,x)\geql \mi{uni}(x),
  \tilde{p}_{0,12}(\tau,x)\geql \tilde{p}_{0,13}(\tau,x)\vee \tilde{p}_{0,9}(\tau,x)\geql \gz,
  \tilde{p}_{0,12}(\tau,x)\gle \tilde{p}_{0,13}(\tau,x)\vee \tilde{p}_{0,13}(\tau,x)\gle \tilde{p}_{0,12}(\tau,x)\vee \tilde{p}_{0,9}(\tau,x)\geql \gu,     
& \\
\phantom{\Big\{}
  \tilde{p}_{0,12}(\tau,x)\geql \tilde{H}_{X_3}(\tau,x),
  \tilde{p}_{0,13}(\tau,x)\geql \tilde{G}_{\mi{positive}_{\dot{t}}}(x),  
& \\
\phantom{\Big\{}
  \tilde{p}_{0,5}(\tau,x)\geql \forall x\, \tilde{p}_{0,7}(\tau,x),
& \\
\phantom{\Big\{}
  \tilde{p}_{0,10}(\tau,x)\gle \tilde{p}_{0,11}(\tau,x)\vee \tilde{p}_{0,10}(\tau,x)\geql \tilde{p}_{0,11}(\tau,x)\vee \tilde{p}_{0,7}(\tau,x)\geql \tilde{p}_{0,11}(\tau,x),  
  \tilde{p}_{0,11}(\tau,x)\gle \tilde{p}_{0,10}(\tau,x)\vee \tilde{p}_{0,7}(\tau,x)\geql \gu,
& \\
\phantom{\Big\{}
  \tilde{p}_{0,10}(\tau,x)\geql \mi{uni}(x),
  \tilde{p}_{0,14}(\tau,x)\geql \tilde{p}_{0,15}(\tau,x)\vee \tilde{p}_{0,11}(\tau,x)\geql \gz,
  \tilde{p}_{0,14}(\tau,x)\gle \tilde{p}_{0,15}(\tau,x)\vee \tilde{p}_{0,15}(\tau,x)\gle \tilde{p}_{0,14}(\tau,x)\vee \tilde{p}_{0,11}(\tau,x)\geql \gu,     
& \\
\phantom{\Big\{}
  \tilde{p}_{0,14}(\tau,x)\geql \tilde{H}_{X_5}(\tau,x),
  \tilde{p}_{0,15}(\tau,x)\geql \tilde{G}_{\mi{negative}_{\dot{r}}}(x)\Big\} \\[1mm]
\hline \hline
\end{IEEEeqnarray*}
\end{minipage}     
\vspace{-2mm}      
\end{table*}        
\begin{table*}[p]
\vspace{-6mm}
\caption{Translation of the theory $T_D$ to clausal form}\label{tab13}
\vspace{-6mm}
\centering
\begin{minipage}[t]{\linewidth}
\scriptsize
\begin{IEEEeqnarray*}{LL}
\hline \hline \\[2mm]
\Big\{
  \tilde{p}_{60,0}\geql \gu,
  \tilde{p}_{60,0}\leftrightarrow \mi{nat}(\tilde{z})\Big\}
& \\
\Big\{
  \tilde{p}_{60,0}\geql \gu, 
  \tilde{p}_{60,0}\geql \mi{nat}(\tilde{z})\Big\}
& \\[2mm]
\Big\{
  \tilde{p}_{61,0}(x)\geql \gu,
  \tilde{p}_{61,0}(x)\leftrightarrow \big(\underbrace{\mi{nat}(\tilde{s}(x))}_{\tilde{p}_{61,1}(x)}\leftrightarrow 
                                          \underbrace{\mi{nat}(x)}_{\tilde{p}_{61,2}(x)}\big)\Big\}
& \hspace{-4mm} (\ref{eq0rr33+}) \\
\Big\{
  \tilde{p}_{61,0}(x)\geql \gu,
  \tilde{p}_{61,1}(x)\gle \tilde{p}_{61,2}(x)\vee \tilde{p}_{61,1}(x)\geql \tilde{p}_{61,2}(x)\vee \tilde{p}_{61,0}(x)\geql \tilde{p}_{61,2}(x),
  \tilde{p}_{61,2}(x)\gle \tilde{p}_{61,1}(x)\vee \tilde{p}_{61,2}(x)\geql \tilde{p}_{61,1}(x)\vee \tilde{p}_{61,0}(x)\geql \tilde{p}_{61,1}(x),
& \\
\phantom{\Big\{}
  \tilde{p}_{61,1}(x)\gle \tilde{p}_{61,2}(x)\vee \tilde{p}_{61,2}(x)\gle \tilde{p}_{61,1}(x)\vee \tilde{p}_{61,0}(x)\geql \gu,
  \tilde{p}_{61,1}(x)\geql \mi{nat}(\tilde{s}(x)), 
  \tilde{p}_{61,2}(x)\geql \mi{nat}(x)\Big\}
& \\[2mm]
\Big\{
  \tilde{p}_{62,0}(x,y)\geql \gu,
  \tilde{p}_{62,0}(x,y)\leftrightarrow \big(\underbrace{\mi{rat}(\mi{frac}(x,\tilde{s}(y)))}_{\tilde{p}_{62,1}(x,y)}\leftrightarrow 
                                            \underbrace{\mi{nat}(x)\wedge \mi{nat}(y)}_{\tilde{p}_{62,2}(x,y)}\big)\Big\}
& \hspace{-4mm} (\ref{eq0rr33+}) \\
\Big\{
  \tilde{p}_{62,0}(x,y)\geql \gu,
  \tilde{p}_{62,1}(x,y)\gle \tilde{p}_{62,2}(x,y)\vee \tilde{p}_{62,1}(x,y)\geql \tilde{p}_{62,2}(x,y)\vee \tilde{p}_{62,0}(x,y)\geql \tilde{p}_{62,2}(x,y),
& \\
\phantom{\Big\{}
  \tilde{p}_{62,2}(x,y)\gle \tilde{p}_{62,1}(x,y)\vee \tilde{p}_{62,2}(x,y)\geql \tilde{p}_{62,1}(x,y)\vee \tilde{p}_{62,0}(x,y)\geql \tilde{p}_{62,1}(x,y),
  \tilde{p}_{62,1}(x,y)\gle \tilde{p}_{62,2}(x,y)\vee \tilde{p}_{62,2}(x,y)\gle \tilde{p}_{62,1}(x,y)\vee \tilde{p}_{62,0}(x,y)\geql \gu,
& \\
\phantom{\Big\{}
  \tilde{p}_{62,1}(x,y)\geql \mi{rat}(\mi{frac}(x,\tilde{s}(y))),
  \tilde{p}_{62,2}(x,y)\leftrightarrow \underbrace{\mi{nat}(x)}_{\tilde{p}_{62,3}(x,y)}\wedge \underbrace{\mi{nat}(y)}_{\tilde{p}_{62,4}(x,y)}\Big\}
& \hspace{-4mm} (\ref{eq0rr1+}) \\
\Big\{
  \tilde{p}_{62,0}(x,y)\geql \gu,
  \tilde{p}_{62,1}(x,y)\gle \tilde{p}_{62,2}(x,y)\vee \tilde{p}_{62,1}(x,y)\geql \tilde{p}_{62,2}(x,y)\vee \tilde{p}_{62,0}(x,y)\geql \tilde{p}_{62,2}(x,y),
& \\
\phantom{\Big\{}
  \tilde{p}_{62,2}(x,y)\gle \tilde{p}_{62,1}(x,y)\vee \tilde{p}_{62,2}(x,y)\geql \tilde{p}_{62,1}(x,y)\vee \tilde{p}_{62,0}(x,y)\geql \tilde{p}_{62,1}(x,y),
  \tilde{p}_{62,1}(x,y)\gle \tilde{p}_{62,2}(x,y)\vee \tilde{p}_{62,2}(x,y)\gle \tilde{p}_{62,1}(x,y)\vee \tilde{p}_{62,0}(x,y)\geql \gu,
& \\
\phantom{\Big\{}
  \tilde{p}_{62,1}(x,y)\geql \mi{rat}(\mi{frac}(x,\tilde{s}(y))),
  \tilde{p}_{62,3}(x,y)\gle \tilde{p}_{62,4}(x,y)\vee \tilde{p}_{62,3}(x,y)\geql \tilde{p}_{62,4}(x,y)\vee \tilde{p}_{62,2}(x,y)\geql \tilde{p}_{62,4}(x,y),
& \\
\phantom{\Big\{}
  \tilde{p}_{62,4}(x,y)\gle \tilde{p}_{62,3}(x,y)\vee \tilde{p}_{62,4}(x,y)\geql \tilde{p}_{62,3}(x,y)\vee \tilde{p}_{62,2}(x,y)\geql \tilde{p}_{62,3}(x,y),
  \tilde{p}_{62,3}(x,y)\geql \mi{nat}(x), \tilde{p}_{62,4}(x,y)\geql \mi{nat}(y)\Big\}
& \\[2mm]
\Big\{
  \tilde{p}_{63,0}(x,y)\geql \gu,
  \tilde{p}_{63,0}(x,y)\leftrightarrow \big(\underbrace{\mi{rat}(\mi{-frac}(\tilde{s}(x),\tilde{s}(y)))}_{\tilde{p}_{63,1}(x,y)}\leftrightarrow 
                                            \underbrace{\mi{nat}(x)\wedge \mi{nat}(y)}_{\tilde{p}_{63,2}(x,y)}\big)\Big\}
& \hspace{-4mm} (\ref{eq0rr33+}) \\
\Big\{
  \tilde{p}_{63,0}(x,y)\geql \gu,
  \tilde{p}_{63,1}(x,y)\gle \tilde{p}_{63,2}(x,y)\vee \tilde{p}_{63,1}(x,y)\geql \tilde{p}_{63,2}(x,y)\vee \tilde{p}_{63,0}(x,y)\geql \tilde{p}_{63,2}(x,y),
& \\
\phantom{\Big\{}
  \tilde{p}_{63,2}(x,y)\gle \tilde{p}_{63,1}(x,y)\vee \tilde{p}_{63,2}(x,y)\geql \tilde{p}_{63,1}(x,y)\vee \tilde{p}_{63,0}(x,y)\geql \tilde{p}_{63,1}(x,y),
  \tilde{p}_{63,1}(x,y)\gle \tilde{p}_{63,2}(x,y)\vee \tilde{p}_{63,2}(x,y)\gle \tilde{p}_{63,1}(x,y)\vee \tilde{p}_{63,0}(x,y)\geql \gu,
& \\
\phantom{\Big\{}
  \tilde{p}_{63,1}(x,y)\geql \mi{rat}(\mi{-frac}(\tilde{s}(x),\tilde{s}(y))),
  \tilde{p}_{63,2}(x,y)\leftrightarrow \underbrace{\mi{nat}(x)}_{\tilde{p}_{63,3}(x,y)}\wedge \underbrace{\mi{nat}(y)}_{\tilde{p}_{63,4}(x,y)}\Big\}
& \hspace{-4mm} (\ref{eq0rr1+}) \\
\Big\{
  \tilde{p}_{63,0}(x,y)\geql \gu,
  \tilde{p}_{63,1}(x,y)\gle \tilde{p}_{63,2}(x,y)\vee \tilde{p}_{63,1}(x,y)\geql \tilde{p}_{63,2}(x,y)\vee \tilde{p}_{63,0}(x,y)\geql \tilde{p}_{63,2}(x,y),
& \\
\phantom{\Big\{}
  \tilde{p}_{63,2}(x,y)\gle \tilde{p}_{63,1}(x,y)\vee \tilde{p}_{63,2}(x,y)\geql \tilde{p}_{63,1}(x,y)\vee \tilde{p}_{63,0}(x,y)\geql \tilde{p}_{63,1}(x,y),
  \tilde{p}_{63,1}(x,y)\gle \tilde{p}_{63,2}(x,y)\vee \tilde{p}_{63,2}(x,y)\gle \tilde{p}_{63,1}(x,y)\vee \tilde{p}_{63,0}(x,y)\geql \gu,
& \\
\phantom{\Big\{}
  \tilde{p}_{63,1}(x,y)\geql \mi{rat}(\mi{-frac}(\tilde{s}(x),\tilde{s}(y))),
  \tilde{p}_{63,3}(x,y)\gle \tilde{p}_{63,4}(x,y)\vee \tilde{p}_{63,3}(x,y)\geql \tilde{p}_{63,4}(x,y)\vee \tilde{p}_{63,2}(x,y)\geql \tilde{p}_{63,4}(x,y),
& \\
\phantom{\Big\{}
  \tilde{p}_{63,4}(x,y)\gle \tilde{p}_{63,3}(x,y)\vee \tilde{p}_{63,4}(x,y)\geql \tilde{p}_{63,3}(x,y)\vee \tilde{p}_{63,2}(x,y)\geql \tilde{p}_{63,3}(x,y),
  \tilde{p}_{63,3}(x,y)\geql \mi{nat}(x), \tilde{p}_{63,4}(x,y)\geql \mi{nat}(y)\Big\}
& \\[2mm]
\Big\{
  \mi{nat}(\mi{frac}(x,y))\geql \gz,
  \mi{nat}(\mi{-frac}(x,y))\geql \gz,
  \mi{rat}(\tilde{z})\geql \gz,
  \mi{rat}(\tilde{s}(x))\geql \gz\Big\}
& \\[2mm]
\Big\{
  \tilde{p}_{64,0}(x)\geql \gu,
  \tilde{p}_{64,0}(x)\leftrightarrow \big(\underbrace{\mi{time}(x)}_{\tilde{p}_{64,1}(x)}\leftrightarrow 
                                          \underbrace{\mi{nat}(x)}_{\tilde{p}_{64,2}(x)}\big)\Big\}
& \hspace{-4mm} (\ref{eq0rr33+}) \\
\Big\{
  \tilde{p}_{64,0}(x)\geql \gu,
  \tilde{p}_{64,1}(x)\gle \tilde{p}_{64,2}(x)\vee \tilde{p}_{64,1}(x)\geql \tilde{p}_{64,2}(x)\vee \tilde{p}_{64,0}(x)\geql \tilde{p}_{64,2}(x),
& \\
\phantom{\Big\{}
  \tilde{p}_{64,2}(x)\gle \tilde{p}_{64,1}(x)\vee \tilde{p}_{64,2}(x)\geql \tilde{p}_{64,1}(x)\vee \tilde{p}_{64,0}(x)\geql \tilde{p}_{64,1}(x),
  \tilde{p}_{64,1}(x)\gle \tilde{p}_{64,2}(x)\vee \tilde{p}_{64,2}(x)\gle \tilde{p}_{64,1}(x)\vee \tilde{p}_{64,0}(x)\geql \gu,
& \\
\phantom{\Big\{}  
  \tilde{p}_{64,1}(x)\geql \mi{time}(x), \tilde{p}_{64,2}(x)\geql \mi{nat}(x)\Big\}
& \\[2mm]
\Big\{
  \tilde{p}_{65,0}(x)\geql \gu,
  \tilde{p}_{65,0}(x)\leftrightarrow \big(\underbrace{\mi{uni}(x)}_{\tilde{p}_{65,1}(x)}\rightarrow
                                          \underbrace{\mi{rat}(x)}_{\tilde{p}_{65,2}(x)}\big)\Big\}
& \hspace{-4mm} (\ref{eq0rr3+}) \\
\Big\{
  \tilde{p}_{65,0}(x)\geql \gu,
  \tilde{p}_{65,1}(x)\gle \tilde{p}_{65,2}(x)\vee \tilde{p}_{65,1}(x)\geql \tilde{p}_{65,2}(x)\vee \tilde{p}_{65,0}(x)\geql \tilde{p}_{65,2}(x),
  \tilde{p}_{65,2}(x)\gle \tilde{p}_{65,1}(x)\vee \tilde{p}_{65,0}(x)\geql \gu,
& \\
\phantom{\Big\{}
  \tilde{p}_{65,1}(x)\geql \mi{uni}(x), \tilde{p}_{65,2}(x)\geql \mi{rat}(x)\Big\} \\[1mm]
\hline \hline
\end{IEEEeqnarray*}
\end{minipage}     
\vspace{-2mm}
\end{table*}
\begin{table*}[p]
\vspace{-6mm}
\caption{Clausal theory $S_D\cup S_U^+$}\label{tab14}
\vspace{-6mm}
\centering
\begin{minipage}[t]{\linewidth-70mm}
\scriptsize
\begin{IEEEeqnarray*}{LR}
\hline \hline \\[2mm]
\framebox{$\tilde{p}_{60,0}\geql \gu$}
& [1] \\
\framebox{$\tilde{p}_{60,0}\geql \mi{nat}(\tilde{z})$}
& [2] \\[2mm]
\framebox{$\tilde{p}_{61,0}(x)\geql \gu$}
& [3] \\
\tilde{p}_{61,1}(x)\gle \tilde{p}_{61,2}(x)\vee \tilde{p}_{61,1}(x)\geql \tilde{p}_{61,2}(x)\vee \tilde{p}_{61,0}(x)\geql \tilde{p}_{61,2}(x)
& [4] \\
\framebox{$\tilde{p}_{61,2}(x)\gle \tilde{p}_{61,1}(x)$}\vee \tilde{p}_{61,2}(x)\geql \tilde{p}_{61,1}(x)\vee \tilde{p}_{61,0}(x)\geql \tilde{p}_{61,1}(x)
& [5] \\
\tilde{p}_{61,1}(x)\gle \tilde{p}_{61,2}(x)\vee \tilde{p}_{61,2}(x)\gle \tilde{p}_{61,1}(x)\vee \tilde{p}_{61,0}(x)\geql \gu
& [6] \\  
\framebox{$\tilde{p}_{61,1}(x)\geql \mi{nat}(\tilde{s}(x))$}
& [7] \\  
\framebox{$\tilde{p}_{61,2}(x)\geql \mi{nat}(x)$}
& [8] \\[2mm]
\tilde{p}_{62,0}(x,y)\geql \gu
& [9] \\
\tilde{p}_{62,1}(x,y)\gle \tilde{p}_{62,2}(x,y)\vee \tilde{p}_{62,1}(x,y)\geql \tilde{p}_{62,2}(x,y)\vee \tilde{p}_{62,0}(x,y)\geql \tilde{p}_{62,2}(x,y)
& \qquad [10] \\
\tilde{p}_{62,2}(x,y)\gle \tilde{p}_{62,1}(x,y)\vee \tilde{p}_{62,2}(x,y)\geql \tilde{p}_{62,1}(x,y)\vee \tilde{p}_{62,0}(x,y)\geql \tilde{p}_{62,1}(x,y)
& [11] \\
\tilde{p}_{62,1}(x,y)\gle \tilde{p}_{62,2}(x,y)\vee \tilde{p}_{62,2}(x,y)\gle \tilde{p}_{62,1}(x,y)\vee \tilde{p}_{62,0}(x,y)\geql \gu
& [12] \\
\tilde{p}_{62,1}(x,y)\geql \mi{rat}(\mi{frac}(x,\tilde{s}(y)))
& [13] \\
\tilde{p}_{62,3}(x,y)\gle \tilde{p}_{62,4}(x,y)\vee \tilde{p}_{62,3}(x,y)\geql \tilde{p}_{62,4}(x,y)\vee \tilde{p}_{62,2}(x,y)\geql \tilde{p}_{62,4}(x,y)
& [14] \\
\tilde{p}_{62,4}(x,y)\gle \tilde{p}_{62,3}(x,y)\vee \tilde{p}_{62,4}(x,y)\geql \tilde{p}_{62,3}(x,y)\vee \tilde{p}_{62,2}(x,y)\geql \tilde{p}_{62,3}(x,y)
& [15] \\
\tilde{p}_{62,3}(x,y)\geql \mi{nat}(x)
& [16] \\
\tilde{p}_{62,4}(x,y)\geql \mi{nat}(y)
& [17] \\[2mm]
\tilde{p}_{63,0}(x,y)\geql \gu
& [18] \\
\tilde{p}_{63,1}(x,y)\gle \tilde{p}_{63,2}(x,y)\vee \tilde{p}_{63,1}(x,y)\geql \tilde{p}_{63,2}(x,y)\vee \tilde{p}_{63,0}(x,y)\geql \tilde{p}_{63,2}(x,y)
& [19] \\
\tilde{p}_{63,2}(x,y)\gle \tilde{p}_{63,1}(x,y)\vee \tilde{p}_{63,2}(x,y)\geql \tilde{p}_{63,1}(x,y)\vee \tilde{p}_{63,0}(x,y)\geql \tilde{p}_{63,1}(x,y)
& [20] \\
\tilde{p}_{63,1}(x,y)\gle \tilde{p}_{63,2}(x,y)\vee \tilde{p}_{63,2}(x,y)\gle \tilde{p}_{63,1}(x,y)\vee \tilde{p}_{63,0}(x,y)\geql \gu
& [21] \\
\tilde{p}_{63,1}(x,y)\geql \mi{rat}(\mi{-frac}(\tilde{s}(x),\tilde{s}(y)))
& [22] \\
\tilde{p}_{63,3}(x,y)\gle \tilde{p}_{63,4}(x,y)\vee \tilde{p}_{63,3}(x,y)\geql \tilde{p}_{63,4}(x,y)\vee \tilde{p}_{63,2}(x,y)\geql \tilde{p}_{63,4}(x,y)
& [23] \\
\tilde{p}_{63,4}(x,y)\gle \tilde{p}_{63,3}(x,y)\vee \tilde{p}_{63,4}(x,y)\geql \tilde{p}_{63,3}(x,y)\vee \tilde{p}_{63,2}(x,y)\geql \tilde{p}_{63,3}(x,y)
& [24] \\
\tilde{p}_{63,3}(x,y)\geql \mi{nat}(x) 
& [25] \\
\tilde{p}_{63,4}(x,y)\geql \mi{nat}(y)
& [26] \\[2mm]
\mi{nat}(\mi{frac}(x,y))\geql \gz
& [27] \\
\mi{nat}(\mi{-frac}(x,y))\geql \gz
& [28] \\
\mi{rat}(\tilde{z})\geql \gz
& [29] \\  
\mi{rat}(\tilde{s}(x))\geql \gz
& [30] \\[2mm]
\framebox{$\tilde{p}_{64,0}(x)\geql \gu$}
& [31] \\
\tilde{p}_{64,1}(x)\gle \tilde{p}_{64,2}(x)\vee \tilde{p}_{64,1}(x)\geql \tilde{p}_{64,2}(x)\vee \tilde{p}_{64,0}(x)\geql \tilde{p}_{64,2}(x)
& [32] \\
\framebox{$\tilde{p}_{64,2}(x)\gle \tilde{p}_{64,1}(x)$}\vee \tilde{p}_{64,2}(x)\geql \tilde{p}_{64,1}(x)\vee \tilde{p}_{64,0}(x)\geql \tilde{p}_{64,1}(x)
& [33] \\  
\tilde{p}_{64,1}(x)\gle \tilde{p}_{64,2}(x)\vee \tilde{p}_{64,2}(x)\gle \tilde{p}_{64,1}(x)\vee \tilde{p}_{64,0}(x)\geql \gu
& [34] \\  
\framebox{$\tilde{p}_{64,1}(x)\geql \mi{time}(x)$} 
& [35] \\
\framebox{$\tilde{p}_{64,2}(x)\geql \mi{nat}(x)$}
& [36] \\[2mm]
\tilde{p}_{65,0}(x)\geql \gu
& [37] \\
\tilde{p}_{65,1}(x)\gle \tilde{p}_{65,2}(x)\vee \tilde{p}_{65,1}(x)\geql \tilde{p}_{65,2}(x)\vee \tilde{p}_{65,0}(x)\geql \tilde{p}_{65,2}(x)
& [38] \\
\tilde{p}_{65,2}(x)\gle \tilde{p}_{65,1}(x)\vee \tilde{p}_{65,0}(x)\geql \gu
& [39] \\
\tilde{p}_{65,1}(x)\geql \mi{uni}(x) 
& [40] \\
\tilde{p}_{65,2}(x)\geql \mi{rat}(x)
& [41] \\[2mm]
\framebox{$\mi{uni}(\tilde{0})\geql \gu$}
& [42] \\
\framebox{$\mi{uni}(\tilde{1})\geql \gu$}
& [43] \\
\framebox{$\mi{uni}(\tilde{2})\geql \gu$}
& [44] \\
\framebox{$\mi{uni}(\tilde{3})\geql \gu$}
& [45] \\
\framebox{$\mi{uni}(\tilde{4})\geql \gu$}
& [46] \\[1mm]
\hline \hline
\end{IEEEeqnarray*}
\end{minipage}     
\vspace{-2mm}
\end{table*}
\begin{table*}[p]
\vspace{-6mm}
\caption{Clausal theory $S_\mbb{A}$}\label{tab15}
\vspace{-6mm}
\centering   
\begin{minipage}[t]{\linewidth-130mm}
\scriptsize
\begin{IEEEeqnarray*}{LR}
\hline \hline \\[2mm]
\framebox{$\tilde{G}_{\mi{low}_k}(\tilde{0})\geql \bar{1}$}
& [47k] \\
\tilde{G}_{\mi{low}_k}(\tilde{1})\geql \overline{0.5}
& [48k] \\
\tilde{G}_{\mi{low}_k}(\tilde{2})\geql \bar{0}
& [49k] \\
\tilde{G}_{\mi{low}_k}(\tilde{3})\geql \bar{0}
& [50k] \\
\tilde{G}_{\mi{low}_k}(\tilde{4})\geql \bar{0}
& [51k] \\[2mm]
\tilde{G}_{\mi{medium}_k}(\tilde{0})\geql \bar{0}
& [52k] \\
\tilde{G}_{\mi{medium}_k}(\tilde{1})\geql \overline{0.5}
& [53k] \\
\framebox{$\tilde{G}_{\mi{medium}_k}(\tilde{2})\geql \bar{1}$}
& [54k] \\
\tilde{G}_{\mi{medium}_k}(\tilde{3})\geql \overline{0.5}
& [55k] \\
\tilde{G}_{\mi{medium}_k}(\tilde{4})\geql \bar{0}
& [56k] \\[2mm]
\tilde{G}_{\mi{high}_k}(\tilde{0})\geql \bar{0}
& [57k] \\
\tilde{G}_{\mi{high}_k}(\tilde{1})\geql \bar{0}
& [58k] \\
\tilde{G}_{\mi{high}_k}(\tilde{2})\geql \bar{0}
& [59k] \\
\tilde{G}_{\mi{high}_k}(\tilde{3})\geql \overline{0.5}
& [60k] \\
\framebox{$\tilde{G}_{\mi{high}_k}(\tilde{4})\geql \bar{1}$}
& [61k] \\[2mm]
\tilde{G}_{\mi{negative}_{\dot{k}}}(\tilde{0})\geql \bar{1}
& [62k] \\
\tilde{G}_{\mi{negative}_{\dot{k}}}(\tilde{1})\geql \overline{0.5}
& [63k] \\
\tilde{G}_{\mi{negative}_{\dot{k}}}(\tilde{2})\geql \bar{0}
& [64k] \\
\tilde{G}_{\mi{negative}_{\dot{k}}}(\tilde{3})\geql \bar{0}
& [65k] \\  
\tilde{G}_{\mi{negative}_{\dot{k}}}(\tilde{4})\geql \bar{0}
& [66k] \\[2mm]
\tilde{G}_{\mi{zero}_{\dot{k}}}(\tilde{0})\geql \bar{0}
& [67k] \\
\tilde{G}_{\mi{zero}_{\dot{k}}}(\tilde{1})\geql \overline{0.5}
& [68k] \\
\tilde{G}_{\mi{zero}_{\dot{k}}}(\tilde{2})\geql \bar{1}
& [69k] \\
\tilde{G}_{\mi{zero}_{\dot{k}}}(\tilde{3})\geql \overline{0.5}
& [70k] \\
\tilde{G}_{\mi{zero}_{\dot{k}}}(\tilde{4})\geql \bar{0}
& [71k] \\[2mm]
\tilde{G}_{\mi{positive}_{\dot{k}}}(\tilde{0})\geql \bar{0}
& [72k] \\
\tilde{G}_{\mi{positive}_{\dot{k}}}(\tilde{1})\geql \bar{0}
& [73k] \\
\tilde{G}_{\mi{positive}_{\dot{k}}}(\tilde{2})\geql \bar{0}
& [74k] \\
\tilde{G}_{\mi{positive}_{\dot{k}}}(\tilde{3})\geql \overline{0.5}
& [75k] \\
\framebox{$\tilde{G}_{\mi{positive}_{\dot{k}}}(\tilde{4})\geql \bar{1}$}
& [76k] \\[2mm]
& \qquad k\in \{t,d,r\} \\[1mm]
\hline \hline
\end{IEEEeqnarray*}
\end{minipage}
\vspace{-2mm}
\end{table*}       
\begin{table*}[p]
\vspace{-6mm}
\caption{Clausal theory $S_B$}\label{tab16}
\vspace{-6mm}
\centering   
\begin{minipage}[t]{\linewidth-50mm}
\scriptsize
\begin{IEEEeqnarray*}{LR}
\hline \hline \\[2mm]
\framebox{$\tilde{p}_{1,0}(\tau,x,y)\geql \gu$}
& [77] \\
\framebox{$\tilde{p}_{1,1}(\tau,x,y)\gle \tilde{p}_{1,2}(\tau,x,y)\vee \tilde{p}_{1,1}(\tau,x,y)\geql \tilde{p}_{1,2}(\tau,x,y)\vee
           \tilde{p}_{1,0}(\tau,x,y)\geql \tilde{p}_{1,2}(\tau,x,y)$}
& [78] \\
\framebox{$\tilde{p}_{1,2}(\tau,x,y)\gle \tilde{p}_{1,1}(\tau,x,y)\vee \tilde{p}_{1,0}(\tau,x,y)\geql \gu$}
& [79] \\
\framebox{$\tilde{p}_{1,3}(\tau,x,y)\gle \tilde{p}_{1,4}(\tau,x,y)\vee \tilde{p}_{1,3}(\tau,x,y)\geql \tilde{p}_{1,4}(\tau,x,y)\vee
           \tilde{p}_{1,1}(\tau,x,y)\geql \tilde{p}_{1,4}(\tau,x,y)$}
& [80] \\
\framebox{$\tilde{p}_{1,4}(\tau,x,y)\gle \tilde{p}_{1,3}(\tau,x,y)\vee \tilde{p}_{1,1}(\tau,x,y)\geql \tilde{p}_{1,3}(\tau,x,y)$}
& [81] \\
\framebox{$\tilde{p}_{1,3}(\tau,x,y)\geql \mi{time}(\tau)$} 
& [82] \\
\framebox{$\tilde{p}_{1,4}(\tau,x,y)\geql \mi{uni}(y)$}
& [83] \\
\tilde{p}_{1,5}(\tau,x,y)\geql \tilde{p}_{1,6}(\tau,x,y)\vee \framebox{$\tilde{p}_{1,2}(\tau,x,y)\geql \gz$}
& [84] \\
\tilde{p}_{1,5}(\tau,x,y)\gle \tilde{p}_{1,6}(\tau,x,y)\vee \tilde{p}_{1,6}(\tau,x,y)\gle \tilde{p}_{1,5}(\tau,x,y)\vee
\tilde{p}_{1,2}(\tau,x,y)\geql \gu
& [85] \\
\framebox{$\tilde{p}_{1,5}(\tau,x,y)\geql \tilde{H}_{X_1}(\tilde{s}(\tau),y)$}
& [86] \\
\framebox{$\tilde{p}_{1,7}(\tau,x,y)\gle \tilde{p}_{1,8}(\tau,x,y)\vee \tilde{p}_{1,7}(\tau,x,y)\geql \tilde{p}_{1,8}(\tau,x,y)\vee
           \tilde{p}_{1,6}(\tau,x,y)\geql \tilde{p}_{1,8}(\tau,x,y)$}
& [87] \\
\framebox{$\tilde{p}_{1,8}(\tau,x,y)\gle \tilde{p}_{1,7}(\tau,x,y)\vee \tilde{p}_{1,6}(\tau,x,y)\geql \tilde{p}_{1,7}(\tau,x,y)$}
& [88] \\
\framebox{$\tilde{p}_{1,7}(\tau,x,y)\geql \exists x\, \tilde{p}_{1,9}(\tau,x,y)$}
& [89] \\
\tilde{p}_{1,10}(\tau,x,y)\gle \tilde{p}_{1,11}(\tau,x,y)\vee \tilde{p}_{1,10}(\tau,x,y)\geql \tilde{p}_{1,11}(\tau,x,y)\vee
\tilde{p}_{1,9}(\tau,x,y)\geql \tilde{p}_{1,11}(\tau,x,y)
& [90] \\
\framebox{$\tilde{p}_{1,11}(\tau,x,y)\gle \tilde{p}_{1,10}(\tau,x,y)$}\vee \tilde{p}_{1,9}(\tau,x,y)\geql \tilde{p}_{1,10}(\tau,x,y)
& [91] \\
\framebox{$\tilde{p}_{1,10}(\tau,x,y)\geql \mi{uni}(x)$}
& [92] \\
\tilde{p}_{1,12}(\tau,x,y)\gle \tilde{p}_{1,13}(\tau,x,y)\vee \tilde{p}_{1,12}(\tau,x,y)\geql \tilde{p}_{1,13}(\tau,x,y)\vee
\tilde{p}_{1,11}(\tau,x,y)\geql \tilde{p}_{1,13}(\tau,x,y)
& \qquad [93] \\
\framebox{$\tilde{p}_{1,13}(\tau,x,y)\gle \tilde{p}_{1,12}(\tau,x,y)$}\vee \tilde{p}_{1,11}(\tau,x,y)\geql \tilde{p}_{1,12}(\tau,x,y)
& [94] \\
\framebox{$\tilde{p}_{1,12}(\tau,x,y)\geql \tilde{H}_{X_0}(\tau,x)$}
& [95] \\
\framebox{$\tilde{p}_{1,13}(\tau,x,y)\geql \tilde{G}_{\mi{low}_t}(x)$}
& [96] \\
\framebox{$\tilde{p}_{1,8}(\tau,x,y)\geql \tilde{G}_{\mi{high}_d}(y)$}
& [97] \\[2mm]
\framebox{$\tilde{p}_{6,0}(\tau,x,y)\geql \gu$}
& [98] \\
\framebox{$\tilde{p}_{6,1}(\tau,x,y)\gle \tilde{p}_{6,2}(\tau,x,y)\vee \tilde{p}_{6,1}(\tau,x,y)\geql \tilde{p}_{6,2}(\tau,x,y)\vee
           \tilde{p}_{6,0}(\tau,x,y)\geql \tilde{p}_{6,2}(\tau,x,y)$}
& [99] \\
\framebox{$\tilde{p}_{6,2}(\tau,x,y)\gle \tilde{p}_{6,1}(\tau,x,y)\vee \tilde{p}_{6,0}(\tau,x,y)\geql \gu$}
& [100] \\
\framebox{$\tilde{p}_{6,3}(\tau,x,y)\gle \tilde{p}_{6,4}(\tau,x,y)\vee \tilde{p}_{6,3}(\tau,x,y)\geql \tilde{p}_{6,4}(\tau,x,y)\vee
           \tilde{p}_{6,1}(\tau,x,y)\geql \tilde{p}_{6,4}(\tau,x,y)$}
& [101] \\
\framebox{$\tilde{p}_{6,4}(\tau,x,y)\gle \tilde{p}_{6,3}(\tau,x,y)\vee \tilde{p}_{6,1}(\tau,x,y)\geql \tilde{p}_{6,3}(\tau,x,y)$}
& [102] \\
\framebox{$\tilde{p}_{6,3}(\tau,x,y)\geql \mi{time}(\tau)$} 
& [103] \\
\framebox{$\tilde{p}_{6,4}(\tau,x,y)\geql \mi{uni}(y)$}
& [104] \\
\tilde{p}_{6,5}(\tau,x,y)\geql \tilde{p}_{6,6}(\tau,x,y)\vee \framebox{$\tilde{p}_{6,2}(\tau,x,y)\geql \gz$}
& [105] \\
\tilde{p}_{6,5}(\tau,x,y)\gle \tilde{p}_{6,6}(\tau,x,y)\vee \tilde{p}_{6,6}(\tau,x,y)\gle \tilde{p}_{6,5}(\tau,x,y)\vee
\tilde{p}_{6,2}(\tau,x,y)\geql \gu
& [106] \\
\framebox{$\tilde{p}_{6,5}(\tau,x,y)\geql \tilde{H}_{X_3}(\tilde{s}(\tau),y)$}
& [107] \\
\framebox{$\tilde{p}_{6,7}(\tau,x,y)\gle \tilde{p}_{6,8}(\tau,x,y)\vee \tilde{p}_{6,7}(\tau,x,y)\geql \tilde{p}_{6,8}(\tau,x,y)\vee
           \tilde{p}_{6,6}(\tau,x,y)\geql \tilde{p}_{6,8}(\tau,x,y)$}
& [108] \\
\framebox{$\tilde{p}_{6,8}(\tau,x,y)\gle \tilde{p}_{6,7}(\tau,x,y)\vee \tilde{p}_{6,6}(\tau,x,y)\geql \tilde{p}_{6,7}(\tau,x,y)$}
& [109] \\
\framebox{$\tilde{p}_{6,7}(\tau,x,y)\geql \exists x\, \tilde{p}_{6,9}(\tau,x,y)$}
& [110] \\
\tilde{p}_{6,10}(\tau,x,y)\gle \tilde{p}_{6,11}(\tau,x,y)\vee \tilde{p}_{6,10}(\tau,x,y)\geql \tilde{p}_{6,11}(\tau,x,y)\vee
\tilde{p}_{6,9}(\tau,x,y)\geql \tilde{p}_{6,11}(\tau,x,y)
& [111] \\
\framebox{$\tilde{p}_{6,11}(\tau,x,y)\gle \tilde{p}_{6,10}(\tau,x,y)$}\vee \tilde{p}_{6,9}(\tau,x,y)\geql \tilde{p}_{6,10}(\tau,x,y)
& [112] \\
\framebox{$\tilde{p}_{6,10}(\tau,x,y)\geql \mi{uni}(x)$}
& [113] \\
\tilde{p}_{6,12}(\tau,x,y)\gle \tilde{p}_{6,13}(\tau,x,y)\vee \tilde{p}_{6,12}(\tau,x,y)\geql \tilde{p}_{6,13}(\tau,x,y)\vee
\tilde{p}_{6,11}(\tau,x,y)\geql \tilde{p}_{6,13}(\tau,x,y)
& [114] \\
\framebox{$\tilde{p}_{6,13}(\tau,x,y)\gle \tilde{p}_{6,12}(\tau,x,y)$}\vee \tilde{p}_{6,11}(\tau,x,y)\geql \tilde{p}_{6,12}(\tau,x,y)
& [115] \\
\framebox{$\tilde{p}_{6,12}(\tau,x,y)\geql \tilde{H}_{X_2}(\tau,x)$}
& [116] \\
\framebox{$\tilde{p}_{6,13}(\tau,x,y)\geql \tilde{G}_{\mi{high}_r}(x)$}
& [117] \\
\framebox{$\tilde{p}_{6,8}(\tau,x,y)\geql \tilde{G}_{\mi{positive}_{\dot{t}}}(y)$}
& [118] \\[1mm]
\hline \hline
\end{IEEEeqnarray*}
\end{minipage}
\vspace{-2mm}
\end{table*}       
\begin{table*}[p]
\vspace{-6mm}
\caption{Clausal theory $S_B$}\label{tab17}
\vspace{-6mm}
\centering
\begin{minipage}[t]{\linewidth-50mm}
\scriptsize
\begin{IEEEeqnarray*}{LR}
\hline \hline \\[2mm]
\framebox{$\tilde{p}_{8,0}(\tau,x,y)\geql \gu$}
& [119] \\
\framebox{$\tilde{p}_{8,1}(\tau,x,y)\gle \tilde{p}_{8,2}(\tau,x,y)\vee \tilde{p}_{8,1}(\tau,x,y)\geql \tilde{p}_{8,2}(\tau,x,y)\vee
           \tilde{p}_{8,0}(\tau,x,y)\geql \tilde{p}_{8,2}(\tau,x,y)$}
& [120] \\
\framebox{$\tilde{p}_{8,2}(\tau,x,y)\gle \tilde{p}_{8,1}(\tau,x,y)\vee \tilde{p}_{8,0}(\tau,x,y)\geql \gu$}
& [121] \\
\framebox{$\tilde{p}_{8,3}(\tau,x,y)\gle \tilde{p}_{8,4}(\tau,x,y)\vee \tilde{p}_{8,3}(\tau,x,y)\geql \tilde{p}_{8,4}(\tau,x,y)\vee
           \tilde{p}_{8,1}(\tau,x,y)\geql \tilde{p}_{8,4}(\tau,x,y)$}
& [122] \\
\framebox{$\tilde{p}_{8,4}(\tau,x,y)\gle \tilde{p}_{8,3}(\tau,x,y)\vee \tilde{p}_{8,1}(\tau,x,y)\geql \tilde{p}_{8,3}(\tau,x,y)$}
& [123] \\
\framebox{$\tilde{p}_{8,3}(\tau,x,y)\geql \mi{time}(\tau)$}
& [124] \\
\framebox{$\tilde{p}_{8,4}(\tau,x,y)\geql \mi{uni}(y)$}
& [125] \\
\tilde{p}_{8,5}(\tau,x,y)\geql \tilde{p}_{8,6}(\tau,x,y)\vee \framebox{$\tilde{p}_{8,2}(\tau,x,y)\geql \gz$}
& [126] \\
\tilde{p}_{8,5}(\tau,x,y)\gle \tilde{p}_{8,6}(\tau,x,y)\vee \tilde{p}_{8,6}(\tau,x,y)\gle \tilde{p}_{8,5}(\tau,x,y)\vee
\tilde{p}_{8,2}(\tau,x,y)\geql \gu
& [127] \\
\framebox{$\tilde{p}_{8,5}(\tau,x,y)\geql \tilde{H}_{X_5}(\tilde{s}(\tau),y)$}
& [128] \\
\framebox{$\tilde{p}_{8,7}(\tau,x,y)\gle \tilde{p}_{8,8}(\tau,x,y)\vee \tilde{p}_{8,7}(\tau,x,y)\geql \tilde{p}_{8,8}(\tau,x,y)\vee 
           \tilde{p}_{8,6}(\tau,x,y)\geql \tilde{p}_{8,8}(\tau,x,y)$}
& [129] \\
\framebox{$\tilde{p}_{8,8}(\tau,x,y)\gle \tilde{p}_{8,7}(\tau,x,y)\vee \tilde{p}_{8,6}(\tau,x,y)\geql \tilde{p}_{8,7}(\tau,x,y)$}
& [130] \\
\tilde{p}_{8,9}(\tau,x,y)\gle \tilde{p}_{8,10}(\tau,x,y)\vee \tilde{p}_{8,9}(\tau,x,y)\geql \tilde{p}_{8,10}(\tau,x,y)\vee 
\tilde{p}_{8,7}(\tau,x,y)\geql \tilde{p}_{8,10}(\tau,x,y)
& [131] \\
\framebox{$\tilde{p}_{8,10}(\tau,x,y)\gle \tilde{p}_{8,9}(\tau,x,y)$}\vee \tilde{p}_{8,7}(\tau,x,y)\geql \tilde{p}_{8,9}(\tau,x,y)
& [132] \\
\framebox{$\tilde{p}_{8,9}(\tau,x,y)\geql \exists x\, \tilde{p}_{8,11}(\tau,x,y)$}
& [133] \\
\tilde{p}_{8,13}(\tau,x,y)\gle \tilde{p}_{8,14}(\tau,x,y)\vee \tilde{p}_{8,13}(\tau,x,y)\geql \tilde{p}_{8,14}(\tau,x,y)\vee
\tilde{p}_{8,11}(\tau,x,y)\geql \tilde{p}_{8,14}(\tau,x,y)
& \qquad [134] \\
\framebox{$\tilde{p}_{8,14}(\tau,x,y)\gle \tilde{p}_{8,13}(\tau,x,y)$}\vee \tilde{p}_{8,11}(\tau,x,y)\geql \tilde{p}_{8,13}(\tau,x,y)
& [135] \\
\framebox{$\tilde{p}_{8,13}(\tau,x,y)\geql \mi{uni}(x)$}
& [136] \\
\tilde{p}_{8,17}(\tau,x,y)\gle \tilde{p}_{8,18}(\tau,x,y)\vee \tilde{p}_{8,17}(\tau,x,y)\geql \tilde{p}_{8,18}(\tau,x,y)\vee
\tilde{p}_{8,14}(\tau,x,y)\geql \tilde{p}_{8,18}(\tau,x,y)
& [137] \\
\framebox{$\tilde{p}_{8,18}(\tau,x,y)\gle \tilde{p}_{8,17}(\tau,x,y)$}\vee \tilde{p}_{8,14}(\tau,x,y)\geql \tilde{p}_{8,17}(\tau,x,y)
& [138] \\
\framebox{$\tilde{p}_{8,17}(\tau,x,y)\geql \tilde{H}_{X_1}(\tau,x)$}
& [139] \\
\framebox{$\tilde{p}_{8,18}(\tau,x,y)\geql \tilde{G}_{\mi{high}_d}(x)$}
& [140] \\
\framebox{$\tilde{p}_{8,10}(\tau,x,y)\geql \exists x\, \tilde{p}_{8,12}(\tau,x,y)$}
& [141] \\
\tilde{p}_{8,15}(\tau,x,y)\gle \tilde{p}_{8,16}(\tau,x,y)\vee \tilde{p}_{8,15}(\tau,x,y)\geql \tilde{p}_{8,16}(\tau,x,y)\vee
\tilde{p}_{8,12}(\tau,x,y)\geql \tilde{p}_{8,16}(\tau,x,y)
& [142] \\
\framebox{$\tilde{p}_{8,16}(\tau,x,y)\gle \tilde{p}_{8,15}(\tau,x,y)$}\vee \tilde{p}_{8,12}(\tau,x,y)\geql \tilde{p}_{8,15}(\tau,x,y)
& [143] \\
\framebox{$\tilde{p}_{8,15}(\tau,x,y)\geql \mi{uni}(x)$}
& [144] \\
\tilde{p}_{8,19}(\tau,x,y)\gle \tilde{p}_{8,20}(\tau,x,y)\vee \tilde{p}_{8,19}(\tau,x,y)\geql \tilde{p}_{8,20}(\tau,x,y)\vee
\tilde{p}_{8,16}(\tau,x,y)\geql \tilde{p}_{8,20}(\tau,x,y)
& [145] \\
\framebox{$\tilde{p}_{8,20}(\tau,x,y)\gle \tilde{p}_{8,19}(\tau,x,y)$}\vee \tilde{p}_{8,16}(\tau,x,y)\geql \tilde{p}_{8,19}(\tau,x,y)
& [146] \\
\framebox{$\tilde{p}_{8,19}(\tau,x,y)\geql \tilde{H}_{X_2}(\tau,x)$}
& [147] \\
\framebox{$\tilde{p}_{8,20}(\tau,x,y)\geql \tilde{G}_{\mi{high}_r}(x)$}
& [148] \\
\framebox{$\tilde{p}_{8,8}(\tau,x,y)\geql \tilde{G}_{\mi{negative}_{\dot{r}}}(y)$}
& [149] \\[1mm]
\hline \hline
\end{IEEEeqnarray*}
\end{minipage}
\vspace{-2mm}
\end{table*}       
\begin{table*}[p]
\vspace{-6mm}
\caption{Clausal theory $S_B$}\label{tab17b}
\vspace{-6mm}
\centering
\begin{minipage}[t]{\linewidth-50mm}
\scriptsize
\begin{IEEEeqnarray*}{LR}
\hline \hline \\[2mm]
\framebox{$\tilde{p}_{29,0}(\tau,x,y)\geql \gu$}
& [150] \\
\framebox{$\tilde{p}_{29,1}(\tau,x,y)\gle \tilde{p}_{29,2}(\tau,x,y)\vee \tilde{p}_{29,1}(\tau,x,y)\geql \tilde{p}_{29,2}(\tau,x,y)\vee
           \tilde{p}_{29,0}(\tau,x,y)\geql \tilde{p}_{29,2}(\tau,x,y)$}
& [151] \\
\framebox{$\tilde{p}_{29,2}(\tau,x,y)\gle \tilde{p}_{29,1}(\tau,x,y)\vee \tilde{p}_{29,0}(\tau,x,y)\geql \gu$}
& [152] \\
\framebox{$\tilde{p}_{29,3}(\tau,x,y)\gle \tilde{p}_{29,4}(\tau,x,y)\vee \tilde{p}_{29,3}(\tau,x,y)\geql \tilde{p}_{29,4}(\tau,x,y)\vee
           \tilde{p}_{29,1}(\tau,x,y)\geql \tilde{p}_{29,4}(\tau,x,y)$}
& [153] \\
\framebox{$\tilde{p}_{29,4}(\tau,x,y)\gle \tilde{p}_{29,3}(\tau,x,y)\vee \tilde{p}_{29,1}(\tau,x,y)\geql \tilde{p}_{29,3}(\tau,x,y)$}
& [154] \\
\framebox{$\tilde{p}_{29,3}(\tau,x,y)\geql \mi{time}(\tau)$} 
& [155] \\
\framebox{$\tilde{p}_{29,4}(\tau,x,y)\geql \mi{uni}(y)$}
& [156] \\
\tilde{p}_{29,5}(\tau,x,y)\geql \tilde{p}_{29,6}(\tau,x,y)\vee \framebox{$\tilde{p}_{29,2}(\tau,x,y)\geql \gz$}
& [157] \\
\tilde{p}_{29,5}(\tau,x,y)\gle \tilde{p}_{29,6}(\tau,x,y)\vee \tilde{p}_{29,6}(\tau,x,y)\gle \tilde{p}_{29,5}(\tau,x,y)\vee
\tilde{p}_{29,2}(\tau,x,y)\geql \gu
& [158] \\
\framebox{$\tilde{p}_{29,5}(\tau,x,y)\geql \tilde{H}_{X_2}(\tilde{s}(\tau),y)$}
& [159] \\
\framebox{$\tilde{p}_{29,7}(\tau,x,y)\gle \tilde{p}_{29,8}(\tau,x,y)\vee \tilde{p}_{29,7}(\tau,x,y)\geql \tilde{p}_{29,8}(\tau,x,y)\vee 
           \tilde{p}_{29,6}(\tau,x,y)\geql \tilde{p}_{29,8}(\tau,x,y)$}
& [160] \\
\framebox{$\tilde{p}_{29,8}(\tau,x,y)\gle \tilde{p}_{29,7}(\tau,x,y)\vee \tilde{p}_{29,6}(\tau,x,y)\geql \tilde{p}_{29,7}(\tau,x,y)$}
& [161] \\
\tilde{p}_{29,9}(\tau,x,y)\gle \tilde{p}_{29,10}(\tau,x,y)\vee \tilde{p}_{29,9}(\tau,x,y)\geql \tilde{p}_{29,10}(\tau,x,y)\vee 
\tilde{p}_{29,7}(\tau,x,y)\geql \tilde{p}_{29,10}(\tau,x,y)
& [162] \\
\framebox{$\tilde{p}_{29,10}(\tau,x,y)\gle \tilde{p}_{29,9}(\tau,x,y)$}\vee \tilde{p}_{29,7}(\tau,x,y)\geql \tilde{p}_{29,9}(\tau,x,y)
& [163] \\
\framebox{$\tilde{p}_{29,9}(\tau,x,y)\geql \exists x\, \tilde{p}_{29,11}(\tau,x,y)$}
& [164] \\
\tilde{p}_{29,13}(\tau,x,y)\gle \tilde{p}_{29,14}(\tau,x,y)\vee \tilde{p}_{29,13}(\tau,x,y)\geql \tilde{p}_{29,14}(\tau,x,y)\vee
\tilde{p}_{29,11}(\tau,x,y)\geql \tilde{p}_{29,14}(\tau,x,y)
& \qquad [165] \\
\framebox{$\tilde{p}_{29,14}(\tau,x,y)\gle \tilde{p}_{29,13}(\tau,x,y)$}\vee \tilde{p}_{29,11}(\tau,x,y)\geql \tilde{p}_{29,13}(\tau,x,y)
& [166] \\
\framebox{$\tilde{p}_{29,13}(\tau,x,y)\geql \mi{uni}(x)$}
& [167] \\
\tilde{p}_{29,17}(\tau,x,y)\gle \tilde{p}_{29,18}(\tau,x,y)\vee \tilde{p}_{29,17}(\tau,x,y)\geql \tilde{p}_{29,18}(\tau,x,y)\vee
\tilde{p}_{29,14}(\tau,x,y)\geql \tilde{p}_{29,18}(\tau,x,y)
& [168] \\
\framebox{$\tilde{p}_{29,18}(\tau,x,y)\gle \tilde{p}_{29,17}(\tau,x,y)$}\vee \tilde{p}_{29,14}(\tau,x,y)\geql \tilde{p}_{29,17}(\tau,x,y)
& [169] \\
\framebox{$\tilde{p}_{29,17}(\tau,x,y)\geql \tilde{H}_{X_2}(\tau,x)$}
& [170] \\
\framebox{$\tilde{p}_{29,18}(\tau,x,y)\geql \tilde{G}_{\mi{medium}_r}(x)$}
& [171] \\
\framebox{$\tilde{p}_{29,10}(\tau,x,y)\geql \exists x\, \tilde{p}_{29,12}(\tau,x,y)$}
& [172] \\
\tilde{p}_{29,15}(\tau,x,y)\gle \tilde{p}_{29,16}(\tau,x,y)\vee \tilde{p}_{29,15}(\tau,x,y)\geql \tilde{p}_{29,16}(\tau,x,y)\vee
\tilde{p}_{29,12}(\tau,x,y)\geql \tilde{p}_{29,16}(\tau,x,y)
& [173] \\
\framebox{$\tilde{p}_{29,16}(\tau,x,y)\gle \tilde{p}_{29,15}(\tau,x,y)$}\vee \tilde{p}_{29,12}(\tau,x,y)\geql \tilde{p}_{29,15}(\tau,x,y)
& [174] \\
\framebox{$\tilde{p}_{29,15}(\tau,x,y)\geql \mi{uni}(x)$}
& [175] \\
\tilde{p}_{29,19}(\tau,x,y)\gle \tilde{p}_{29,20}(\tau,x,y)\vee \tilde{p}_{29,19}(\tau,x,y)\geql \tilde{p}_{29,20}(\tau,x,y)\vee
\tilde{p}_{29,16}(\tau,x,y)\geql \tilde{p}_{29,20}(\tau,x,y)
& [176] \\
\framebox{$\tilde{p}_{29,20}(\tau,x,y)\gle \tilde{p}_{29,19}(\tau,x,y)$}\vee \tilde{p}_{29,16}(\tau,x,y)\geql \tilde{p}_{29,19}(\tau,x,y)
& [177] \\
\framebox{$\tilde{p}_{29,19}(\tau,x,y)\geql \tilde{H}_{X_5}(\tau,x)$}
& [178] \\
\framebox{$\tilde{p}_{29,20}(\tau,x,y)\geql \tilde{G}_{\mi{positive}_{\dot{r}}}(x)$}
& [179] \\
\framebox{$\tilde{p}_{29,8}(\tau,x,y)\geql \tilde{G}_{\mi{high}_r}(y)$}
& [180] \\[1mm]  
\hline \hline
\end{IEEEeqnarray*}
\end{minipage}
\vspace{-2mm} 
\end{table*}  
\begin{table*}[p]
\vspace{-6mm}
\caption{Clausal theory $S_{e_0}(\tau/\tilde{z})\cup S_{\phi_r}$}\label{tab18}
\vspace{-6mm}
\centering
\begin{minipage}[t]{\linewidth-70mm}
\scriptsize
\begin{IEEEeqnarray*}{LR}
\hline \hline \\[2mm]
\framebox{$\tilde{H}_{X_0}(\tilde{z},\tilde{0})\geql \bar{1}$}
& [181] \\
\tilde{H}_{X_0}(\tilde{z},\tilde{1})\geql \overline{0.5}
& [182] \\
\tilde{H}_{X_0}(\tilde{z},\tilde{2})\geql \bar{0}
& [183] \\
\tilde{H}_{X_0}(\tilde{z},\tilde{3})\geql \bar{0}
& [184] \\
\tilde{H}_{X_0}(\tilde{z},\tilde{4})\geql \bar{0}
& [185] \\[2mm]
\tilde{H}_{X_2}(\tilde{z},\tilde{0})\geql \bar{0}
& [186] \\
\tilde{H}_{X_2}(\tilde{z},\tilde{1})\geql \overline{0.5}
& [187] \\
\framebox{$\tilde{H}_{X_2}(\tilde{z},\tilde{2})\geql \bar{1}$}
& [188] \\
\tilde{H}_{X_2}(\tilde{z},\tilde{3})\geql \overline{0.5}
& [189] \\
\tilde{H}_{X_2}(\tilde{z},\tilde{4})\geql \bar{0}
& [190] \\[2mm]
\tilde{H}_{X_5}(\tilde{z},\tilde{0})\geql \bar{0}
& [191] \\
\tilde{H}_{X_5}(\tilde{z},\tilde{1})\geql \bar{0}
& [192] \\
\tilde{H}_{X_5}(\tilde{z},\tilde{2})\geql \bar{0}
& [193] \\
\tilde{H}_{X_5}(\tilde{z},\tilde{3})\geql \overline{0.5}
& [194] \\
\framebox{$\tilde{H}_{X_5}(\tilde{z},\tilde{4})\geql \bar{1}$}
& [195] \\[2mm]
\framebox{$\tilde{p}_{0,0}(\tau,x)\gle \gu$}
& [196] \\
\framebox{$\tilde{p}_{0,0}(\tau,x)\geql \exists \tau\, \tilde{p}_{0,1}(\tau,x)$}
& [197] \\
\tilde{p}_{0,2}(\tau,x)\gle \tilde{p}_{0,3}(\tau,x)\vee \tilde{p}_{0,2}(\tau,x)\geql \tilde{p}_{0,3}(\tau,x)\vee 
\tilde{p}_{0,1}(\tau,x)\geql \tilde{p}_{0,3}(\tau,x)
& [198] \\  
\framebox{$\tilde{p}_{0,3}(\tau,x)\gle \tilde{p}_{0,2}(\tau,x)$}\vee \tilde{p}_{0,1}(\tau,x)\geql \tilde{p}_{0,2}(\tau,x)
& [199] \\
\framebox{$\tilde{p}_{0,2}(\tau,x)\geql \mi{time}(\tau)$}
& [200] \\
\tilde{p}_{0,4}(\tau,x)\gle \tilde{p}_{0,5}(\tau,x)\vee \tilde{p}_{0,4}(\tau,x)\geql \tilde{p}_{0,5}(\tau,x)\vee \tilde{p}_{0,3}(\tau,x)\geql \tilde{p}_{0,5}(\tau,x)
& [201] \\
\framebox{$\tilde{p}_{0,5}(\tau,x)\gle \tilde{p}_{0,4}(\tau,x)$}\vee \tilde{p}_{0,3}(\tau,x)\geql \tilde{p}_{0,4}(\tau,x)
& [202] \\
\framebox{$\tilde{p}_{0,4}(\tau,x)\geql \forall x\, \tilde{p}_{0,6}(\tau,x)$}
& [203] \\
\tilde{p}_{0,8}(\tau,x)\gle \tilde{p}_{0,9}(\tau,x)\vee \tilde{p}_{0,8}(\tau,x)\geql \tilde{p}_{0,9}(\tau,x)\vee \tilde{p}_{0,6}(\tau,x)\geql \tilde{p}_{0,9}(\tau,x)
& [204] \\
\framebox{$\tilde{p}_{0,9}(\tau,x)\gle \tilde{p}_{0,8}(\tau,x)$}\vee \tilde{p}_{0,6}(\tau,x)\geql \gu
& [205] \\
\framebox{$\tilde{p}_{0,8}(\tau,x)\geql \mi{uni}(x)$}
& [206] \\
\tilde{p}_{0,12}(\tau,x)\geql \tilde{p}_{0,13}(\tau,x)\vee \tilde{p}_{0,9}(\tau,x)\geql \gz
& [207] \\
\framebox{$\tilde{p}_{0,12}(\tau,x)\gle \tilde{p}_{0,13}(\tau,x)\vee \tilde{p}_{0,13}(\tau,x)\gle \tilde{p}_{0,12}(\tau,x)$}\vee 
           \tilde{p}_{0,9}(\tau,x)\geql \gu     
& [208] \\
\framebox{$\tilde{p}_{0,12}(\tau,x)\geql \tilde{H}_{X_3}(\tau,x)$}
& [209] \\
\framebox{$\tilde{p}_{0,13}(\tau,x)\geql \tilde{G}_{\mi{positive}_{\dot{t}}}(x)$}
& [210] \\
\framebox{$\tilde{p}_{0,5}(\tau,x)\geql \forall x\, \tilde{p}_{0,7}(\tau,x)$}
& [211] \\
\tilde{p}_{0,10}(\tau,x)\gle \tilde{p}_{0,11}(\tau,x)\vee \tilde{p}_{0,10}(\tau,x)\geql \tilde{p}_{0,11}(\tau,x)\vee \tilde{p}_{0,7}(\tau,x)\geql \tilde{p}_{0,11}(\tau,x)  
& \qquad [212] \\
\framebox{$\tilde{p}_{0,11}(\tau,x)\gle \tilde{p}_{0,10}(\tau,x)$}\vee \tilde{p}_{0,7}(\tau,x)\geql \gu
& [213] \\
\framebox{$\tilde{p}_{0,10}(\tau,x)\geql \mi{uni}(x)$}
& [214] \\
\tilde{p}_{0,14}(\tau,x)\geql \tilde{p}_{0,15}(\tau,x)\vee \tilde{p}_{0,11}(\tau,x)\geql \gz
& [215] \\
\framebox{$\tilde{p}_{0,14}(\tau,x)\gle \tilde{p}_{0,15}(\tau,x)\vee \tilde{p}_{0,15}(\tau,x)\gle \tilde{p}_{0,14}(\tau,x)$}\vee 
           \tilde{p}_{0,11}(\tau,x)\geql \gu     
& [216] \\
\framebox{$\tilde{p}_{0,14}(\tau,x)\geql \tilde{H}_{X_5}(\tau,x)$}
& [217] \\
\framebox{$\tilde{p}_{0,15}(\tau,x)\geql \tilde{G}_{\mi{negative}_{\dot{r}}}(x)$}
& [218] \\[1mm]
\hline \hline
\end{IEEEeqnarray*}
\end{minipage}
\vspace{-2mm} 
\end{table*}
\begin{table*}[p]
\vspace{-6mm}
\caption{Refutation of $S_D\cup S_U\cup S_\mbb{A}\cup S_B\cup S_{e_0}(\tau/\tilde{z})\cup S_{\phi_r}$}\label{tab19}
\vspace{-6mm}
\centering
\begin{minipage}[t]{\linewidth-70mm}
\scriptsize
\begin{IEEEeqnarray*}{LR}
\hline \hline \\[2mm]
\text{\bf Rule (\cref{ceq4hr10})} : \mi{nat}(\tilde{z}) :                                                                              
& \\
\framebox{$\mi{nat}(\tilde{z})\gle \gu$}\vee \mi{nat}(\tilde{z})\geql \gu
& [219] \\
\text{\bf Rule (\cref{ceq4hr5})} : [1] [2] [219] : 
& \\
\framebox{$\mi{nat}(\tilde{z})\geql \gu$}
& [220] \\
\text{\bf Rule (\cref{ceq4hr5})} : [5] [8]; x/\tilde{z} : [220] :
& \\
\framebox{$\tilde{p}_{61,2}(\tilde{z})\geql \tilde{p}_{61,1}(\tilde{z})\vee \tilde{p}_{61,0}(\tilde{z})\geql \tilde{p}_{61,1}(\tilde{z})$}
& [221] \\
\text{\bf Rule (\cref{ceq4hr10})} : \mi{nat}(\tilde{s}(\tilde{z})) :           
& \\
\framebox{$\mi{nat}(\tilde{s}(\tilde{z}))\gle \gu$}\vee \mi{nat}(\tilde{s}(\tilde{z}))\geql \gu
& [222] \\
\text{repeatedly \bf Rule (\cref{ceq4hr5})} : [3] [7] [8]; x/\tilde{z} : [220] [221] [222] :
& \\
\framebox{$\mi{nat}(\tilde{s}(\tilde{z}))\geql \gu$}
& [223] \\
\text{\bf Rule (\cref{ceq4hr5})} : [5] [8]; x/\tilde{s}(\tilde{z}) : [223] :
& \\
\framebox{$\tilde{p}_{61,2}(\tilde{s}(\tilde{z}))\geql \tilde{p}_{61,1}(\tilde{s}(\tilde{z}))\vee 
           \tilde{p}_{61,0}(\tilde{s}(\tilde{z}))\geql \tilde{p}_{61,1}(\tilde{s}(\tilde{z}))$}
& [224] \\
\text{\bf Rule (\cref{ceq4hr10})} : \mi{nat}(\tilde{s}(\tilde{s}(\tilde{z}))) :
& \\
\framebox{$\mi{nat}(\tilde{s}(\tilde{s}(\tilde{z})))\gle \gu$}\vee \mi{nat}(\tilde{s}(\tilde{s}(\tilde{z})))\geql \gu
& [225] \\
\text{repeatedly \bf Rule (\cref{ceq4hr5})} : [3] [7] [8]; x/\tilde{s}(\tilde{z}) : [223] [224] [225] :
& \\
\framebox{$\mi{nat}(\tilde{s}(\tilde{s}(\tilde{z})))\geql \gu$}
& [226] \\
\text{\bf Rule (\cref{ceq4hr5})} : [33] [36]; x/\tilde{z} : [220] :
& \\
\framebox{$\tilde{p}_{64,2}(\tilde{z})\geql \tilde{p}_{64,1}(\tilde{z})\vee \tilde{p}_{64,0}(\tilde{z})\geql \tilde{p}_{64,1}(\tilde{z})$}
& [227] \\
\text{\bf Rule (\cref{ceq4hr10})} : \mi{time}(\tilde{z}) :
& \\
\framebox{$\mi{time}(\tilde{z})\gle \gu$}\vee \mi{time}(\tilde{z})\geql \gu
& [228] \\
\text{repeatedly \bf Rule (\cref{ceq4hr5})} : [31] [35] [36]; x/\tilde{z} : [220] [227] [228] :
& \\
\framebox{$\mi{time}(\tilde{z})\geql \gu$}
& [229] \\
\text{\bf Rule (\cref{ceq4hr5})} : [33] [36]; x/\tilde{s}(\tilde{z}) : [223] :
& \\
\framebox{$\tilde{p}_{64,2}(\tilde{s}(\tilde{z}))\geql \tilde{p}_{64,1}(\tilde{s}(\tilde{z}))\vee 
           \tilde{p}_{64,0}(\tilde{s}(\tilde{z}))\geql \tilde{p}_{64,1}(\tilde{s}(\tilde{z}))$}
& [230] \\
\text{\bf Rule (\cref{ceq4hr10})} : \mi{time}(\tilde{s}(\tilde{z})) :
& \\
\framebox{$\mi{time}(\tilde{s}(\tilde{z}))\gle \gu$}\vee \mi{time}(\tilde{s}(\tilde{z}))\geql \gu
& [231] \\
\text{repeatedly \bf Rule (\cref{ceq4hr5})} : [31] [35] [36]; x/\tilde{s}(\tilde{z}) : [223] [230] [231] :
& \\
\framebox{$\mi{time}(\tilde{s}(\tilde{z}))\geql \gu$}
& [232] \\
\text{\bf Rule (\cref{ceq4hr5})} : [33] [36]; x/\tilde{s}(\tilde{s}(\tilde{z})) : [226] :
& \\
\framebox{$\tilde{p}_{64,2}(\tilde{s}(\tilde{s}(\tilde{z})))\geql \tilde{p}_{64,1}(\tilde{s}(\tilde{s}(\tilde{z})))\vee
           \tilde{p}_{64,0}(\tilde{s}(\tilde{s}(\tilde{z})))\geql \tilde{p}_{64,1}(\tilde{s}(\tilde{s}(\tilde{z})))$}  
& \qquad [233] \\
\text{\bf Rule (\cref{ceq4hr10})} : \mi{time}(\tilde{s}(\tilde{s}(\tilde{z}))) :
& \\
\framebox{$\mi{time}(\tilde{s}(\tilde{s}(\tilde{z})))\gle \gu$}\vee \mi{time}(\tilde{s}(\tilde{s}(\tilde{z})))\geql \gu
& [234] \\
\text{repeatedly \bf Rule (\cref{ceq4hr5})} : [31] [35] [36]; x/\tilde{s}(\tilde{s}(\tilde{z})) : [226] [233] [234] :
& \\
\framebox{$\mi{time}(\tilde{s}(\tilde{s}(\tilde{z})))\geql \gu$}
& [235] \\[1mm]  
\hline \hline
\end{IEEEeqnarray*}
\end{minipage}
\vspace{-2mm}
\end{table*}
\begin{table*}[p]
\vspace{-6mm}
\caption{Refutation of $S_D\cup S_U\cup S_\mbb{A}\cup S_B\cup S_{e_0}(\tau/\tilde{z})\cup S_{\phi_r}$}\label{tab20}
\vspace{-6mm}
\centering
\begin{minipage}[t]{\linewidth-40mm}
\scriptsize
\begin{IEEEeqnarray*}{LR}
\hline \hline \\[2mm]
\text{\bf Rule (\cref{ceq4hr10})} : \tilde{p}_{1,7}(\tau,x,y) :
& \\
\framebox{$\tilde{p}_{1,7}(\tau,x,y)\gle \gu$}\vee \tilde{p}_{1,7}(\tau,x,y)\geql \gu
& [236] \\
\text{\bf Rule (\cref{ceq4hr7})} : \exists x\, \tilde{p}_{1,9}(\tau,x,y) :
& \\
\framebox{$\tilde{p}_{1,9}(\tau,x,y)\gle \exists x\, \tilde{p}_{1,9}(\tau,x,y)\vee 
           \tilde{p}_{1,9}(\tau,x,y)\geql \exists x\, \tilde{p}_{1,9}(\tau,x,y)$}
& [237] \\
\text{\bf Rule (\cref{ceq4hr5})} : [47t] : [94] [96]; x/\tilde{0} : 
& \\
\framebox{$\tilde{p}_{1,11}(\tau,\tilde{0},y)\geql \tilde{p}_{1,12}(\tau,\tilde{0},y)$}
& [238] \\
\text{\bf Rule (\cref{ceq4hr5})} : [91] [95]; \tau/\tilde{z}, x/\tilde{0} : [181] : [238]; \tau/\tilde{z} : 
& \\
\framebox{$\tilde{p}_{1,9}(\tilde{z},\tilde{0},y)\geql \tilde{p}_{1,10}(\tilde{z},\tilde{0},y)$}
& [239] \\
\text{repeatedly \bf Rule (\cref{ceq4hr5})} : [42] [239] : [89] [236]; \tau/\tilde{z} : [92] [237]; \tau/\tilde{z}, x/\tilde{0} :
& \\
\framebox{$\tilde{p}_{1,7}(\tilde{z},x,y)\geql \gu$}
& [240] \\
\text{\bf Rule (\cref{ceq4hr12})} : [87] [88] :
& \\
\framebox{$\tilde{p}_{1,7}(\tau,x,y)\gle \tilde{p}_{1,8}(\tau,x,y)$}\vee \tilde{p}_{1,6}(\tau,x,y)\geql \tilde{p}_{1,8}(\tau,x,y)
& [241] \\
\text{\bf Rule (\cref{ceq4hr5})} : [240] : [241]; \tau/\tilde{z} :
& \\
\framebox{$\tilde{p}_{1,6}(\tilde{z},x,y)\geql \tilde{p}_{1,8}(\tilde{z},x,y)$}
& [242] \\
\text{\bf Rule (\cref{ceq4hr12})} : [80] [81] :
& \\
\framebox{$\tilde{p}_{1,3}(\tau,x,y)\gle \tilde{p}_{1,4}(\tau,x,y)$}\vee \tilde{p}_{1,1}(\tau,x,y)\geql \tilde{p}_{1,4}(\tau,x,y)
& [243] \\
\text{\bf Rule (\cref{ceq4hr5})} : [82] [243]; \tau/\tilde{z} : [229] :
& \\
\framebox{$\tilde{p}_{1,1}(\tilde{z},x,y)\geql \tilde{p}_{1,4}(\tilde{z},x,y)$}
& [244] \\
\text{\bf Rule (\cref{ceq4hr12})} : [78] [79] :
& \\
\framebox{$\tilde{p}_{1,1}(\tau,x,y)\gle \tilde{p}_{1,2}(\tau,x,y)$}\vee \tilde{p}_{1,0}(\tau,x,y)\geql \tilde{p}_{1,2}(\tau,x,y)
& [245] \\
\text{\bf Rule (\cref{ceq4hr5})} : [42-46] : [83] [245]; \tau/\tilde{z}, y/\tilde{u} : [244]; y/\tilde{u} : \tilde{u}\in \tilde{\mbb{U}} :
& \\
\framebox{$\tilde{p}_{1,0}(\tilde{z},x,\tilde{u})\geql \tilde{p}_{1,2}(\tilde{z},x,\tilde{u}), \tilde{u}\in \tilde{\mbb{U}}$}
& \qquad [246-250] \\
\mbi{\gz}\mbi{\gle} \mbi{\gu}\in \mi{ordtcons}(S_D\cup S_U\cup S_\mbb{A}\cup S_B\cup S_{e_0}(\tau/\tilde{z})\cup S_{\phi_r})                                                                    
& \\
\framebox{$\gz\gle \gu$}                                                                                                  
& [251] \\
\text{\bf Rule (\cref{ceq4hr5})} : [77] [84]; \tau/\tilde{z}, y/\tilde{u} : [246-250] [251] : \tilde{u}\in \tilde{\mbb{U}} :
& \\
\framebox{$\tilde{p}_{1,5}(\tilde{z},x,\tilde{u})\geql \tilde{p}_{1,6}(\tilde{z},x,\tilde{u}), \tilde{u}\in \tilde{\mbb{U}}$}
& [252-256] \\
\text{\bf Rule (\cref{ceq4hr55})} : \tilde{H}_{X_1}(\tilde{s}(\tau),y'), \tilde{G}_{\mi{high}_d}(y'') :
& \\
\framebox{$\tilde{H}_{X_1}(\tilde{s}(\tau),y')\gle \tilde{G}_{\mi{high}_d}(y'')$}\vee 
\tilde{H}_{X_1}(\tilde{s}(\tau),y')\geql \tilde{G}_{\mi{high}_d}(y'')\vee
\framebox{$\tilde{G}_{\mi{high}_d}(y'')\gle \tilde{H}_{X_1}(\tilde{s}(\tau),y')$}
& [257] \\
\text{repeatedly \bf Rule (\cref{ceq4hr5})} : [86] [97]; \tau/\tilde{z}, y/\tilde{u} : [242]; y/\tilde{u} : [252-256] : 
                                              [257]; \tau/\tilde{z}, y'/\tilde{u}, y''/\tilde{u} : \tilde{u}\in \tilde{\mbb{U}} :
& \\
\framebox{$\tilde{H}_{X_1}(\tilde{s}(\tilde{z}),\tilde{u})\geql \tilde{G}_{\mi{high}_d}(\tilde{u}), \tilde{u}\in \tilde{\mbb{U}}$}
& [258-262] \\[1mm]
\hline \hline
\end{IEEEeqnarray*}
\end{minipage}
\vspace{-2mm}
\end{table*}
\begin{table*}[p]
\vspace{-6mm}
\caption{Refutation of $S_D\cup S_U\cup S_\mbb{A}\cup S_B\cup S_{e_0}(\tau/\tilde{z})\cup S_{\phi_r}$}\label{tab21}
\vspace{-6mm}
\centering
\begin{minipage}[t]{\linewidth-40mm}
\scriptsize
\begin{IEEEeqnarray*}{LR}
\hline \hline \\[2mm]
\text{\bf Rule (\cref{ceq4hr10})} : \tilde{p}_{29,9}(\tau,x,y) :
& \\
\framebox{$\tilde{p}_{29,9}(\tau,x,y)\gle \gu$}\vee \tilde{p}_{29,9}(\tau,x,y)\geql \gu
& [263] \\
\text{\bf Rule (\cref{ceq4hr7})} : \exists x\, \tilde{p}_{29,11}(\tau,x,y) :
& \\
\framebox{$\tilde{p}_{29,11}(\tau,x,y)\gle \exists x\, \tilde{p}_{29,11}(\tau,x,y)\vee 
           \tilde{p}_{29,11}(\tau,x,y)\geql \exists x\, \tilde{p}_{29,11}(\tau,x,y)$}
& [264] \\
\text{\bf Rule (\cref{ceq4hr5})} : [54r] : [169] [171]; x/\tilde{2} : 
& \\
\framebox{$\tilde{p}_{29,14}(\tau,\tilde{2},y)\geql \tilde{p}_{29,17}(\tau,\tilde{2},y)$}
& [265] \\
\text{\bf Rule (\cref{ceq4hr5})} : [166] [170]; \tau/\tilde{z}, x/\tilde{2} : [188] : [265]; \tau/\tilde{z} : 
& \\
\framebox{$\tilde{p}_{29,11}(\tilde{z},\tilde{2},y)\geql \tilde{p}_{29,13}(\tilde{z},\tilde{2},y)$}
& [266] \\
\text{repeatedly \bf Rule (\cref{ceq4hr5})} : [44] [266] : [164] [263]; \tau/\tilde{z} : [167] [264]; \tau/\tilde{z}, x/\tilde{2} :
& \\
\framebox{$\tilde{p}_{29,9}(\tilde{z},x,y)\geql \gu$}
& [267] \\
\text{\bf Rule (\cref{ceq4hr10})} : \tilde{p}_{29,10}(\tau,x,y) :
& \\
\framebox{$\tilde{p}_{29,10}(\tau,x,y)\gle \gu$}\vee \tilde{p}_{29,10}(\tau,x,y)\geql \gu
& [268] \\
\text{\bf Rule (\cref{ceq4hr7})} : \exists x\, \tilde{p}_{29,12}(\tau,x,y) :
& \\
\framebox{$\tilde{p}_{29,12}(\tau,x,y)\gle \exists x\, \tilde{p}_{29,12}(\tau,x,y)\vee
           \tilde{p}_{29,12}(\tau,x,y)\geql \exists x\, \tilde{p}_{29,12}(\tau,x,y)$} 
& [269] \\
\text{\bf Rule (\cref{ceq4hr5})} : [76r] : [177] [179]; x/\tilde{4} :
& \\
\framebox{$\tilde{p}_{29,16}(\tau,\tilde{4},y)\geql \tilde{p}_{29,19}(\tau,\tilde{4},y)$}
& [270] \\
\text{\bf Rule (\cref{ceq4hr5})} : [174] [178]; \tau/\tilde{z}, x/\tilde{4} : [195] : [270]; \tau/\tilde{z} :
& \\
\framebox{$\tilde{p}_{29,12}(\tilde{z},\tilde{4},y)\geql \tilde{p}_{29,15}(\tilde{z},\tilde{4},y)$}
& [271] \\
\text{repeatedly \bf Rule (\cref{ceq4hr5})} : [46] [271] : [172] [268]; \tau/\tilde{z} : [175] [269]; \tau/\tilde{z}, x/\tilde{4} :
& \\
\framebox{$\tilde{p}_{29,10}(\tilde{z},x,y)\geql \gu$}
& [272] \\
\text{\bf Rule (\cref{ceq4hr5})} : [272] : [163]; \tau/\tilde{z} :
& \\
\framebox{$\tilde{p}_{29,7}(\tilde{z},x,y)\geql \tilde{p}_{29,9}(\tilde{z},x,y)$}
& [273] \\
\text{\bf Rule (\cref{ceq4hr12})} : [160] [161] :
& \\
\framebox{$\tilde{p}_{29,7}(\tau,x,y)\gle \tilde{p}_{29,8}(\tau,x,y)$}\vee \tilde{p}_{29,6}(\tau,x,y)\geql \tilde{p}_{29,8}(\tau,x,y)
& [274] \\
\text{\bf Rule (\cref{ceq4hr5})} : [267] [273] : [274]; \tau/\tilde{z} :
& \\
\framebox{$\tilde{p}_{29,6}(\tilde{z},x,y)\geql \tilde{p}_{29,8}(\tilde{z},x,y)$}
& [275] \\
\text{\bf Rule (\cref{ceq4hr12})} : [153] [154] :
& \\
\framebox{$\tilde{p}_{29,3}(\tau,x,y)\gle \tilde{p}_{29,4}(\tau,x,y)$}\vee \tilde{p}_{29,1}(\tau,x,y)\geql \tilde{p}_{29,4}(\tau,x,y)
& [276] \\
\text{\bf Rule (\cref{ceq4hr5})} : [155] [276]; \tau/\tilde{z} : [229] :
& \\
\framebox{$\tilde{p}_{29,1}(\tilde{z},x,y)\geql \tilde{p}_{29,4}(\tilde{z},x,y)$}
& [277] \\
\text{\bf Rule (\cref{ceq4hr12})} : [151] [152] :
& \\
\framebox{$\tilde{p}_{29,1}(\tau,x,y)\gle \tilde{p}_{29,2}(\tau,x,y)$}\vee \tilde{p}_{29,0}(\tau,x,y)\geql \tilde{p}_{29,2}(\tau,x,y)
& [278] \\
\text{\bf Rule (\cref{ceq4hr5})} : [42-46] : [156] [278]; \tau/\tilde{z}, y/\tilde{u} : [277]; y/\tilde{u} : \tilde{u}\in \tilde{\mbb{U}} :
& \\
\framebox{$\tilde{p}_{29,0}(\tilde{z},x,\tilde{u})\geql \tilde{p}_{29,2}(\tilde{z},x,\tilde{u}), \tilde{u}\in \tilde{\mbb{U}}$}
& \qquad [279-283] \\
\mbi{\gz}\mbi{\gle} \mbi{\gu}\in \mi{ordtcons}(S_D\cup S_U\cup S_\mbb{A}\cup S_B\cup S_{e_0}(\tau/\tilde{z})\cup S_{\phi_r})                                                                    
& \\
\framebox{$\gz\gle \gu$}                                                                                                  
& [284] \\
\text{\bf Rule (\cref{ceq4hr5})} : [150] [157]; \tau/\tilde{z}, y/\tilde{u} : [279-283] [284] : \tilde{u}\in \tilde{\mbb{U}} :
& \\
\framebox{$\tilde{p}_{29,5}(\tilde{z},x,\tilde{u})\geql \tilde{p}_{29,6}(\tilde{z},x,\tilde{u}), \tilde{u}\in \tilde{\mbb{U}}$}
& [285-289] \\
\text{\bf Rule (\cref{ceq4hr55})} : \tilde{H}_{X_2}(\tilde{s}(\tau),y'), \tilde{G}_{\mi{high}_r}(y'') :
& \\
\framebox{$\tilde{H}_{X_2}(\tilde{s}(\tau),y')\gle \tilde{G}_{\mi{high}_r}(y'')$}\vee 
\tilde{H}_{X_2}(\tilde{s}(\tau),y')\geql \tilde{G}_{\mi{high}_r}(y'')\vee
\framebox{$\tilde{G}_{\mi{high}_r}(y'')\gle \tilde{H}_{X_2}(\tilde{s}(\tau),y')$}
& [290] \\
\text{repeatedly \bf Rule (\cref{ceq4hr5})} : [159] [180]; \tau/\tilde{z}, y/\tilde{u} : [275]; y/\tilde{u} : [285-289] : 
                                              [290]; \tau/\tilde{z}, y'/\tilde{u}, y''/\tilde{u} : \tilde{u}\in \tilde{\mbb{U}} :
& \\
\framebox{$\tilde{H}_{X_2}(\tilde{s}(\tilde{z}),\tilde{u})\geql \tilde{G}_{\mi{high}_r}(\tilde{u}), \tilde{u}\in \tilde{\mbb{U}}$}
& [291-295] \\[1mm]
\hline \hline
\end{IEEEeqnarray*}
\end{minipage}
\vspace{-2mm}
\end{table*}
\begin{table*}[p]
\vspace{-6mm}
\caption{Refutation of $S_D\cup S_U\cup S_\mbb{A}\cup S_B\cup S_{e_0}(\tau/\tilde{z})\cup S_{\phi_r}$}\label{tab22}
\vspace{-6mm}
\centering
\begin{minipage}[t]{\linewidth-30mm}
\scriptsize
\begin{IEEEeqnarray*}{LR}
\hline \hline \\[2mm]
\text{\bf Rule (\cref{ceq4hr10})} : \tilde{p}_{6,7}(\tau,x,y) :
& \\
\framebox{$\tilde{p}_{6,7}(\tau,x,y)\gle \gu$}\vee \tilde{p}_{6,7}(\tau,x,y)\geql \gu
& [296] \\
\text{\bf Rule (\cref{ceq4hr7})} : \exists x\, \tilde{p}_{6,9}(\tau,x,y) :
& \\
\framebox{$\tilde{p}_{6,9}(\tau,x,y)\gle \exists x\, \tilde{p}_{6,9}(\tau,x,y)\vee 
           \tilde{p}_{6,9}(\tau,x,y)\geql \exists x\, \tilde{p}_{6,9}(\tau,x,y)$}
& [297] \\
\text{\bf Rule (\cref{ceq4hr5})} : [61r] : [115] [117]; x/\tilde{4} : 
& \\
\framebox{$\tilde{p}_{6,11}(\tau,\tilde{4},y)\geql \tilde{p}_{6,12}(\tau,\tilde{4},y)$}
& [298] \\
\text{\bf Rule (\cref{ceq4hr5})} : [61r] [295] : [112] [116]; \tau/\tilde{s}(\tilde{z}), x/\tilde{4} : [298]; \tau/\tilde{s}(\tilde{z}) : 
& \\
\framebox{$\tilde{p}_{6,9}(\tilde{s}(\tilde{z}),\tilde{4},y)\geql \tilde{p}_{6,10}(\tilde{s}(\tilde{z}),\tilde{4},y)$}
& [299] \\
\text{repeatedly \bf Rule (\cref{ceq4hr5})} : [46] [299] : [110] [296]; \tau/\tilde{s}(\tilde{z}) : 
                                              [113] [297]; \tau/\tilde{s}(\tilde{z}), x/\tilde{4} :
& \\
\framebox{$\tilde{p}_{6,7}(\tilde{s}(\tilde{z}),x,y)\geql \gu$}
& [300] \\
\text{\bf Rule (\cref{ceq4hr12})} : [108] [109] :
& \\
\framebox{$\tilde{p}_{6,7}(\tau,x,y)\gle \tilde{p}_{6,8}(\tau,x,y)$}\vee \tilde{p}_{6,6}(\tau,x,y)\geql \tilde{p}_{6,8}(\tau,x,y)
& [301] \\
\text{\bf Rule (\cref{ceq4hr5})} : [300] : [301]; \tau/\tilde{s}(\tilde{z}) :
& \\
\framebox{$\tilde{p}_{6,6}(\tilde{s}(\tilde{z}),x,y)\geql \tilde{p}_{6,8}(\tilde{s}(\tilde{z}),x,y)$}
& [302] \\
\text{\bf Rule (\cref{ceq4hr12})} : [101] [102] :
& \\
\framebox{$\tilde{p}_{6,3}(\tau,x,y)\gle \tilde{p}_{6,4}(\tau,x,y)$}\vee \tilde{p}_{6,1}(\tau,x,y)\geql \tilde{p}_{6,4}(\tau,x,y)
& [303] \\
\text{\bf Rule (\cref{ceq4hr5})} : [103] [303]; \tau/\tilde{s}(\tilde{z}) : [232] :
& \\
\framebox{$\tilde{p}_{6,1}(\tilde{s}(\tilde{z}),x,y)\geql \tilde{p}_{6,4}(\tilde{s}(\tilde{z}),x,y)$}
& [304] \\
\text{\bf Rule (\cref{ceq4hr12})} : [99] [100] :
& \\
\framebox{$\tilde{p}_{6,1}(\tau,x,y)\gle \tilde{p}_{6,2}(\tau,x,y)$}\vee \tilde{p}_{6,0}(\tau,x,y)\geql \tilde{p}_{6,2}(\tau,x,y)
& [305] \\
\text{\bf Rule (\cref{ceq4hr5})} : [42-46] : [104] [305]; \tau/\tilde{s}(\tilde{z}), y/\tilde{u} : [304]; y/\tilde{u} : 
                                   \tilde{u}\in \tilde{\mbb{U}} :
& \\
\framebox{$\tilde{p}_{6,0}(\tilde{s}(\tilde{z}),x,\tilde{u})\geql \tilde{p}_{6,2}(\tilde{s}(\tilde{z}),x,\tilde{u}), 
           \tilde{u}\in \tilde{\mbb{U}}$}
& \qquad [306-310] \\
\mbi{\gz}\mbi{\gle} \mbi{\gu}\in \mi{ordtcons}(S_D\cup S_U\cup S_\mbb{A}\cup S_B\cup S_{e_0}(\tau/\tilde{z})\cup S_{\phi_r})                                                                    
& \\
\framebox{$\gz\gle \gu$}                                                                                                  
& [311] \\
\text{\bf Rule (\cref{ceq4hr5})} : [98] [105]; \tau/\tilde{s}(\tilde{z}), y/\tilde{u} : [306-310] [311] : \tilde{u}\in \tilde{\mbb{U}} :
& \\
\framebox{$\tilde{p}_{6,5}(\tilde{s}(\tilde{z}),x,\tilde{u})\geql \tilde{p}_{6,6}(\tilde{s}(\tilde{z}),x,\tilde{u}), 
           \tilde{u}\in \tilde{\mbb{U}}$}
& [312-316] \\
\text{\bf Rule (\cref{ceq4hr55})} : \tilde{H}_{X_3}(\tilde{s}(\tau),y'), \tilde{G}_{\mi{positive}_{\dot{t}}}(y'') :
& \\
\framebox{$\tilde{H}_{X_3}(\tilde{s}(\tau),y')\gle \tilde{G}_{\mi{positive}_{\dot{t}}}(y'')$}\vee 
\tilde{H}_{X_3}(\tilde{s}(\tau),y')\geql \tilde{G}_{\mi{positive}_{\dot{t}}}(y'')\vee
\framebox{$\tilde{G}_{\mi{positive}_{\dot{t}}}(y'')\gle \tilde{H}_{X_3}(\tilde{s}(\tau),y')$}
& [317] \\
\text{repeatedly \bf Rule (\cref{ceq4hr5})} : [107] [118]; \tau/\tilde{s}(\tilde{z}), y/\tilde{u} : [302]; y/\tilde{u} : [312-316] : 
                                              [317]; \tau/\tilde{s}(\tilde{z}), y'/\tilde{u}, y''/\tilde{u} : \tilde{u}\in \tilde{\mbb{U}} :
& \\
\framebox{$\tilde{H}_{X_3}(\tilde{s}(\tilde{s}(\tilde{z})),\tilde{u})\geql \tilde{G}_{\mi{positive}_{\dot{t}}}(\tilde{u}), 
           \tilde{u}\in \tilde{\mbb{U}}$}
& [318-322] \\[1mm]
\hline \hline
\end{IEEEeqnarray*}
\end{minipage}
\vspace{-2mm}
\end{table*}
\begin{table*}[p]
\vspace{-6mm}
\caption{Refutation of $S_D\cup S_U\cup S_\mbb{A}\cup S_B\cup S_{e_0}(\tau/\tilde{z})\cup S_{\phi_r}$}\label{tab23}
\vspace{-6mm}
\centering
\begin{minipage}[t]{\linewidth-30mm}
\scriptsize
\begin{IEEEeqnarray*}{LR}
\hline \hline \\[2mm]
\text{\bf Rule (\cref{ceq4hr10})} : \tilde{p}_{8,9}(\tau,x,y) :
& \\
\framebox{$\tilde{p}_{8,9}(\tau,x,y)\gle \gu$}\vee \tilde{p}_{8,9}(\tau,x,y)\geql \gu
& [323] \\
\text{\bf Rule (\cref{ceq4hr7})} : \exists x\, \tilde{p}_{8,11}(\tau,x,y) :
& \\
\framebox{$\tilde{p}_{8,11}(\tau,x,y)\gle \exists x\, \tilde{p}_{8,11}(\tau,x,y)\vee 
           \tilde{p}_{8,11}(\tau,x,y)\geql \exists x\, \tilde{p}_{8,11}(\tau,x,y)$}
& [324] \\
\text{\bf Rule (\cref{ceq4hr5})} : [61d] : [138] [140]; x/\tilde{4} : 
& \\
\framebox{$\tilde{p}_{8,14}(\tau,\tilde{4},y)\geql \tilde{p}_{8,17}(\tau,\tilde{4},y)$}
& [325] \\
\text{\bf Rule (\cref{ceq4hr5})} : [61d] [262] : [135] [139]; \tau/\tilde{s}(\tilde{z}), x/\tilde{4} : [325]; \tau/\tilde{s}(\tilde{z}) : 
& \\
\framebox{$\tilde{p}_{8,11}(\tilde{s}(\tilde{z}),\tilde{4},y)\geql \tilde{p}_{8,13}(\tilde{s}(\tilde{z}),\tilde{4},y)$}
& [326] \\
\text{repeatedly \bf Rule (\cref{ceq4hr5})} : [46] [326] : [133] [323]; \tau/\tilde{s}(\tilde{z}) : 
                                              [136] [324]; \tau/\tilde{s}(\tilde{z}), x/\tilde{4} :
& \\
\framebox{$\tilde{p}_{8,9}(\tilde{s}(\tilde{z}),x,y)\geql \gu$}
& [327] \\
\text{\bf Rule (\cref{ceq4hr10})} : \tilde{p}_{8,10}(\tau,x,y) :
& \\
\framebox{$\tilde{p}_{8,10}(\tau,x,y)\gle \gu$}\vee \tilde{p}_{8,10}(\tau,x,y)\geql \gu
& [328] \\
\text{\bf Rule (\cref{ceq4hr7})} : \exists x\, \tilde{p}_{8,12}(\tau,x,y) :
& \\
\framebox{$\tilde{p}_{8,12}(\tau,x,y)\gle \exists x\, \tilde{p}_{8,12}(\tau,x,y)\vee
           \tilde{p}_{8,12}(\tau,x,y)\geql \exists x\, \tilde{p}_{8,12}(\tau,x,y)$} 
& [329] \\
\text{\bf Rule (\cref{ceq4hr5})} : [61r] : [146] [148]; x/\tilde{4} :
& \\
\framebox{$\tilde{p}_{8,16}(\tau,\tilde{4},y)\geql \tilde{p}_{8,19}(\tau,\tilde{4},y)$}
& [330] \\
\text{\bf Rule (\cref{ceq4hr5})} : [61r] [295] : [143] [147]; \tau/\tilde{s}(\tilde{z}), x/\tilde{4} : [330]; \tau/\tilde{s}(\tilde{z}) :
& \\
\framebox{$\tilde{p}_{8,12}(\tilde{s}(\tilde{z}),\tilde{4},y)\geql \tilde{p}_{8,15}(\tilde{s}(\tilde{z}),\tilde{4},y)$}
& [331] \\
\text{repeatedly \bf Rule (\cref{ceq4hr5})} : [46] [331] : [141] [328]; \tau/\tilde{s}(\tilde{z}) : 
                                              [144] [329]; \tau/\tilde{s}(\tilde{z}), x/\tilde{4} :
& \\
\framebox{$\tilde{p}_{8,10}(\tilde{s}(\tilde{z}),x,y)\geql \gu$}
& [332] \\
\text{\bf Rule (\cref{ceq4hr5})} : [332] : [132]; \tau/\tilde{s}(\tilde{z}) :
& \\
\framebox{$\tilde{p}_{8,7}(\tilde{s}(\tilde{z}),x,y)\geql \tilde{p}_{8,9}(\tilde{s}(\tilde{z}),x,y)$}
& [333] \\
\text{\bf Rule (\cref{ceq4hr12})} : [129] [130] :
& \\
\framebox{$\tilde{p}_{8,7}(\tau,x,y)\gle \tilde{p}_{8,8}(\tau,x,y)$}\vee \tilde{p}_{8,6}(\tau,x,y)\geql \tilde{p}_{8,8}(\tau,x,y)
& [334] \\
\text{\bf Rule (\cref{ceq4hr5})} : [327] [333] : [334]; \tau/\tilde{s}(\tilde{z}) :
& \\
\framebox{$\tilde{p}_{8,6}(\tilde{s}(\tilde{z}),x,y)\geql \tilde{p}_{8,8}(\tilde{s}(\tilde{z}),x,y)$}
& [335] \\
\text{\bf Rule (\cref{ceq4hr12})} : [122] [123] :
& \\
\framebox{$\tilde{p}_{8,3}(\tau,x,y)\gle \tilde{p}_{8,4}(\tau,x,y)$}\vee \tilde{p}_{8,1}(\tau,x,y)\geql \tilde{p}_{8,4}(\tau,x,y)
& [336] \\
\text{\bf Rule (\cref{ceq4hr5})} : [124] [336]; \tau/\tilde{s}(\tilde{z}) : [232] :
& \\
\framebox{$\tilde{p}_{8,1}(\tilde{s}(\tilde{z}),x,y)\geql \tilde{p}_{8,4}(\tilde{s}(\tilde{z}),x,y)$}
& [337] \\
\text{\bf Rule (\cref{ceq4hr12})} : [120] [121] :
& \\
\framebox{$\tilde{p}_{8,1}(\tau,x,y)\gle \tilde{p}_{8,2}(\tau,x,y)$}\vee \tilde{p}_{8,0}(\tau,x,y)\geql \tilde{p}_{8,2}(\tau,x,y)
& [338] \\
\text{\bf Rule (\cref{ceq4hr5})} : [42-46] : [125] [338]; \tau/\tilde{s}(\tilde{z}), y/\tilde{u} : 
                                   [337]; y/\tilde{u} : \tilde{u}\in \tilde{\mbb{U}} :
& \\
\framebox{$\tilde{p}_{8,0}(\tilde{s}(\tilde{z}),x,\tilde{u})\geql \tilde{p}_{8,2}(\tilde{s}(\tilde{z}),x,\tilde{u}), \tilde{u}\in \tilde{\mbb{U}}$}
& \qquad [339-343] \\
\mbi{\gz}\mbi{\gle} \mbi{\gu}\in \mi{ordtcons}(S_D\cup S_U\cup S_\mbb{A}\cup S_B\cup S_{e_0}(\tau/\tilde{z})\cup S_{\phi_r})                                                                    
& \\
\framebox{$\gz\gle \gu$}                                                                                                  
& [344] \\
\text{\bf Rule (\cref{ceq4hr5})} : [119] [126]; \tau/\tilde{s}(\tilde{z}), y/\tilde{u} : [339-343] [344] : \tilde{u}\in \tilde{\mbb{U}} :
& \\
\framebox{$\tilde{p}_{8,5}(\tilde{s}(\tilde{z}),x,\tilde{u})\geql \tilde{p}_{8,6}(\tilde{s}(\tilde{z}),x,\tilde{u}), \tilde{u}\in \tilde{\mbb{U}}$}
& [345-349] \\
\text{\bf Rule (\cref{ceq4hr55})} : \tilde{H}_{X_5}(\tilde{s}(\tau),y'), \tilde{G}_{\mi{negative}_{\dot{r}}}(y'') :
& \\
\framebox{$\tilde{H}_{X_5}(\tilde{s}(\tau),y')\gle \tilde{G}_{\mi{negative}_{\dot{r}}}(y'')$}\vee 
\tilde{H}_{X_5}(\tilde{s}(\tau),y')\geql \tilde{G}_{\mi{negative}_{\dot{r}}}(y'')\vee
\framebox{$\tilde{G}_{\mi{negative}_{\dot{r}}}(y'')\gle \tilde{H}_{X_5}(\tilde{s}(\tau),y')$}
& [350] \\
\text{repeatedly \bf Rule (\cref{ceq4hr5})} : [128] [149]; \tau/\tilde{s}(\tilde{z}), y/\tilde{u} : [335]; y/\tilde{u} : [345-349] : 
                                              [350]; \tau/\tilde{s}(\tilde{z}), y'/\tilde{u}, y''/\tilde{u} : \tilde{u}\in \tilde{\mbb{U}} :
& \\
\framebox{$\tilde{H}_{X_5}(\tilde{s}(\tilde{s}(\tilde{z})),\tilde{u})\geql \tilde{G}_{\mi{negative}_{\dot{r}}}(\tilde{u}), \tilde{u}\in \tilde{\mbb{U}}$}
& [351-355] \\[1mm]
\hline \hline
\end{IEEEeqnarray*}
\end{minipage}
\vspace{-2mm}
\end{table*}
\begin{table*}[p]
\vspace{-6mm}
\caption{Refutation of $S_D\cup S_U\cup S_\mbb{A}\cup S_B\cup S_{e_0}(\tau/\tilde{z})\cup S_{\phi_r}$}\label{tab24}
\vspace{-6mm}
\centering
\begin{minipage}[t]{\linewidth-60mm}
\scriptsize
\begin{IEEEeqnarray*}{LR}
\hline \hline \\[2mm]
\text{repeatedly \bf Rule (\cref{ceq4hr5})} : [208] [209] [210]; \tau/\tilde{s}(\tilde{s}(\tilde{z})), x/\tilde{u} : [318-322] : 
                                              \tilde{u}\in \tilde{\mbb{U}} :
& \\
\framebox{$\tilde{p}_{0,9}(\tilde{s}(\tilde{s}(\tilde{z})),\tilde{u})\geql \gu, \tilde{u}\in \tilde{\mbb{U}}$}
& \qquad [356-360] \\
\text{\bf Rule (\cref{ceq4hr5})} : [205]; \tau/\tilde{s}(\tilde{s}(\tilde{z})), x/\tilde{u} : [356-360] : \tilde{u}\in \tilde{\mbb{U}} :
& \\
\framebox{$\tilde{p}_{0,6}(\tilde{s}(\tilde{s}(\tilde{z})),\tilde{u})\geql \gu, \tilde{u}\in \tilde{\mbb{U}}$}
& [361-365] \\ 
\text{\bf Rule (\cref{ceq4hr14})} : \tilde{p}_{0,6}(\tilde{s}(\tilde{s}(\tilde{z})),x)\geql \gu : [361-365] :
& \\
\framebox{$\mi{uni}(x)\geql \gz$}\vee \tilde{p}_{0,6}(\tilde{s}(\tilde{s}(\tilde{z})),x)\geql \gu
& [366] \\
\text{\bf Rule (\cref{ceq4hr8})} : \forall x\, \tilde{p}_{0,6}(\tilde{s}(\tilde{s}(\tilde{z})),x), \gu :
& \\
\tilde{p}_{0,6}(\tilde{s}(\tilde{s}(\tilde{z})),\tilde{w}_{(0,0)})\gle \gu\vee 
\gu\geql \forall x\, \tilde{p}_{0,6}(\tilde{s}(\tilde{s}(\tilde{z})),x)\vee 
\framebox{$\gu\gle \forall x\, \tilde{p}_{0,6}(\tilde{s}(\tilde{s}(\tilde{z})),x)$}
& [367] \\
\text{\bf Rule (\cref{ceq4hr5})} : [367] :
& \\
\framebox{$\tilde{p}_{0,6}(\tilde{s}(\tilde{s}(\tilde{z})),\tilde{w}_{(0,0)})\gle \gu$}\vee 
\gu\geql \forall x\, \tilde{p}_{0,6}(\tilde{s}(\tilde{s}(\tilde{z})),x)
& [368] \\
\text{\bf Rule (\cref{ceq4hr5})} : [205] [206] [366] :
& \\
\framebox{$\tilde{p}_{0,6}(\tau,x)\geql \gu\vee \tilde{p}_{0,6}(\tilde{s}(\tilde{s}(\tilde{z})),x)\geql \gu$}
& [369] \\
\text{\bf Rule (\cref{ceq4hr5})} : [368] : [369]; \tau/\tilde{s}(\tilde{s}(\tilde{z})), x/\tilde{w}_{(0,0)} :
& \\
\framebox{$\forall x\, \tilde{p}_{0,6}(\tilde{s}(\tilde{s}(\tilde{z})),x)\geql \gu$}
& [370] \\
\text{repeatedly \bf Rule (\cref{ceq4hr5})} : [216] [217] [218]; \tau/\tilde{s}(\tilde{s}(\tilde{z})), x/\tilde{u} : [351-355] :
                                              \tilde{u}\in \tilde{\mbb{U}} :
& \\
\framebox{$\tilde{p}_{0,11}(\tilde{s}(\tilde{s}(\tilde{z})),\tilde{u})\geql \gu, \tilde{u}\in \tilde{\mbb{U}}$}
& [371-375] \\
\text{\bf Rule (\cref{ceq4hr5})} : [213]; \tau/\tilde{s}(\tilde{s}(\tilde{z})), x/\tilde{u} : [371-375] : \tilde{u}\in \tilde{\mbb{U}} :
& \\
\framebox{$\tilde{p}_{0,7}(\tilde{s}(\tilde{s}(\tilde{z})),\tilde{u})\geql \gu, \tilde{u}\in \tilde{\mbb{U}}$}
& [376-380] \\ 
\text{\bf Rule (\cref{ceq4hr14})} : \tilde{p}_{0,7}(\tilde{s}(\tilde{s}(\tilde{z})),x)\geql \gu : [376-380] :
& \\
\framebox{$\mi{uni}(x)\geql \gz$}\vee \tilde{p}_{0,7}(\tilde{s}(\tilde{s}(\tilde{z})),x)\geql \gu
& [381] \\
\text{\bf Rule (\cref{ceq4hr8})} : \forall x\, \tilde{p}_{0,7}(\tilde{s}(\tilde{s}(\tilde{z})),x), \gu :
& \\
\tilde{p}_{0,7}(\tilde{s}(\tilde{s}(\tilde{z})),\tilde{w}_{(1,1)})\gle \gu\vee
\gu\geql \forall x\, \tilde{p}_{0,7}(\tilde{s}(\tilde{s}(\tilde{z})),x)\vee 
\framebox{$\gu\gle \forall x\, \tilde{p}_{0,7}(\tilde{s}(\tilde{s}(\tilde{z})),x)$}
& [382] \\
\text{\bf Rule (\cref{ceq4hr5})} : [382] :
& \\
\framebox{$\tilde{p}_{0,7}(\tilde{s}(\tilde{s}(\tilde{z})),\tilde{w}_{(1,1)})\gle \gu$}\vee
\gu\geql \forall x\, \tilde{p}_{0,7}(\tilde{s}(\tilde{s}(\tilde{z})),x)
& [383] \\
\text{\bf Rule (\cref{ceq4hr5})} : [213] [214] [381] :
& \\
\framebox{$\tilde{p}_{0,7}(\tau,x)\geql \gu\vee \tilde{p}_{0,7}(\tilde{s}(\tilde{s}(\tilde{z})),x)\geql \gu$}
& [384] \\
\text{\bf Rule (\cref{ceq4hr5})} : [383] : [384]; \tau/\tilde{s}(\tilde{s}(\tilde{z})), x/\tilde{w}_{(1,1)} :
& \\
\framebox{$\forall x\, \tilde{p}_{0,7}(\tilde{s}(\tilde{s}(\tilde{z})),x)\geql \gu$}
& [385] \\
\text{\bf Rule (\cref{ceq4hr5})} : [202] [211]; \tau/\tilde{s}(\tilde{s}(\tilde{z})) : [385] :
& \\
\framebox{$\tilde{p}_{0,3}(\tilde{s}(\tilde{s}(\tilde{z})),x)\geql \tilde{p}_{0,4}(\tilde{s}(\tilde{s}(\tilde{z})),x)$}
& [386] \\
\text{\bf Rule (\cref{ceq4hr5})} : [199] [203]; \tau/\tilde{s}(\tilde{s}(\tilde{z})) : [370] [386] :
& \\
\framebox{$\tilde{p}_{0,1}(\tilde{s}(\tilde{s}(\tilde{z})),x)\geql \tilde{p}_{0,2}(\tilde{s}(\tilde{s}(\tilde{z})),x)$}
& [387] \\
\text{\bf Rule (\cref{ceq4hr7})} : \exists \tau\, \tilde{p}_{0,1}(\tau,x) :
& \\
\framebox{$\tilde{p}_{0,1}(\tau,x)\gle \exists \tau\, \tilde{p}_{0,1}(\tau,x)\vee
           \tilde{p}_{0,1}(\tau,x)\geql \exists \tau\, \tilde{p}_{0,1}(\tau,x)$}
& [388] \\  
\text{repeatedly \bf Rule (\cref{ceq4hr5})} : [196] [197] [235] [387] : [200] [388]; \tau/\tilde{s}(\tilde{s}(\tilde{z})) :  
& \\
\square
& [389] \\[1mm]
\hline \hline
\end{IEEEeqnarray*}
\end{minipage}
\vspace{-2mm}
\end{table*}

\begin{table*}[p]
\caption{Binary interpolation rules for $\wedge$, $\vee$, $\rightarrow$, $\leftrightarrow$, $\geql$, $\gle$}\label{tab2}
\vspace{-6mm}
\centering
\begin{minipage}[t]{\linewidth-30mm}
\footnotesize
\begin{IEEEeqnarray}{*LL}
\hline \hline \notag \\[0mm]
\notag 
\text{\bf Case} & \\[1mm]
\hline \notag \\[2mm]
\label{eq0rr1+}
\mbi{\theta=\theta_1\wedge \theta_2} & 
\dfrac{\tilde{p}_\mbbm{i}(\bar{x})\leftrightarrow \theta_1\wedge \theta_2}
      {\left\{\begin{array}{l}
              \tilde{p}_{\mbbm{i}_1}(\bar{x})\gle \tilde{p}_{\mbbm{i}_2}(\bar{x})\vee \tilde{p}_{\mbbm{i}_1}(\bar{x})\geql \tilde{p}_{\mbbm{i}_2}(\bar{x})\vee \tilde{p}_\mbbm{i}(\bar{x})\geql \tilde{p}_{\mbbm{i}_2}(\bar{x}), \\
              \tilde{p}_{\mbbm{i}_2}(\bar{x})\gle \tilde{p}_{\mbbm{i}_1}(\bar{x})\vee \tilde{p}_\mbbm{i}(\bar{x})\geql \tilde{p}_{\mbbm{i}_1}(\bar{x}),               
              \tilde{p}_{\mbbm{i}_1}(\bar{x})\leftrightarrow \theta_1, \tilde{p}_{\mbbm{i}_2}(\bar{x})\leftrightarrow \theta_2
              \end{array}\right\}} \\[2mm]
\IEEEeqnarraymulticol{2}{l}{
|\text{Consequent}|=
15+10\cdot |\bar{x}|+|\tilde{p}_{\mbbm{i}_1}(\bar{x})\leftrightarrow \theta_1|+|\tilde{p}_{\mbbm{i}_2}(\bar{x})\leftrightarrow \theta_2|\leq
27\cdot (1+|\bar{x}|)+|\tilde{p}_{\mbbm{i}_1}(\bar{x})\leftrightarrow \theta_1|+|\tilde{p}_{\mbbm{i}_2}(\bar{x})\leftrightarrow \theta_2|} \notag \\[6mm]
\label{eq0rr2+}
\mbi{\theta=\theta_1\vee \theta_2} & 
\dfrac{\tilde{p}_\mbbm{i}(\bar{x})\leftrightarrow (\theta_1\vee \theta_2)}
      {\left\{\begin{array}{l}
              \tilde{p}_{\mbbm{i}_1}(\bar{x})\gle \tilde{p}_{\mbbm{i}_2}(\bar{x})\vee \tilde{p}_{\mbbm{i}_1}(\bar{x})\geql \tilde{p}_{\mbbm{i}_2}(\bar{x})\vee \tilde{p}_\mbbm{i}(\bar{x})\geql \tilde{p}_{\mbbm{i}_1}(\bar{x}), \\
              \tilde{p}_{\mbbm{i}_2}(\bar{x})\gle \tilde{p}_{\mbbm{i}_1}(\bar{x})\vee \tilde{p}_\mbbm{i}(\bar{x})\geql \tilde{p}_{\mbbm{i}_2}(\bar{x}),               
              \tilde{p}_{\mbbm{i}_1}(\bar{x})\leftrightarrow \theta_1, \tilde{p}_{\mbbm{i}_2}(\bar{x})\leftrightarrow \theta_2
              \end{array}\right\}} \\[2mm]
\IEEEeqnarraymulticol{2}{l}{
|\text{Consequent}|=
15+10\cdot |\bar{x}|+|\tilde{p}_{\mbbm{i}_1}(\bar{x})\leftrightarrow \theta_1|+|\tilde{p}_{\mbbm{i}_2}(\bar{x})\leftrightarrow \theta_2|\leq
27\cdot (1+|\bar{x}|)+|\tilde{p}_{\mbbm{i}_1}(\bar{x})\leftrightarrow \theta_1|+|\tilde{p}_{\mbbm{i}_2}(\bar{x})\leftrightarrow \theta_2|} \notag \\[6mm]
\label{eq0rr3+}
\begin{array}{l}
\mbi{\theta=\theta_1\rightarrow \theta_2,} \\ 
\mbi{\theta_2\neq \gz} 
\end{array} & 
\dfrac{\tilde{p}_\mbbm{i}(\bar{x})\leftrightarrow (\theta_1\rightarrow \theta_2)}
      {\left\{\begin{array}{l}
              \tilde{p}_{\mbbm{i}_1}(\bar{x})\gle \tilde{p}_{\mbbm{i}_2}(\bar{x})\vee \tilde{p}_{\mbbm{i}_1}(\bar{x})\geql \tilde{p}_{\mbbm{i}_2}(\bar{x})\vee \tilde{p}_\mbbm{i}(\bar{x})\geql \tilde{p}_{\mbbm{i}_2}(\bar{x}), \\
              \tilde{p}_{\mbbm{i}_2}(\bar{x})\gle \tilde{p}_{\mbbm{i}_1}(\bar{x})\vee \tilde{p}_\mbbm{i}(\bar{x})\geql \gu,               
              \tilde{p}_{\mbbm{i}_1}(\bar{x})\leftrightarrow \theta_1, \tilde{p}_{\mbbm{i}_2}(\bar{x})\leftrightarrow \theta_2
              \end{array}\right\}} \\[2mm]
\IEEEeqnarraymulticol{2}{l}{
|\text{Consequent}|=
15+9\cdot |\bar{x}|+|\tilde{p}_{\mbbm{i}_1}(\bar{x})\leftrightarrow \theta_1|+|\tilde{p}_{\mbbm{i}_2}(\bar{x})\leftrightarrow \theta_2|\leq
27\cdot (1+|\bar{x}|)+|\tilde{p}_{\mbbm{i}_1}(\bar{x})\leftrightarrow \theta_1|+|\tilde{p}_{\mbbm{i}_2}(\bar{x})\leftrightarrow \theta_2|} \notag \\[6mm]
\label{eq0rr33+}
\mbi{\theta=\theta_1\leftrightarrow \theta_2} & 
\dfrac{\tilde{p}_\mbbm{i}(\bar{x})\leftrightarrow (\theta_1\leftrightarrow \theta_2)}
      {\left\{\begin{array}{l}
              \tilde{p}_{\mbbm{i}_1}(\bar{x})\gle \tilde{p}_{\mbbm{i}_2}(\bar{x})\vee \tilde{p}_{\mbbm{i}_1}(\bar{x})\geql \tilde{p}_{\mbbm{i}_2}(\bar{x})\vee \tilde{p}_\mbbm{i}(\bar{x})\geql \tilde{p}_{\mbbm{i}_2}(\bar{x}), \\
              \tilde{p}_{\mbbm{i}_2}(\bar{x})\gle \tilde{p}_{\mbbm{i}_1}(\bar{x})\vee \tilde{p}_{\mbbm{i}_2}(\bar{x})\geql \tilde{p}_{\mbbm{i}_1}(\bar{x})\vee \tilde{p}_\mbbm{i}(\bar{x})\geql \tilde{p}_{\mbbm{i}_1}(\bar{x}), \\
              \tilde{p}_{\mbbm{i}_1}(\bar{x})\gle \tilde{p}_{\mbbm{i}_2}(\bar{x})\vee \tilde{p}_{\mbbm{i}_2}(\bar{x})\gle \tilde{p}_{\mbbm{i}_1}(\bar{x})\vee \tilde{p}_\mbbm{i}(\bar{x})\geql \gu, 
              \tilde{p}_{\mbbm{i}_1}(\bar{x})\leftrightarrow \theta_1, \tilde{p}_{\mbbm{i}_2}(\bar{x})\leftrightarrow \theta_2
              \end{array}\right\}} \\[2mm]
\IEEEeqnarraymulticol{2}{l}{
|\text{Consequent}|=
27+17\cdot |\bar{x}|+|\tilde{p}_{\mbbm{i}_1}(\bar{x})\leftrightarrow \theta_1|+|\tilde{p}_{\mbbm{i}_2}(\bar{x})\leftrightarrow \theta_2|\leq
27\cdot (1+|\bar{x}|)+|\tilde{p}_{\mbbm{i}_1}(\bar{x})\leftrightarrow \theta_1|+|\tilde{p}_{\mbbm{i}_2}(\bar{x})\leftrightarrow \theta_2|} \notag \\[6mm]
\label{eq0rr7+}
\begin{array}{l}
\mbi{\theta=\theta_1\geql \theta_2,} \\
\mbi{\theta_i\neq \gz, \gu} 
\end{array} & 
\dfrac{\tilde{p}_\mbbm{i}(\bar{x})\leftrightarrow (\theta_1\geql \theta_2)}
      {\left\{\begin{array}{l}
              \tilde{p}_{\mbbm{i}_1}(\bar{x})\geql \tilde{p}_{\mbbm{i}_2}(\bar{x})\vee \tilde{p}_\mbbm{i}(\bar{x})\geql \gz, \\
              \tilde{p}_{\mbbm{i}_1}(\bar{x})\gle \tilde{p}_{\mbbm{i}_2}(\bar{x})\vee \tilde{p}_{\mbbm{i}_2}(\bar{x})\gle \tilde{p}_{\mbbm{i}_1}(\bar{x})\vee \tilde{p}_\mbbm{i}(\bar{x})\geql \gu, 
              \tilde{p}_{\mbbm{i}_1}(\bar{x})\leftrightarrow \theta_1, \tilde{p}_{\mbbm{i}_2}(\bar{x})\leftrightarrow \theta_2
              \end{array}\right\}} \\[2mm]
\IEEEeqnarraymulticol{2}{l}{
|\text{Consequent}|=
15+8\cdot |\bar{x}|+|\tilde{p}_{\mbbm{i}_1}(\bar{x})\leftrightarrow \theta_1|+|\tilde{p}_{\mbbm{i}_2}(\bar{x})\leftrightarrow \theta_2|\leq
27\cdot (1+|\bar{x}|)+|\tilde{p}_{\mbbm{i}_1}(\bar{x})\leftrightarrow \theta_1|+|\tilde{p}_{\mbbm{i}_2}(\bar{x})\leftrightarrow \theta_2|} \notag \\[6mm]
\label{eq0rr8+}
\begin{array}{l}
\mbi{\theta=\theta_1\gle \theta_2,} \\ 
\mbi{\theta_1\neq \gz, \theta_2\neq \gu} 
\end{array} \qquad & 
\dfrac{\tilde{p}_\mbbm{i}(\bar{x})\leftrightarrow (\theta_1\gle \theta_2)}
      {\left\{\begin{array}{l}
              \tilde{p}_{\mbbm{i}_1}(\bar{x})\gle \tilde{p}_{\mbbm{i}_2}(\bar{x})\vee \tilde{p}_\mbbm{i}(\bar{x})\geql \gz, \\
              \tilde{p}_{\mbbm{i}_2}(\bar{x})\gle \tilde{p}_{\mbbm{i}_1}(\bar{x})\vee \tilde{p}_{\mbbm{i}_2}(\bar{x})\geql \tilde{p}_{\mbbm{i}_1}(\bar{x})\vee \tilde{p}_\mbbm{i}(\bar{x})\geql \gu, 
              \tilde{p}_{\mbbm{i}_1}(\bar{x})\leftrightarrow \theta_1, \tilde{p}_{\mbbm{i}_2}(\bar{x})\leftrightarrow \theta_2
              \end{array}\right\}} \\[2mm]
\IEEEeqnarraymulticol{2}{l}{
|\text{Consequent}|=
15+8\cdot |\bar{x}|+|\tilde{p}_{\mbbm{i}_1}(\bar{x})\leftrightarrow \theta_1|+|\tilde{p}_{\mbbm{i}_2}(\bar{x})\leftrightarrow \theta_2|\leq
27\cdot (1+|\bar{x}|)+|\tilde{p}_{\mbbm{i}_1}(\bar{x})\leftrightarrow \theta_1|+|\tilde{p}_{\mbbm{i}_2}(\bar{x})\leftrightarrow \theta_2|} \notag \\[2mm]
\hline \hline \notag
\end{IEEEeqnarray}
\end{minipage}
\vspace{-2mm}
\end{table*}
\begin{table*}[p]
\caption{Unary interpolation rules for $\rightarrow$, $\geql$, $\gle$, $\forall$, $\exists$}\label{tab3}
\vspace{-6mm}
\centering
\begin{minipage}[t]{\linewidth-65mm}
\footnotesize
\begin{IEEEeqnarray}{*LL}
\hline \hline \notag \\[0mm]
\notag 
\text{\bf Case} & \\[1mm]
\hline \notag \\[2mm]
\label{eq0rr4+}
\mbi{\theta=\theta_1\rightarrow \gz} \qquad & 
\dfrac{\tilde{p}_\mbbm{i}(\bar{x})\leftrightarrow (\theta_1\rightarrow \gz)}
      {\{\tilde{p}_{\mbbm{i}_1}(\bar{x})\geql \gz\vee \tilde{p}_\mbbm{i}(\bar{x})\geql \gz,
         \gz\gle \tilde{p}_{\mbbm{i}_1}(\bar{x})\vee \tilde{p}_\mbbm{i}(\bar{x})\geql \gu, 
         \tilde{p}_{\mbbm{i}_1}(\bar{x})\leftrightarrow \theta_1\}} \\[2mm]
\IEEEeqnarraymulticol{2}{l}{
|\text{Consequent}|=
12+4\cdot |\bar{x}|+|\tilde{p}_{\mbbm{i}_1}(\bar{x})\leftrightarrow \theta_1|\leq
27\cdot (1+|\bar{x}|)+|\tilde{p}_{\mbbm{i}_1}(\bar{x})\leftrightarrow \theta_1|} \notag \\[6mm]
\label{eq0rr77+}
\mbi{\theta=\theta_1\geql \gz} & 
\dfrac{\tilde{p}_\mbbm{i}(\bar{x})\leftrightarrow (\theta_1\geql \gz)}
      {\{\tilde{p}_{\mbbm{i}_1}(\bar{x})\geql \gz\vee \tilde{p}_\mbbm{i}(\bar{x})\geql \gz, 
         \gz\gle \tilde{p}_{\mbbm{i}_1}(\bar{x})\vee \tilde{p}_\mbbm{i}(\bar{x})\geql \gu, 
         \tilde{p}_{\mbbm{i}_1}(\bar{x})\leftrightarrow \theta_1\}} \\[2mm]
\IEEEeqnarraymulticol{2}{l}{
|\text{Consequent}|=
12+4\cdot |\bar{x}|+|\tilde{p}_{\mbbm{i}_1}(\bar{x})\leftrightarrow \theta_1|\leq
27\cdot (1+|\bar{x}|)+|\tilde{p}_{\mbbm{i}_1}(\bar{x})\leftrightarrow \theta_1|} \notag \\[6mm]
\label{eq0rr777+}
\mbi{\theta=\theta_1\geql \gu} & 
\dfrac{\tilde{p}_\mbbm{i}(\bar{x})\leftrightarrow (\theta_1\geql \gu)}
      {\{\tilde{p}_{\mbbm{i}_1}(\bar{x})\geql \gu\vee \tilde{p}_\mbbm{i}(\bar{x})\geql \gz, 
         \tilde{p}_{\mbbm{i}_1}(\bar{x})\gle \gu\vee \tilde{p}_\mbbm{i}(\bar{x})\geql \gu, 
         \tilde{p}_{\mbbm{i}_1}(\bar{x})\leftrightarrow \theta_1\}} \\[2mm]
\IEEEeqnarraymulticol{2}{l}{
|\text{Consequent}|=
12+4\cdot |\bar{x}|+|\tilde{p}_{\mbbm{i}_1}(\bar{x})\leftrightarrow \theta_1|\leq
27\cdot (1+|\bar{x}|)+|\tilde{p}_{\mbbm{i}_1}(\bar{x})\leftrightarrow \theta_1|} \notag \\[6mm]
\label{eq0rr88+}
\mbi{\theta=\gz\gle \theta_1} & 
\dfrac{\tilde{p}_\mbbm{i}(\bar{x})\leftrightarrow (\gz\gle \theta_1)}
      {\{\gz\gle \tilde{p}_{\mbbm{i}_1}(\bar{x})\vee \tilde{p}_\mbbm{i}(\bar{x})\geql \gz, 
         \tilde{p}_{\mbbm{i}_1}(\bar{x})\geql \gz\vee \tilde{p}_\mbbm{i}(\bar{x})\geql \gu,                              
         \tilde{p}_{\mbbm{i}_1}(\bar{x})\leftrightarrow \theta_1\}} \\[2mm]
\IEEEeqnarraymulticol{2}{l}{
|\text{Consequent}|=
12+4\cdot |\bar{x}|+|\tilde{p}_{\mbbm{i}_1}(\bar{x})\leftrightarrow \theta_1|\leq
27\cdot (1+|\bar{x}|)+|\tilde{p}_{\mbbm{i}_1}(\bar{x})\leftrightarrow \theta_1|} \notag \\[6mm] 
\label{eq0rr888+}
\mbi{\theta=\theta_1\gle \gu} & 
\dfrac{\tilde{p}_\mbbm{i}(\bar{x})\leftrightarrow (\theta_1\gle \gu)}
      {\{\tilde{p}_{\mbbm{i}_1}(\bar{x})\gle \gu\vee \tilde{p}_\mbbm{i}(\bar{x})\geql \gz, 
         \tilde{p}_{\mbbm{i}_1}(\bar{x})\geql \gu\vee \tilde{p}_\mbbm{i}(\bar{x})\geql \gu, 
         \tilde{p}_{\mbbm{i}_1}(\bar{x})\leftrightarrow \theta_1\}} \\[2mm]
\IEEEeqnarraymulticol{2}{l}{
|\text{Consequent}|=
12+4\cdot |\bar{x}|+|\tilde{p}_{\mbbm{i}_1}(\bar{x})\leftrightarrow \theta_1|\leq
27\cdot (1+|\bar{x}|)+|\tilde{p}_{\mbbm{i}_1}(\bar{x})\leftrightarrow \theta_1|} \notag \\[6mm]
\label{eq0rr5+}
\mbi{\theta=\forall x\, \theta_1} & 
\dfrac{\tilde{p}_\mbbm{i}(\bar{x})\leftrightarrow \forall x\, \theta_1}
      {\{\tilde{p}_\mbbm{i}(\bar{x})\geql \forall x\, \tilde{p}_{\mbbm{i}_1}(\bar{x}),
         \tilde{p}_{\mbbm{i}_1}(\bar{x})\leftrightarrow \theta_1\}} \\[2mm]
\IEEEeqnarraymulticol{2}{l}{
|\text{Consequent}|=
5+2\cdot |\bar{x}|+|\tilde{p}_{\mbbm{i}_1}(\bar{x})\leftrightarrow \theta_1|\leq
27\cdot (1+|\bar{x}|)+|\tilde{p}_{\mbbm{i}_1}(\bar{x})\leftrightarrow \theta_1|} \notag \\[6mm]
\label{eq0rr6+}
\mbi{\theta=\exists x\, \theta_1} & 
\dfrac{\tilde{p}_\mbbm{i}(\bar{x})\leftrightarrow \exists x\, \theta_1}
      {\{\tilde{p}_\mbbm{i}(\bar{x})\geql \exists x\, \tilde{p}_{\mbbm{i}_1}(\bar{x}),
         \tilde{p}_{\mbbm{i}_1}(\bar{x})\leftrightarrow \theta_1\}} \\[2mm]
\IEEEeqnarraymulticol{2}{l}{
|\text{Consequent}|=
5+2\cdot |\bar{x}|+|\tilde{p}_{\mbbm{i}_1}(\bar{x})\leftrightarrow \theta_1|\leq
27\cdot (1+|\bar{x}|)+|\tilde{p}_{\mbbm{i}_1}(\bar{x})\leftrightarrow \theta_1|} \notag \\[2mm]
\hline \hline \notag
\end{IEEEeqnarray}
\end{minipage}
\vspace{-2mm}
\end{table*}

\end{document}